\documentclass[11pt]{article}
\usepackage{amsfonts}
\usepackage{amsmath}
\usepackage{amssymb}
\usepackage{mathtools}
\usepackage{graphicx}
\usepackage{geometry}
\usepackage{listings}
\usepackage{mathrsfs}
\usepackage{natbib}
\usepackage{bm}
\usepackage{subcaption}
\usepackage{multirow}
\usepackage{hyperref}
\usepackage{booktabs, caption, makecell}

\usepackage{threeparttable}
\usepackage{enumitem}
\usepackage{color, colortbl}
\usepackage{float}
\usepackage{hhline}

\usepackage[title]{appendix}
\definecolor{Gray}{gray}{0.5}

\usepackage{algorithm,algcompatible}
\algnewcommand\INPUT{\item[\textbf{Input:}]}%
\algnewcommand\OUTPUT{\item[\textbf{Output:}]}%

\providecommand{\U}[1]{\protect\rule{.1in}{.1in}}

\newtheorem{theorem}{Theorem}
\newtheorem{corollary}{Corollary}
\newtheorem{lemma}{Lemma}

\newtheorem{example}{Example}
\newtheorem{remark}{Remark}
\newenvironment{proof}[1][Proof]{\noindent\textbf{#1.} }{\ \rule{0.5em}{0.5em}}
\numberwithin{equation}{section}
\begin{document}

\title{Robust estimation and model diagnostic of insurance loss data: a weighted likelihood approach}

\author{Tsz Chai Fung\thanks{Address: Department of Risk Management and Insurance, Georgia State University, Atlanta, GA 30303. Email: \href{mailto:tfung@gsu.edu}{tfung@gsu.edu}.}}
\maketitle

\abstract{This paper presents a score-based weighted likelihood estimator (SWLE) for robust estimations of generalized linear model (GLM) for insurance loss data. The SWLE exhibits a limited sensitivity to the outliers, theoretically justifying its robustness against model contaminations. Also, with the specially designed weight function to effectively diminish the contributions of extreme losses to the GLM parameter estimations, most statistical quantities can still be derived analytically, minimizing the computational burden for parameter calibrations. Apart from robust estimations, the SWLE can also act as a quantitative diagnostic tool to detect outliers and systematic model misspecifications. Motivated by the coverage modifications which make insurance losses often random censored and truncated, the SWLE is extended to accommodate censored and truncated data. We exemplify the SWLE on three simulation studies and two real insurance datasets. Empirical results suggest that the SWLE produces more reliable parameter estimates than the MLE if outliers contaminate the dataset. The SWLE diagnostic tool also successfully detects any systematic model misspecifications with high power, accompanying some potential model improvements.}

\medskip


\section{Introduction}
Insurance loss modelling and diagnostics have long been challenging actuarial problems essential for general insurance ratemaking and reserving. Insurance losses often exhibit peculiar distributional characteristics, including multimodality, tail-heaviness, and outlier contamination. These losses are also influenced by the policyholder attributes, potentially in a complicated way. 
The levels of peculiarity and complexity also vary significantly across different insurance datasets. For example, a simple log-normal model already fits well the Secura Re loss data analyzed by \cite{BLOSTEIN201935}. In contrast, the Greek automobile insurance data (\cite{fung2021mixture}) requires a model that captures distributional multimodality with heterogeneous regression links across different parts of the distribution.
Furthermore, insurance claim losses are also subject to policy coverage modifications, including deductibles (claims below a certain threshold are not reported) and policy limits (claims are reported only to a certain capped amount). As a result, losses are often random censored and truncated.

Amongst all statistical models, the generalized linear model (GLM, \cite{nelder1972generalized}) is widely regarded as a benchmark for insurance loss modelling, due to its analytical tractability and interpretability with a minimal computational burden. However, GLM is heavily criticized for producing maximum likelihood estimators (MLE) highly sensitive to outlier contaminations. The GLM is also inflexible to capture peculiar model characteristics (e.g., non-linear regression, tail-heaviness, and distributional multimodality), which may or may not appear in an insurance loss dataset. It is, therefore, essential to explore how to estimate the GLM parameters robustly and understand whether or not the GLM is fundamentally appropriate in modelling an insurance dataset of interest. 

To address the robustness issues underlying the MLE approach, several research works propose alternative estimation approaches to reduce the sensitivities to outlier contaminations. In the actuarial loss modelling literature, \cite{brazauskas2000robust}, \cite{serfling2002efficient} and \cite{brazauskas2003favorable} consider quantile and median estimators, while \cite{zhao2018robust}, \cite{poudyal2021robust} and \cite{poudyal2021truncated} propose truncated or winsorized method of moments (MoM) to robustly estimate the insurance loss distributions. However, extending their use to models involving many parameters is rather challenging. This issue is recently addressed by \cite{fung2021maximum}, who explores the maximum weighted likelihood estimator (MWLE), which down-weights the observations that may harm the model robustness for robust tail estimation of the finite mixture model. 
In the context of regression models, \cite{cantoni2001robust} first considers a robustification of the quasi-likelihood function for robust estimations of the GLM. Various alternative forms of robust M-estimators for the GLM are then extensively studied in the statistics literature, including \cite{valdora2014robust}, \cite{aeberhard2014robust} and \cite{ghosh2016robust}. \cite{wong2014robust} and \cite{aeberhard2021robust} further extends the robustification techniques to the generalized additive models (GAM). Given that the methods above are designed for a complete dataset, it remains an open problem to develop a robust model estimation approach for random censored and truncated insurance loss data.

With regards to the limited flexibility of GLM, there are more sophisticated insurance loss models extensive studied in actuarial literature, such as extreme value distributions (\cite{embrechts1999extreme}), composite models (\cite{cooray2005modeling}), finite mixture models (\cite{MILJKOVIC2016387}), GLM with varying dispersion (\cite{tzougas2020algorithm}), and neural network (\cite{delong2021gamma}).
These models capture more complex data characteristics at some costs of model interpretability, computational burden, and mathematical tractability.
A fundamental question is: When is the GLM sufficient such that the more complex models mentioned above are unnecessary?

We need model diagnostic tools to assess the goodness-of-fit in light of the above question. Actuaries often rely on ad-hoc qualitative tools such as P-P or Q-Q plots and quantitative tools such as Anderson-Darling (AD), Kolmogorov-Smirnov (KS), or chi-squared tests to determine the goodness-of-fit of the probability distributions, or deviance residual analysis plot to visually detect any abnormal regression patterns. These basic diagnostic tools are designed only for complete data and fail to detect misspecifications on both probability distribution and regression link simultaneously. Hence, it is desirable to develop a more universal statistical diagnostic tool, which provides quantitative and comprehensive assessments on the appropriateness of the GLM to an insurance loss data subject to coverage modifications. This gives us a relatively more objective rule in deciding whether or not we should reject the GLM and consider the more complex model classes as mentioned above. 

Motivated by the issues above, in this paper, we make the following contributions:

Firstly, we introduce a score-based weighted likelihood estimation (SWLE), adapted and extended from the MWLE proposed by \cite{fung2021maximum} for robust estimation of regression models. Using a weight function, the proposed SWLE diminishes the contributions of extreme observations to the GLM parameter estimations. In this way, the estimated GLM parameters are less sensitive to outliers and, hence, more reliable and robust against model contamination. We show that the proposed SWLE produces consistent and asymptotic normal parameters estimations. Also, with careful selections of weight functions, most statistical quantities (e.g., score functions, information matrix) under SWLE are analytically tractable. Therefore, the computational burden for calibrating the GLM parameters by the SWLE is almost no different from the MLE.

Secondly, we develop a novel Wald-based test statistic based on the sensitivity of the SWLE weight functions to the estimated GLM parameters to quantitatively assess the overall appropriateness of using the GLM to model the dataset. Higher sensitivities suggest rejection of the GLM. If our proposed Wald-type statistic rejects the GLM, an analysis on where the sensitivities come from also provides some guidance to improve the GLM.

Thirdly, we extend the SWLE to cater for random censoring and truncation of data. The corresponding statistical inference quantities, including score function, information matrix, and Wald-type diagnostic test statistic, are derived.

The remainder of this article proceeds as follows. The class of generalized linear model (GLM) is first revisited in Section \ref{sec:glm} with relevant notations. Then, Section \ref{sec:likelihood} reviews some likelihood-based inference techniques, which motivate us to introduce a novel score-based weighted likelihood estimation (SWLE) approach for robust estimations of insurance loss models. In Section \ref{sec:swle}, we formally construct the SWLE framework for the GLM, supplemented by estimation algorithms and several theoretical properties to justify the computational tractability, consistency, and robustness of the proposed SWLE. Section \ref{sec:diag} introduces an alternative use of the SWLE as a model diagnostic tool to quantitatively detect model misspecifications. The SWLE is further extended in Section \ref{sec:censtrun} to cater for censored and truncated data. The performance and practical applicability of the proposed SWLE are analyzed through three simulation studies in Section \ref{sec:sim} and two real insurance datasets in Section \ref{sec:data}. Section \ref{sec:conclude} concludes.

\section{Generalized Linear Model (GLM)} \label{sec:glm}
In this section, we briefly review the class of Generalized Linear Model (GLM) and define the relevant notations used throughout the paper. Suppose that there are $n$ independent (transformed) loss severities $\bm{Y}=(Y_1,\ldots,Y_n)$ with realizations $\bm{y}=(y_1,\ldots,y_n)$. Each loss $y_i\in\mathcal{Y}$ is accompanied by $P$ covariates (policyholder attributes) denoted as $\bm{x}_i=(x_{i1},\ldots,x_{iP})^T\in\mathcal{X}$ for $i=1,\ldots,n$. Also, define $\bm{X}=(\bm{x}_1,\ldots,\bm{x}_n)^T$ as an $n\times P$-design matrix consisting of the attributes of all policyholders. We model $Y_i|\bm{x}_i$ through GLM with density function given by
\begin{align} \label{eq:glm}
f(y_i;\bm{x}_i,\bm{\Psi})=\exp\left\{\frac{\theta_iy_i-A(\theta_i)}{\phi}+C(y_i,\phi)\right\},\qquad y_i\in\mathcal{Y}, \bm{x}_i\in\mathcal{X}
\end{align}
for $i=1,\ldots,n$. Here, $\theta_i:=\theta(\bm{x}_i,\bm{\beta})$ is a canonical parameter which depends on the covariates $\bm{x}_i$ and regression coefficients $\bm{\beta}:=(\beta_1,\ldots,\beta_P)^T$. The scale parameter is $\phi$ and the set of all model parameters is $\bm{\Psi}=(\bm{\beta},\phi)$. The mean and variance of $Y_i|\bm{x}_i$ are $E[Y_i|\bm{x}_i,\bm{\Psi}]:=\mu_i=A'(\theta_i)$ and $\text{Var}[Y_i|\bm{x}_i,\bm{\Psi}]:=\sigma_i^2=\phi A''(\theta_i)$ respectively. For linear regression, it is assumed that $\bm{x}_i^T\bm{\beta}=\eta(\mu_i)=\eta(A'(\theta_i))$, where $\eta(\cdot)$ is called the link function. We also define $\xi(\cdot)=(\eta\circ A')^{-1}(\cdot)$ such that $\theta_i=\xi(\bm{x}_i^T\bm{\beta})$. In the special case where $\xi(\cdot)$ or $\eta(A'(\cdot))$ is an identity function such that $\theta_i=\bm{x}_i^T\bm{\beta}$, we call the resulting $\eta(\cdot)$ a canonical link. Throughout this paper, we further particularize the class of GLM to be considered, by assuming that the function $C(y_i,\phi)$ in Equation (\ref{eq:glm}) can be decomposed as

\begin{align} \label{eq:glm_C}
C(y_i,\phi)=\left(\frac{1}{\phi}-c\right)g(y_i)+a(y_i)+b(\phi)
\end{align}
for a constant $c$ and some functions $g(\cdot)$, $a(\cdot)$ and $b(\cdot)$. This assumption is not restrictive in actuarial practice as most widely adopted GLMs (e.g., Gamma and inverse Gaussian) satisfy this assumption.

\section{Likelihood-based inference techniques} \label{sec:likelihood}
This section briefly reviews several likelihood-based inference tools and proposes a score-based weighted likelihood estimation (SWLE) approach for robust parameter estimations of regression models.

\subsection{Maximum likelihood estimation (MLE)}
Statistical inference for loss regression models is exclusively dominated by the maximum likelihood estimation (MLE) approach in actuarial practice, which maximizes the log-likelihood function
\begin{align}
\mathcal{L}^{\text{MLE}}_n(\bm{\Psi};\bm{y},\bm{X})=\sum_{i=1}^{n}\log f(y_i;\bm{x}_i,\bm{\Psi})
\end{align}
with respect to the parameters $\bm{\Psi}$. 
MLE is not robust to model contamination, i.e., a few outliers may significantly distort the estimated MLE parameters. It is, therefore, worthwhile to explore robust estimation methods alternative to the MLE to obtain more stable and reliable estimates of parameters.


\subsection{Maximum weighted likelihood estimation (MWLE)}
The maximum weighted likelihood estimation (MWLE) is developed by \cite{fung2021maximum} for robust estimations of loss distributions. The idea is to impose observation-dependent weights on each observation so that the observations deemed to harm the model robustness would make fewer impacts on parameter estimations. Slightly extending the MWLE proposed by \cite{fung2021maximum} to the regression setting, the following weighted log-likelihood function is maximized:
\begin{align} \label{eq:inference:mwle}
\mathcal{L}^{\text{MWLE}}_n(\bm{\Psi};\bm{y},\bm{X})=\sum_{i=1}^{n}W(y_i,\bm{x}_i)\log \frac{f(y_i;\bm{x}_i,\bm{\Psi})W(y_i,\bm{x}_i)}{\int_{\mathcal{Y}}f(u;\bm{x}_i,\bm{\Psi})W(u,\bm{x}_i)du}:=\sum_{i=1}^{n}W(y_i,\bm{x}_i)\log f^*(y_i;\bm{x}_i,\bm{\Psi}),
\end{align}
where $W(\cdot)$ is the weight function, and $f^*(y_i;\bm{x}_i,\bm{\Psi})$ is a transformed density function given by
\begin{equation}
f^*(y_i;\bm{x}_i,\bm{\Psi})=\frac{f(y_i;\bm{x}_i,\bm{\Psi})W(y_i,\bm{x}_i)}{\int_{\mathcal{Y}}f(u;\bm{x}_i,\bm{\Psi})W(u,\bm{x}_i)du}.
\end{equation}

\cite{fung2021maximum} assumes that $0\leq W(\cdot)\leq 1$ and $W(\cdot)$ is a non-decreasing function to address tail robustness issue. In this paper, we focus on robustness against the outliers, and hence we do not impose such an assumption on $W(\cdot)$. The appropriate choice of $W(\cdot)$ not only needs to address the modelling need (model robustness) but also has to result in an analytically tractable transformed density function $f^*(y_i;\bm{x}_i,\bm{\Psi})$ such that computational burden is minimized. Details will be covered in Section \ref{sec:swle:const}. An adjustment term $\int_{\mathcal{Y}}f(u;\bm{x}_i,\bm{\Psi})W(u,\bm{x}_i)du$ is incorporated into Equation (\ref{eq:inference:mwle}) to adjust for the estimation biases introduced by weighting the log-likelihood function. With this regard, the MWLE is consistent and asymptotically normal under several mild regularity conditions.

\subsection{Score-based weighted likelihood estimation (SWLE)}
The main shortcoming of the above MWLE approach is that it is difficult to be extended to cater for incomplete data where the true value of insurance loss severity $y_i$ may not be fully observed in exact due to censoring and truncation. This is because the weight function $W(y_i,\bm{x}_i)$ in Equation (\ref{eq:inference:mwle}) relies on the exact observable $y_i$, and it is difficult to determine the appropriate weight applied to an inexact loss. In the insurance loss modelling perspective, loss severities are expected to be censored and truncated due to various forms of coverage modifications such as deductibles (which lead to left truncation) and policy limits (which lead to right censoring) applied to insurance policies. 

With this regard, we propose a novel score-based weighted likelihood estimation (SWLE) for robust estimations of regression models while retaining its extensibility to the aforementioned incomplete insurance loss data. The SWLE is obtained by solving the following set of score functions w.r.t. $\bm{\Psi}$:
\begin{align} \label{eq:inference:swle}
\mathcal{S}_n(\bm{\Psi};\bm{y},\bm{X}):=\mathcal{S}^{\text{SWLE}}_n(\bm{\Psi};\bm{y},\bm{X})=\sum_{i=1}^{n}\left(\int_{\mathcal{Y}}f(u;\bm{x}_i,\bm{\Psi})W(u,\bm{x}_i)du\right)\frac{f^*(y_i;\bm{x}_i,\bm{\Psi})}{f(y_i;\bm{x}_i,\bm{\Psi})}\frac{\partial}{\partial\bm{\Psi}}\log f^*(y_i;\bm{x}_i,\bm{\Psi})=\bm{0}.
\end{align}

It is easy to show that the SWLE score function above is simply the derivative of the weighted log-likelihood function (Equation (\ref{eq:inference:mwle})) w.r.t. $\bm{\Psi}$, and hence they are equivalent when the observed data is complete. However, unlike the MWLE, SWLE avoids an explicit expression of the weight function $W(y_i,\bm{x}_i)$ into the score function, addressing the aforementioned extensibility problem. Note that the above SWLE score function is for complete data and has not been extended to the case where the insurance losses are censored and/or truncated. We will leverage the censoring-truncation mechanism for insurance losses and the corresponding extension of Equation (\ref{eq:inference:swle}) to Section \ref{sec:censtrun}.

\section{SWLE for GLM} \label{sec:swle}
This section examines the theoretical properties of the proposed SWLE under the GLM modelling framework for insurance loss regression analysis. We first construct an appropriate class of weight functions such that: (i) the contributions of extreme observations or outliers are effectively down-weighted to ensure robust model estimations; (ii) the resulting statistical quantities, including the score function, are analytically tractable, to ensure computational desirability. Then, we present the estimation algorithm for SWLE and discuss its connection to the MLE estimation algorithm. We will prove that the SWLE is consistent and asymptotically normal under mild regularity conditions with an analytically tractable information matrix. We will also show that the proposed SWLE results in a bounded influence function (IF), ensuring robustness against outliers.

\subsection{Construction} \label{sec:swle:const}
We consider the weight function with the following form
\begin{align} \label{eq:thm:weight}
W(y_i,\bm{x}_i)\propto\exp\left\{\frac{\tilde{\theta}_iy_i}{\tilde{\phi}}+\left(\frac{1}{\tilde{\phi}}-c\right)g(y_i)\right\},
\end{align}
where $\tilde{\theta}_i=\xi(\bm{x}_i^T\tilde{\bm{\beta}})$, and $\tilde{\bm{\Psi}}:=(\tilde{\bm{\beta}},\tilde{\phi})$ are the hyperparameters of the weight function, which control the extent that extreme observations are down-weighted and govern the trade-off between robust modelling and estimation efficiency. We start with the following lemma which suggests that the score function under the proposed SWLE is analytically tractable
\begin{lemma} \label{lem:density}
Suppose that the weight function $W(y_i,\bm{x}_i)$ and density function $f(y_i;\bm{x}_i,\bm{\Psi})$ are given by Equations (\ref{eq:thm:weight}) and (\ref{eq:glm}) respectively. Then we have:
\begin{enumerate}
\item The bias adjustment term is
\begin{align}
\lambda_i^*(\bm{\Psi};\bm{x}_i):=\int_{\mathcal{Y}}f(u;\bm{x}_i,\bm{\Psi})W(u,\bm{x}_i)du
\propto\exp\left\{\frac{A(\theta^*_i)}{\phi^*}-b(\phi^*)-\frac{A(\theta_i)}{\phi}+b(\phi)\right\};
\end{align}
\item The transformed density function is
\begin{align}
f^*(y_i;\bm{x}_i,\bm{\Psi})=\exp\left\{\frac{\theta_i^*y_i-A(\theta_i^*)}{\phi^*}+C(y_i,\phi^*)\right\},
\end{align}
where $\phi^*=(\phi^{-1}+\tilde{\phi}^{-1}-c)^{-1}$ and $\theta_i^*=(\theta_i/\phi+\tilde{\theta}_i/\tilde{\phi})\phi^*$.
\end{enumerate}
\end{lemma}

\begin{corollary} \label{cor:thm:density}
The following results based on Lemma \ref{lem:density} hold under the following two special cases:
\begin{enumerate}
\item If $\tilde{\theta}_i:=\tilde{\theta}$ is independent of $\bm{x}_i$ (e.g. choose $\tilde{\bm{\beta}}$ such that only the intercept term is non-zero), then we have $\theta_i^*=\xi^{*}(\bm{x}_i^T\bm{\beta})$ with the transformed mapping function $\xi^{*}(z)=(\xi(z)/\phi+\tilde{\theta}/\tilde{\phi})\phi^*$.
\item If a canonical link is selected for GLM such that $\xi(z)=z$, $\theta_i=\bm{x}_i^T\bm{\beta}$ and $\tilde{\theta}_i=\bm{x}_i^T\tilde{\bm{\beta}}$, then the transformed density function can be written as $f^*(y_i;\bm{x}_i,\bm{\Psi})=f(y_i;\bm{x}_i,\bm{\Psi}^*)$, where $\bm{\Psi}^*=(\bm{\beta}^*,\phi^*)$ with $\bm{\beta}^*=(\bm{\beta}/\phi+\tilde{\bm{\beta}}/\tilde{\phi})\phi^*$.
\end{enumerate}
\end{corollary}

We can still write the transformed density as an Exponential dispersion model (EDM) with shifted parameter values from the above results. If we choose a canonical link, the resulting transformed density will still be expressed as a GLM with transformed parameters. With the above analytically tractable formulas, the SWLE score function can be written as
\begin{align} \label{eq:thm:score_glm}
\mathcal{S}_n(\bm{\Psi};\bm{y},\bm{X})
:=\sum_{i=1}^{n}\mathcal{S}(\bm{\Psi};y_i,\bm{x}_i)
=\sum_{i=1}^{n}W(y_i,\bm{x}_i)\frac{\partial}{\partial\bm{\Psi}}\left[\frac{\theta_i^*y_i-A(\theta_i^*)}{\phi^*}+C(y_i,\phi^*)\right].
\end{align}

Taking a derivative with the usage of chain rule, SWLE for GLM requires solving the following two sets of equations simultaneously for $\bm{\Psi}$
\begin{align} \label{eq:thm:score_glm_t}
\mathcal{S}_{n,\theta}(\bm{\Psi};\bm{y},\bm{X})
&:=\sum_{i=1}^{n}\mathcal{S}_{\theta}(\bm{\Psi};y_i,\bm{x}_i)
:=\frac{1}{\phi}\sum_{i=1}^{n}W(y_i,\bm{x}_i)\left(y_i-A'(\theta_i^*)\right)\xi'(\bm{x}_i^T\bm{\beta})\bm{x}_i=\bm{0},
\end{align}
\begin{align} \label{eq:thm:score_glm_p}
&\mathcal{S}_{n,\phi}(\bm{\Psi};\bm{y},\bm{X})
:=\sum_{i=1}^{n}\mathcal{S}_{\theta}(\bm{\Psi};y_i,\bm{x}_i)\nonumber\\
&:=\frac{1}{\phi^2}\sum_{i=1}^{n}W(y_i,\bm{x}_i)
\left\{\left(y_i-A'(\theta_i^*)\right)\phi^*\left(\left(c-\frac{1}{\tilde{\phi}}\right)\theta_i+\frac{\tilde{\theta_i}}{\tilde{\phi}}\right)-\left[\theta_i^*y_i-A(\theta_i^*)+g(y_i)\right]+\phi^{*2}b'(\phi^*)\right\}=0,
\end{align}
which are both analytically tractable.

\begin{example} \label{eg:thm}
We hereby discuss the use of SWLE to three example GLM classes, which are commonly adopted for actuarial loss modelling and ratemaking purposes.
\begin{enumerate}
\item (\textbf{Gamma GLM}) Its density function is given by Equations (\ref{eq:glm}) and (\ref{eq:glm_C}) with $A(\theta)=-\log (-\theta)$, $c=1$, $g(y)=\log y$, $a(y)=0$ and $b(\phi)=\phi^{-1}\log(\phi^{-1})-\log\Gamma(\phi^{-1})$. To apply the proposed SWLE to Gamma GLM, the plausible weight function, according to Equation (\ref{eq:thm:weight}), is gamma density function itself (because $a(y)=0$). For simplicity, one may particularize $\tilde{\phi}=1$ and $\tilde{\theta}_i=\tilde{\theta}$ such that $W(y_i,\bm{x}_i):=W(y_i)=\exp\{\tilde{\theta}y_i\}$ exponentially distributed with $\tilde{\theta}<0$ being the only adjustable hyperparameter. In this case, $W(y_i)$ is a decreasing function of $y_i$, down-weighting large losses. Larger losses are down-weighted more significantly as $\tilde{\theta}$ becomes more negative, while $W(y_i)$ becomes flat (SWLE recovers back to MLE) as $\tilde{\theta}\rightarrow 0$. Under the SWLE, the transformed density function $f^{*}(y_i;\bm{x}_i,\bm{\Psi})$ still follows Gamma GLM with an identical transformed dispersion parameter $\phi^{*}=\phi$ and a shifted canonical parameter $\theta_i^*=\theta_i+\tilde{\theta}\phi$. The left panel of Figure \ref{fig:prelim} (using $\theta=-1$, $\phi=0.5$ and $\tilde{\theta}=-0.5$) shows that the transformed density function peaks more in the body part and eventually becomes lighter tailed. This is a natural consequence of incorporating weight functions to reduce the influence of large losses.
\item (\textbf{Linear model}) Actuaries often model log-normal losses via an exponential transformation of the linear model. We have $A(\theta)=\theta^2/2$, $c=0$, $g(y)=-y^2/2$, $a(y)=0$ and $b(\phi)=-(1/2)\log(2\pi\phi)$. Since $a(y)=0$, one can choose a linear model itself as the weight function. One may particularize $\tilde{\theta}_i:=\tilde{\theta}$ fixed as e.g. sample mean and treat $\tilde{\phi}>0$ as the only adjustable weight function hyperparameter. The weight function $W(y_i,\bm{x}_i):=W(y_i)$ is then a normal density with fixed mean $\tilde{\theta}$ and adjustable variance $\tilde{\phi}$, down-weighting the observations from both tails. Smaller $\tilde{\phi}$ represents outliers are down-weighted by a larger extent. The middle panel of Figure \ref{fig:prelim} (using $\theta=\tilde{\theta}=0$, $\phi=0.5$ and $\tilde{\phi}=2$) shows that the transformed density under SWLE is still normal distributed with a sharper peak, coinciding with the Gamma case.
\item (\textbf{Inverse-Gaussian GLM}) We have $A(\theta)=-(-2\theta)^{1/2}$, $c=0$, $g(y)=-1/(2y)$, $a(y)=-\ln (2\pi y^3)/2$ and $b(\phi)=\ln(1/\phi)/2$. As $a(y)\neq 0$, the weight function cannot be Inverse-Gaussian itself. Instead, we refer to Equation (\ref{eq:thm:weight}) and choose $W(y_i,\bm{x}_i)=\exp\{\tilde{\theta}_iy_i/\tilde{\phi}-(2y_i\tilde{\phi})^{-1}\}$ as the weight function. One may particularize the weight function $W(y_i,\bm{x}_i):=W(y_i)$ such that $\tilde{\theta}_i:=\tilde{\theta}<0$ is fixed. The adjustable hyperparameter $\tilde{\phi}>0$ governs how the outliers are down-weighted. Similar to the previous two cases, the transformed density under SWLE is still Inverse-Gaussian distributed with a sharper node in the body, as demonstrated by the right panel of Figure \ref{fig:prelim} (using $\theta=-0.5$, $\phi=1$, $\tilde{\theta}=-2$ and $\tilde{\phi}=1$).
\end{enumerate}
\end{example}

\begin{figure}[!h]
\begin{center}
\includegraphics[width=\linewidth]{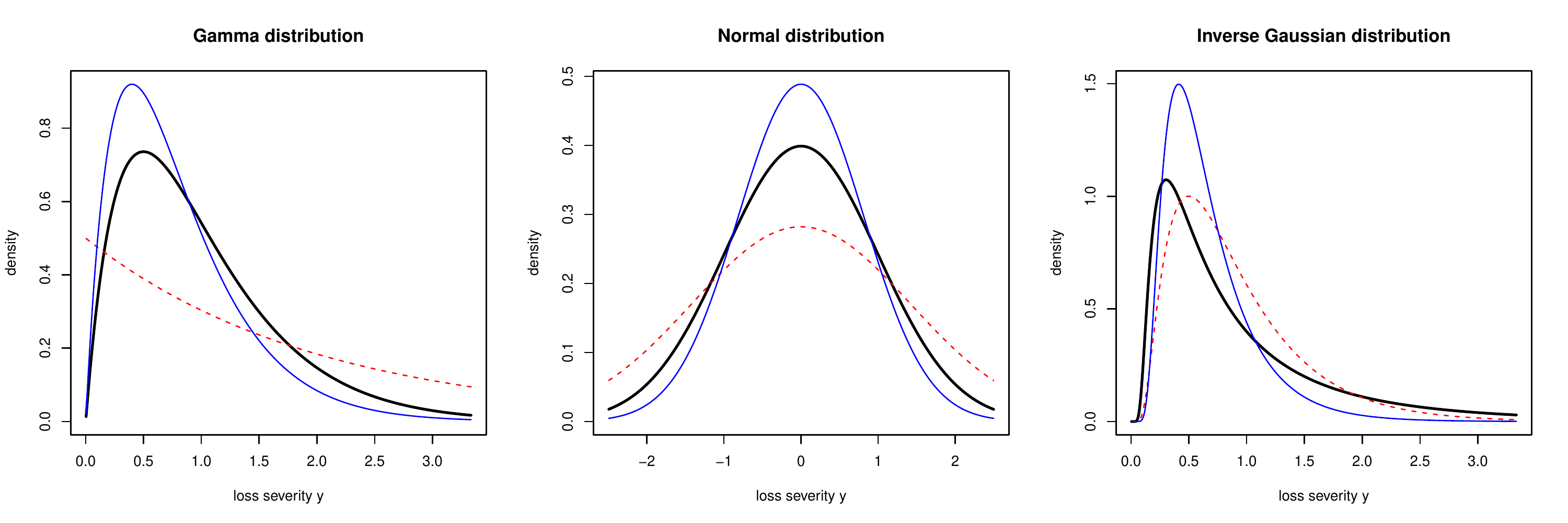}
\end{center}
\vspace{-0.5cm}
\caption{Weight function $W(y_i;\bm{x}_i)$ (dotted curves), original density function $f(y_i;\bm{x}_i,\bm{\Psi})$ (thick solid curves) and transformed density function $f^{*}(y_i;\bm{x}_i,\bm{\Psi})$ (thin solid curves), under Gamma (left panel), Normal (middle panel) and Inverse-Gaussian (right panel) distributions.}
\label{fig:prelim}
\end{figure}

\subsection{Estimation} \label{sec:swle:est}
With a carefully chosen weight function, the SWLE score functions (Equations (\ref{eq:thm:score_glm_t}) and (\ref{eq:thm:score_glm_p})) are analytically tractable. In particular, Equation (\ref{eq:thm:score_glm_t}) looks almost the same as the standard system of GLM score functions for the MLE regression parameter estimations. Therefore, it is still possible to employ a standard IRLS-type approach to estimate the GLM parameters under the proposed SWLE approach. Making use of this desirable property, we propose an SWLE model calibration algorithm, which allows borrowing and adapting the existing software packages (such as \texttt{glm} function in \texttt{R}) to estimate SWLE parameters efficiently. The steps of the algorithm are outlined as follows:
\begin{enumerate}
\item Set the initial parameters as $\bm{\Psi}^{[0]}:=(\bm{\beta}^{[0]},\phi^{[0]})$.
\item For each iteration $r=1,2,\ldots$, do:
\begin{itemize}
	\item Update the regression coefficients $\bm{\beta}^{[r]}$ by solving Equation (\ref{eq:thm:score_glm_t}) with the dispersion parameter $\phi=\phi^{[r-1]}$ fixed. This can be done by the IRLS procedures directly using \texttt{glm} function in \texttt{R}, setting prior weights \texttt{weights} as $W(y_i,\bm{x}_i)$ and custom link function \texttt{link} as $\eta^*(\cdot)=(A'\circ\xi^*)^{-1}(\cdot)$, where $\xi^*(\cdot)$ is the transformed mapping function (Corollary \ref{cor:thm:density}).
	\item Update the dispersion parameter $\phi^{[r]}$ by solving Equation (\ref{eq:thm:score_glm_p}) with the regression coefficients being fixed as $\bm{\beta}=\bm{\beta}^{[r-1]}$. This can be simply done by Newton-Raphson method using \texttt{uniroot} function in \texttt{R}.
\end{itemize}
\item Continue iterating step 2 until the absolute change of the iterated parameter values is smaller than a particular threshold (e.g. $|\bm{\Psi}^{[r]}-\bm{\Psi}^{[r-1]}|<10^{-6}$).
\end{enumerate}

If the canonical link function is chosen, from Corollary \ref{cor:thm:density}, the resulting transformed density function $f^{*}(y_i;\bm{x}_i,\bm{\Psi}^{*})$ is still a GLM with the same link function but transformed parameters $\bm{\Psi}^{*}=(\bm{\beta}^{*},\phi^{*})$. Further, the score function in Equation (\ref{eq:thm:score_glm_t}) can be simplified as
\begin{align} \label{eq:est:score_glm_t}
\sum_{i=1}^{n}W(y_i,\bm{x}_i)\left(y-A'(\theta_i^*)\right)\bm{x}_i=\bm{0},
\end{align}
which is a set of standard weighted GLM score functions which depends only on the transformed regression coefficients $\bm{\beta}^{*}$. As a result, it is even more computationally appealing to estimate $\bm{\Psi}^{*}$ first and then transform it back to $\bm{\Psi}$. The proposed algorithm consists of the following steps:
\begin{enumerate}
\item Obtain an IRLS estimate of the transformed regression parameters $\hat{\bm{\beta}}^{*}$ from Equation (\ref{eq:est:score_glm_t}), using \texttt{glm} function in \texttt{R} with \texttt{weights} being $W(y_i,\bm{x}_i)$ and \texttt{link} being the canonical function.
\item Obtain an estimated transformed dispersion parameter $\hat{\phi}^{*}$ by solving Equation (\ref{eq:thm:score_glm_p}) with the transformed regression coefficients being fixed as the estimated values in the previous step. This can be done by \texttt{uniroot} function in \texttt{R}.
\item Revert the transformed parameters $\hat{\bm{\Psi}}^{*}$ to obtain the estimated parameters $\hat{\bm{\Psi}}$:
\begin{align}
\hat{\phi}=\left(\frac{1}{\hat{\phi}^{*}}-\frac{1}{\tilde{\phi}}+c\right)^{-1}
,\qquad \hat{\beta}=\left(\frac{\hat{\bm{\beta}}^{*}}{\hat{\phi}^{*}}-\frac{\tilde{\bm{\beta}}}{\tilde{\phi}}\right)\phi.
\end{align}
\end{enumerate}

No further iterations are needed in steps 1 and 2 above. In this case, the computational burden under SWLE is almost the same as that under MLE.

\subsection{Asymptotic properties} \label{sec:swle:asymp}
The following theorem shows that the proposed SWLE approach leads to the convergence to true model parameters asymptotically as $n\rightarrow\infty$ and quantify the asymptotic parameter uncertainties:

\begin{theorem} \label{thm:asymp}
Suppose that $Y_i|\bm{x}_i$ follows GLM with density function in the form of Equation (\ref{eq:glm}) and true model parameters $\bm{\Psi}_0:=(\bm{\beta}_0,\phi_0)$. Assume that the mild regularity conditions outlined in Section \ref{apx:sec:reg_com_asymp} of the Appendix are satisfied. Then, there exists a solution $\hat{\bm{\Psi}}_n:=(\hat{\bm{\beta}}_n,\hat{\phi}_n)$ of the SWLE score equations $\mathcal{S}_{n,\theta}(\bm{\Psi};\bm{y},\bm{X})$ and $\mathcal{S}_{n,\phi}(\bm{\Psi};\bm{y},\bm{X})$ (Equations (\ref{eq:thm:score_glm_t}) and (\ref{eq:thm:score_glm_p})) such that
\begin{align}
\sqrt{n}(\hat{\bm{\Psi}}_n-\bm{\Psi}_0)\overset{d}{\rightarrow}\mathcal{N}(\bm{0},\bm{\Sigma}),
\end{align}
where $\bm{\Sigma}:=\bm{\Sigma}(\bm{\Psi}_0)=(\Gamma^{-1})\Lambda(\Gamma^{-1})^T$, with $\Gamma$ and $\Lambda$ being $(P+1)\times (P+1)$ matrices given by
\begin{align} \label{eq:asymp:matrix}
\Gamma:=\Gamma(\bm{\Psi}_0)=
\begin{pmatrix}
\Gamma_{\theta\theta}(\bm{\Psi}_0) & \Gamma_{\theta\phi}(\bm{\Psi}_0) \\
\Gamma_{\phi\theta}(\bm{\Psi}_0) & \Gamma_{\phi\phi}(\bm{\Psi}_0)
\end{pmatrix};\qquad
\Lambda:=\Lambda(\bm{\Psi}_0)=
\begin{pmatrix}
\Lambda_{\theta\theta}(\bm{\Psi}_0) & \Lambda_{\theta\phi}(\bm{\Psi}_0) \\
\Lambda_{\phi\theta}(\bm{\Psi}_0) & \Lambda_{\phi\phi}(\bm{\Psi}_0)
\end{pmatrix},
\end{align}
where the elements of the matrices are expressed as $\Gamma_{\theta\theta}(\bm{\Psi})=E_{\bm{x}}\left[W_{\theta\theta}(\bm{\Psi},\bm{x})\bm{x}\bm{x}^T\right]$, $\Gamma_{\theta\phi}(\bm{\Psi})=\Gamma_{\phi\theta}(\bm{\Psi})^T=E_{\bm{x}}\left[W_{\theta\phi}(\bm{\Psi},\bm{x})\bm{x}\right]$, $\Gamma_{\phi\phi}(\bm{\Psi})
=E_{\bm{x}}\left[W_{\phi\phi}(\bm{\Psi},\bm{x})\right]$, $\Lambda_{\theta\theta}(\bm{\Psi})=E_{\bm{x}}\left[V_{\theta\theta}(\bm{\Psi},\bm{x})\bm{x}\bm{x}^T\right]$, $\Lambda_{\theta\phi}(\bm{\Psi})
=\Lambda_{\phi\theta}(\bm{\Psi})^T=E_{\bm{x}}\left[V_{\theta\phi}(\bm{\Psi},\bm{x})\bm{x}\right]$ and $\Lambda_{\phi\phi}(\bm{\Psi})=E_{\bm{x}}\left[V_{\phi\phi}(\bm{\Psi},\bm{x})\right]$. Here, $E_{\bm{x}}[\cdot]$ is the expectation taken on $\bm{x}$.The analytical expressions of $W_{\theta\theta}(\bm{\Psi},\bm{x})$, $W_{\theta\phi}(\bm{\Psi},\bm{x})$, $W_{\phi\phi}(\bm{\Psi},\bm{x})$, $V_{\theta\theta}(\bm{\Psi},\bm{x})$, $V_{\theta\phi}(\bm{\Psi},\bm{x})$ and $V_{\phi\phi}(\bm{\Psi},\bm{x})$ are listed in Equations (\ref{apx:eq:w_tt})--(\ref{apx:eq:w_pp}) and (\ref{apx:eq:v_tt})--(\ref{apx:eq:v_pp}) in Section \ref{apx:sec:cov_asymp} of the Appendix.
If a canonical regression link is selected (i.e., $\xi(\cdot)$ is an identity function), the solution $\hat{\bm{\Psi}}_n$ will be unique.
\end{theorem}

In practice, we are unable to obtain the exact covariance matrix $\bm{\Sigma}$ above because the true model parameters $\bm{\Psi}_0$ are unobserved, and the randomness of $\bm{x}$ is not modelled explicitly to compute the expectation $E_{\bm{x}}[\cdot]$. We, therefore, estimate the uncertainties of fitted model parameters $\hat{\bm{\Psi}}_n$ as follows:
\begin{align}
\text{Var}(\hat{\bm{\Psi}}_n)\approx\frac{1}{n}(\tilde{\Gamma}^{-1})\tilde{\Lambda}(\tilde{\Gamma}^{-1})^T,
\end{align}
where
\begin{align}
\tilde{\Gamma}:=
\begin{pmatrix}
\tilde{\Gamma}_{\theta\theta}(\hat{\bm{\Psi}}_n) & \tilde{\Gamma}_{\theta\phi}(\hat{\bm{\Psi}}_n) \\
\tilde{\Gamma}_{\phi\theta}(\hat{\bm{\Psi}}_n) & \tilde{\Gamma}_{\phi\phi}(\hat{\bm{\Psi}}_n)
\end{pmatrix}
=\frac{1}{n}
\begin{pmatrix}
\bm{X}^T\tilde{W}_{\theta\theta}(\hat{\bm{\Psi}}_n,\bm{X})\bm{X} & \tilde{W}_{\theta\phi}(\hat{\bm{\Psi}}_n,\bm{X})\bm{X} \\
\bm{X}^T\tilde{W}_{\phi\theta}(\hat{\bm{\Psi}}_n,\bm{X}) & \tilde{W}_{\phi\phi}(\hat{\bm{\Psi}}_n,\bm{X})
\end{pmatrix}
\end{align}
and
\begin{align}
\tilde{\Lambda}:=
\begin{pmatrix}
\tilde{\Lambda}_{\theta\theta}(\hat{\bm{\Psi}}_n) & \tilde{\Lambda}_{\theta\phi}(\hat{\bm{\Psi}}_n) \\
\tilde{\Lambda}_{\phi\theta}(\hat{\bm{\Psi}}_n) & \tilde{\Lambda}_{\phi\phi}(\hat{\bm{\Psi}}_n)
\end{pmatrix}
=\frac{1}{n}
\begin{pmatrix}
\bm{X}^T\tilde{V}_{\theta\theta}(\hat{\bm{\Psi}}_n,\bm{X})\bm{X} & \tilde{V}_{\theta\phi}(\hat{\bm{\Psi}}_n,\bm{X})\bm{X} \\
\bm{X}^T\tilde{V}_{\phi\theta}(\hat{\bm{\Psi}}_n,\bm{X}) & \tilde{V}_{\phi\phi}(\hat{\bm{\Psi}}_n,\bm{X})
\end{pmatrix}
\end{align}
with $\tilde{W}_{\theta\theta}(\hat{\bm{\Psi}}_n,\bm{X})=\text{diag}(\{W_{\theta\theta}(\hat{\bm{\Psi}}_n,\bm{x}_i)\}_{i=1,\ldots,n})$, 
$\tilde{W}_{\theta\phi}(\hat{\bm{\Psi}}_n,\bm{X})=\tilde{W}_{\phi\theta}(\hat{\bm{\Psi}}_n,\bm{X})^T=(\{W_{\theta\phi}(\hat{\bm{\Psi}}_n,\bm{x}_i)\}_{i=1,\ldots,n})$, 
$\tilde{W}_{\phi\phi}(\hat{\bm{\Psi}}_n,\bm{X})=\sum_{i=1}^{n}W_{\phi\phi}(\hat{\bm{\Psi}}_n,\bm{x}_i)$, 
$\tilde{V}_{\theta\theta}(\hat{\bm{\Psi}}_n,\bm{X})=\text{diag}(\{V_{\theta\theta}(\hat{\bm{\Psi}}_n,\bm{x}_i)\}_{i=1,\ldots,n})$, 
$\tilde{V}_{\theta\phi}(\hat{\bm{\Psi}}_n,\bm{X})=\tilde{V}_{\phi\theta}(\hat{\bm{\Psi}}_n,\bm{X})^T=(\{V_{\theta\phi}(\hat{\bm{\Psi}}_n,\bm{x}_i)\}_{i=1,\ldots,n})$, and 
$\tilde{V}_{\phi\phi}(\hat{\bm{\Psi}}_n,\bm{X})=\sum_{i=1}^{n}V_{\phi\phi}(\hat{\bm{\Psi}}_n,\bm{x}_i)$.

\subsection{Robustness analysis}
The robustness of the proposed SWLE can also be justified by showing that the SWLE has a bounded sensitivity against model perturbations. Assume that $(Y_i,\bm{x}_i)$ are generated by a contamination model
\begin{align} \label{eq:robust:model}
\tilde{F}(y_i,\bm{x}_i;\epsilon,\Delta,\bm{\Psi}_0)=(1-\epsilon)F(y_i,\bm{x}_i;\bm{\Psi}_0)+\epsilon\Delta(y_i,\bm{x}_i),
\end{align}
where $F(y_i,\bm{x}_i;\bm{\Psi}_0)=F(y_i;\bm{x}_i,\bm{\Psi}_0)H(\bm{x}_i)$ is a joint distribution on $(Y_i,\bm{x}_i)$ with $F(y_i;\bm{x}_i,\bm{\Psi}_0)$ being the cdf of $Y_i|\bm{x}_i$ under the GLM in Equation (\ref{eq:glm}), whereas $\Delta(y_i,\bm{x}_i)$ is a contamination distribution function on $(Y_i,\bm{x}_i)$. The sensitivity of the estimated parameters by model contaminations can be evaluated using the influence function (IF, \cite{hampel1974influence}):
\begin{align} \label{eq:robust:if}
\text{IF}(\bm{\Psi}_0;F,\Delta)=\lim_{\epsilon\rightarrow 0}\frac{\tilde{\bm{\Psi}}^{\epsilon,\Delta}-\bm{\Psi}_0}{\epsilon},
\end{align}
where $\tilde{\bm{\Psi}}^{\epsilon,\Delta}$ are the asymptotic estimated parameters if the data generating model is $\tilde{F}(y_i,\bm{x}_i;\epsilon,\Delta,\bm{\Psi}_0)$ in Equation (\ref{eq:robust:model}). It is obvious that the influence function can be unbounded under the MLE approach when the contamination distribution function $\Delta(y_i,\bm{x}_i)$ assigns a probability mass on an arbitrarily (extremely) large $y_i$. On the other hand, we have the following theorem to ensure that the influence function is bounded under the SWLE.

\begin{theorem} \label{thm:if}
Suppose that the following assumptions are satisfied:
\begin{enumerate}[label=(\roman*)]
\item There exists a compact space $\bar{\mathcal{X}}$ such that $\int_{\mathcal{Y}\times\bar{\mathcal{X}}}d\Delta(y_i,\bm{x}_i)=1$, i.e., the covariates generated by the perturbed distribution are bounded by $\bar{\mathcal{X}}$ almost surely.
\item There exists a finite function $\mathcal{P}(\bm{x}_i)$ such that $|W(y_i,\bm{x}_i)|\leq \mathcal{P}(\bm{x}_i)$, $|W(y_i,\bm{x}_i)y_i|\leq \mathcal{P}(\bm{x}_i)$ and $|W(y_i,\bm{x}_i)g(y_i)|\leq \mathcal{P}(\bm{x}_i)$ for every $y_i\in\mathcal{Y}$.
\item The mild regularity conditions outlined in Section \ref{apx:sec:reg_com_asymp} of the Appendix are satisfied.
\end{enumerate}
Denote $\bm{\Upsilon}$ as a collection of all distribution functions on $(y_i,\bm{x}_i)$. Under the SWLE, we have
\begin{align} \label{eq:robust:if_bound}
\sup_{\Delta\in\bm{\Upsilon}}|\text{IF}(\bm{\Psi}_0;F,\Delta)|<\infty.
\end{align}
\end{theorem}

\begin{remark} \label{rmk:robust}
The rationale of stating Assumption (i) in Theorem \ref{thm:if} is that the proposed SWLE primarily aims to ensure the robustness of the estimated parameters against the outliers on $y_i$ instead of $\bm{x}_i$. In insurance practice, it often makes sense to consider bounded covariates space only: For the European automobile dataset in Section \ref{sec:data:euro}, variables are either categorical (e.g., car fuel type) or discrete (e.g. policyholder age) with practical upper limits (e.g., no policyholders are older than 120 years).
\end{remark}

Finally, we need to check the validity of Assumption (ii) of Theorem \ref{thm:if} under some example cases:

\begin{corollary} \label{cor:if}
Suppose that $F(y_i;\bm{x}_i,\bm{\Psi}_0)$ is the cdf of the Gamma GLM, linear model or inverse-Gaussian GLM specified in Example \ref{eg:thm}, and the weight function is in the form of Equation (\ref{eq:thm:weight}), with the hyperparameters selected in accordance to Example \ref{eg:thm}. Then, Assumption (ii) of Theorem \ref{thm:if} holds.
\end{corollary}

\section{Model diagnostic with SWLE} \label{sec:diag}
The previous sections introduce the SWLE as a statistical inference tool that is less sensitive against outliers. In insurance practice, the modelling challenge may not arise just from model contamination but also from a more fundamental problem that the true (data generating) model systematically deviates from the GLM. Such deviations may include tail-heaviness, distributional multimodality, non-linear regression links, and dispersion heterogeneity. In this case, the GLM may produce misleading pricing recommendations, and hence actuaries must reconsider alternative modelling frameworks. The research question is, are there any quantitative measures to recommend if one is worthwhile to try alternative models. In other words, we want to test the null hypothesis that
$$
H_0:\text{The data is generated by a specified class of GLM}
$$
against the alternative hypothesis $H_1$ that $H_0$ is false. In this section, we propose a novel Wald-test statistic based on the SWLE to quantitatively assess the hypothesis. Under $H_0$, $Y_i|\bm{x}_i$ follows the GLM with density function given by Equation (\ref{eq:glm}) and true (unknown) model parameters $\bm{\Psi}_0:=(\bm{\beta}_0,\phi_0)$.
For $k=1,\ldots,K$, denote $\hat{\bm{\Psi}}^{(k)}_n$ as the estimated parameters under the SWLE where the weight function is given by Equation (\ref{eq:thm:weight}) with hyperparameters $\tilde{\bm{\Psi}}^{(k)}:=(\tilde{\bm{\beta}}^{(k)},\tilde{\phi}^{(k)})$ and $\tilde{\theta}_i=\xi(\bm{x}_i^T\tilde{\bm{\beta}}^{(k)})$. Here, $K$ represents the total number of different sets of weight functions hyperparameters we consider. 

Under $H_0$, we expect that the choice of weight function hyperparameters would not impact the estimated parameters significantly because the SWLE is asymptotically consistent (Theorem \ref{thm:asymp}) regardless of the hyperparameters chosen. In other words, large absolute difference of estimated parameters $|\hat{\bm{\Psi}}^{(k)}_n-\hat{\bm{\Psi}}^{(k')}_n|$ (for some $k\neq k'$) is an evidence to reject $H_0$. In this case, the data generating model may not be within the specified GLM class, and further explorations of alternative models are recommended. This motivates us to introduce the following two theorems, which provide a foundation on the construction of Wald statistic based on $(\hat{\bm{\Psi}}^{(2)}_n-\hat{\bm{\Psi}}^{(1)}_n,\hat{\bm{\Psi}}^{(3)}_n-\hat{\bm{\Psi}}^{(2)}_n,\ldots,\hat{\bm{\Psi}}^{(K)}_n-\hat{\bm{\Psi}}^{(K-1)}_n)$.
We first denote $\hat{\bm{\Psi}}_n^{\text{meta}}=(\hat{\bm{\Psi}}_n^{(1)},\ldots,\hat{\bm{\Psi}}_n^{(K)})$ as a $1\times(P+1)K$ horizontal vector containing all estimated parameters under various weight function hyperparameters. Also define $\bm{\Psi}_0^{\text{meta}}=(\bm{\Psi}_0,\ldots,\bm{\Psi}_0)$ as a $1\times(P+1)K$ horizontal vector containing $K$ sets of true model parameters. The following results hold:

\begin{theorem} \label{thm:diag:asymp}
Under $H_0$ with true model parameters $\bm{\Psi}_0$, and given that the mild regularity conditions outlined in Section \ref{apx:sec:reg_com_diag} of the Appendix hold, we have
\begin{align}
\sqrt{n}\left(\hat{\bm{\Psi}}_n^{\text{meta}}-\bm{\Psi}_0^{\text{meta}}\right)\overset{d}{\rightarrow}\mathcal{N}(\bm{0},\bm{\Sigma}^{\text{meta}}),
\end{align}
where $\bm{\Sigma}^{\text{meta}}$ is a $(P+1)K\times(P+1)K$ matrix given by
\begin{align} \label{eq:diag:sigma}
{\tiny
\Sigma^{\text{meta}}=
\begin{pmatrix}
\left(\left[\Gamma^{(1)}\right]^{-1}\right)\Lambda^{(1,1)}\left(\left[\Gamma^{(1)}\right]^{-1}\right)^T & \left(\left[\Gamma^{(1)}\right]^{-1}\right)\Lambda^{(1,2)}\left(\left[\Gamma^{(2)}\right]^{-1}\right)^T & 
\dots & 
\left(\left[\Gamma^{(1)}\right]^{-1}\right)\Lambda^{(1,K)}\left(\left[\Gamma^{(K)}\right]^{-1}\right)^T \\
\left(\left[\Gamma^{(2)}\right]^{-1}\right)\Lambda^{(2,1)}\left(\left[\Gamma^{(1)}\right]^{-1}\right)^T & 
\left(\left[\Gamma^{(2)}\right]^{-1}\right)\Lambda^{(2,2)}\left(\left[\Gamma^{(2)}\right]^{-1}\right)^T & 
\dots & 
\left(\left[\Gamma^{(2)}\right]^{-1}\right)\Lambda^{(2,K)}\left(\left[\Gamma^{(K)}\right]^{-1}\right)^T \\
\vdots & \vdots & \ddots & \vdots \\
\left(\left[\Gamma^{(K)}\right]^{-1}\right)\Lambda^{(K,1)}\left(\left[\Gamma^{(1)}\right]^{-1}\right)^T & 
\left(\left[\Gamma^{(K)}\right]^{-1}\right)\Lambda^{(K,2)}\left(\left[\Gamma^{(2)}\right]^{-1}\right)^T & 
\dots & 
\left(\left[\Gamma^{(K)}\right]^{-1}\right)\Lambda^{(K,K)}\left(\left[\Gamma^{(K)}\right]^{-1}\right)^T
\end{pmatrix},
}%
\end{align}
with the matrix elements analytically expressed in Section \ref{apx:sec:cov_diag} of the Appendix.
\end{theorem}

\begin{theorem} \label{thm:diag:chisq}
Denote $\bm{J}$ as a $Q\times(P+1)K$ design matrix. Then, we have
\begin{align}
n\left[\bm{J}\left(\hat{\bm{\Psi}}_n^{\text{meta}}-\bm{\Psi}_0^{\text{meta}}\right)^T\right]^T\left(\bm{J}\bm{\Sigma}^{\text{meta}}\bm{J}^T\right)^{-1}\left[\bm{J}\left(\hat{\bm{\Psi}}_n^{\text{meta}}-\bm{\Psi}_0^{\text{meta}}\right)^T\right]
\overset{d}{\rightarrow}\chi^2_{Q},
\end{align}
where $\chi^2_{Q}$ is a chi-square distribution with $Q$ degrees of freedom.
\end{theorem}

Based on the above theorem, one can develop various versions of Wald-type statistics as follows:
\begin{enumerate}
\item Meta Wald statistic: We aggregate all the estimated parameters differences $\Delta\hat{\bm{\Psi}}^{\text{meta}}_n:=(\hat{\bm{\Psi}}^{(2)}_n-\hat{\bm{\Psi}}^{(1)}_n,\hat{\bm{\Psi}}^{(3)}_n-\hat{\bm{\Psi}}^{(2)}_n,\ldots,\hat{\bm{\Psi}}^{(K)}_n-\hat{\bm{\Psi}}^{(K-1)}_n)$ to create a combined (meta) test statistic. The $(P+1)(K-1)\times(P+1)K$ design matrix $\bm{J}$ in the above theorem is chosen as
\begin{align}
\bm{J}^{\text{meta}}=
\begin{pmatrix}
\bm{I} & -\bm{I} & \bm{0} & \dots & \bm{0} & \bm{0}\\
\bm{0} & \bm{I} & -\bm{I} & \dots & \bm{0} & \bm{0}\\
\vdots & & \ddots & & & \vdots\\
\bm{0} & \bm{0} & \bm{0} & \dots & \bm{I} & -\bm{I}
\end{pmatrix},
\end{align}
where $\bm{I}$ is a $(P+1)\times (P+1)$ diagonal matrix. The Wald-type meta statistic is given by
\begin{align} \label{eq:diag_wald_meta}
Z^{\text{meta}}_n=n\left(\Delta\hat{\bm{\Psi}}^{\text{meta}}_n\right)\left[\left(\bm{J}^{\text{meta}}\right)\left(\hat{\bm{\Sigma}}^{\text{meta}}\right)\left(\bm{J}^{\text{meta}}\right)^T\right]^{-1}\left(\Delta\hat{\bm{\Psi}}^{\text{meta}}_n\right)^T\overset{\cdot}{\sim}\chi^2_{(P+1)(K-1)},
\end{align}
where $\hat{\bm{\Sigma}}^{\text{meta}}$ is the estimated covariance matrix $\bm{\Sigma}^{\text{meta}}$ in Theorem \ref{thm:diag:asymp} evaluated at the fitted MLE parameters.
The meta Wald statistic provides a single value to quantitatively assess the overall adequateness of using GLM to fit the data.
\item Individual Wald statistic: We perform pairwise comparisons of estimated parameters between two specific sets of weight function hyperparameters (say, $\tilde{\bm{\Psi}}^{(k)}$ and $\tilde{\bm{\Psi}}^{(k')}$). In this case, the design matrix $\bm{J}^{\text{ind}}=(\bm{0},\cdots,\bm{0},\bm{I},\bm{0},\cdots,\bm{0},-\bm{I},\bm{0},\cdots,\bm{0})$ has a dimension of $(P+1)\times(P+1)K$, with only the $k$-th and $k'$-th blocks being non-zero. The individual Wald statistic is
\begin{align} \label{eq:diag_wald_ind}
Z_n^{\text{ind}}=n\left(\hat{\bm{\Psi}}_n^{(k)}-\hat{\bm{\Psi}}_n^{(k')}\right)\left[\hat{\bm{\Sigma}}^{(k,k')}\right]^{-1}\left(\hat{\bm{\Psi}}_n^{(k)}-\hat{\bm{\Psi}}_n^{(k')}\right)^T\overset{\cdot}{\sim}\chi^2_{(P+1)},
\end{align}
{\sloppy where we have $\bm{\Sigma}^{(k,k')}=\left(\left[\Gamma^{(k)}\right]^{-1}\right)\Lambda^{(k,k)}\left(\left[\Gamma^{(k)}\right]^{-1}\right)^T+\left(\left[\Gamma^{(k')}\right]^{-1}\right)\Lambda^{(k',k')}\left(\left[\Gamma^{(k')}\right]^{-1}\right)^T-\left(\left[\Gamma^{(k')}\right]^{-1}\right)\Lambda^{(k',k)}\left(\left[\Gamma^{(k)}\right]^{-1}\right)^T-\left(\left[\Gamma^{(k)}\right]^{-1}\right)\Lambda^{(k,k')}\left(\left[\Gamma^{(k')}\right]^{-1}\right)^T$, and $\hat{\bm{\Sigma}}^{(k,k')}$ is the estimated covariance $\bm{\Sigma}^{(k,k')}$ evaluated at fitted MLE parameters.}
\item Parameter-specific meta Wald statistic: We focus on a certain parameter of interest. From an insurance ratemaking perspective, actuaries are more interested in the regression coefficients to differentiate policyholders into various risk categories. From a risk management perspective, actuaries may be more interested in the dispersion parameter $\phi$, which governs the distribution's extreme losses. Suppose that the $p$-th parameter $\Psi_p$ is the parameter of interest ($p=1,\ldots,P+1$). The design matrix will then have a dimension of $(K-1)\times (P+1)K$, given by
\begin{align}
\bm{J}_p^{\text{meta}}=
\begin{pmatrix}
\bm{e}_p & -\bm{e}_p & \bm{0} & \dots & \bm{0} & \bm{0}\\
\bm{0} & \bm{e}_p & -\bm{e}_p & \dots & \bm{0} & \bm{0}\\
\vdots & & \ddots & & & \vdots\\
\bm{0} & \bm{0} & \bm{0} & \dots & \bm{e}_p & -\bm{e}_p
\end{pmatrix},
\end{align}
where $\bm{e}_p$ is a $(P+1)$ horizontal vector with the $p$-th element equals to one and zero otherwise. Denote $\Delta\hat{\Psi}^{\text{meta}}_{p,n}=(\hat{\Psi}^{(2)}_{p,n}-\hat{\Psi}^{(1)}_{p,n},\ldots,\hat{\Psi}^{(K)}_{p,n}-\hat{\Psi}^{(K-1)}_{p,n})$ as the aggregations of estimated parameter differences corresponding to the $p$-th parameter. The Wald-type statistic is
\begin{align} \label{eq:diag_wald_param_meta}
Z^{\text{meta}}_{p,n}=n\left(\Delta\hat{\Psi}^{\text{meta}}_{p,n}\right)\left[\left(\bm{J}^{\text{meta}}_p\right)\left(\hat{\bm{\Sigma}}^{\text{meta}}\right)\left(\bm{J}^{\text{meta}}_p\right)^T\right]^{-1}\left(\Delta\hat{\Psi}^{\text{meta}}_{p,n}\right)^T\overset{\cdot}{\sim}\chi^2_{(K-1)}.
\end{align}
\item Individual parameter-specific Wald statistic: We do pairwise comparisons of a single parameter between two sets of weight function hyperparameters (i.e., comparing $\hat{\Psi}_{p,n}^{(k)}$ to $\hat{\Psi}_{p,n}^{(k')}$). The Wald statistic is
\begin{align}
Z^{\text{ind}}_{p,n}=n(\hat{\Psi}_{p,n}^{(k)}-\hat{\Psi}_{p,n}^{(k')})^2\left[\hat{\bm{\Sigma}}^{(k,k')}\right]_{p,p}^{-1}\overset{\cdot}{\sim}\chi^2_{(1)},
\end{align}
where $\left[\hat{\bm{\Sigma}}^{(k,k')}\right]_{p,p}$ is the $(p,p)$-th element of $\hat{\bm{\Sigma}}^{(k,k')}$.
\end{enumerate}

\section{Extending SWLE to censored and truncated data} \label{sec:censtrun}
In general insurance practice, the actual values of (transformed) losses $\bm{y}=(y_1,\ldots,y_n)$ may not be observed in exact (censoring) and may not be fully observed (truncation) due to coverage modifications of insurance policies including deductibles and policy limits. As a result, extending the above SWLE framework is vital for random censored and truncated regression data.

We formulate the censoring and truncation mechanisms in accordance to \cite{FUNG2020MoECensTrun}. Denote $\mathcal{T}_i\subseteq\mathcal{Y}$ as a random truncation interval of observation $i$, meaning that a loss $Y_i$ is observed conditioned on $Y_i\in\mathcal{T}_i$. Further, define $\mathcal{U}_i\subseteq\mathcal{T}_i$ and $\mathcal{C}_i=\mathcal{T}_i\symbol{92}\mathcal{U}_i$ as the random uncensoring and censoring regions respectively. Denote $\{\mathcal{I}_{i1},\ldots,\mathcal{I}_{iM_i}\}$ as $M_i$ disjoint random censoring intervals of observation $i$ with $\cup_{m=1}^{M_i}\mathcal{I}_{im}=\mathcal{C}_i$. Then, $\mathcal{R}_i:=(\mathcal{U}_i,\mathcal{C}_i,M_i,\{\mathcal{I}_{i1},\ldots,\mathcal{I}_{iM_i}\})$ is called the censoring mechanism of observation $i$. Under censoring framework, the loss $Y_i$ is observed in exact if $Y_i\in\mathcal{U}_i$, while we would only know which censoring interval the loss belongs to (i.e. $1\{Y_i\in\mathcal{I}_{i1}\},\ldots,1\{Y_i\in\mathcal{I}_{iM_i}\}$) if $Y_i\in\mathcal{C}_i$. As a result, the observed (incomplete) information for loss $i$ is given by $\mathcal{D}_i:=(\mathcal{R}_i,\mathcal{T}_i,y_i1\{y_i\in\mathcal{U}_i\},\{1\{y_i\in\mathcal{I}_{im}\}\}_{m=1,\ldots,M_i})$. Denote $\mathcal{D}:=\{\mathcal{D}_i\}_{i=1,\ldots,n}$ as the observed information of all losses.

Note that the censoring and truncation mechanisms $(\mathcal{R}_i,\mathcal{T}_i)$ are observed but may differ across $i$, so the data is censored and truncated in random. This makes sense from an insurance perspective because different policyholders may choose different deductibles or policy limits which affect the censoring and truncation points. Also, the above formalism represents a general framework that includes left, right, and interval censoring and truncation. See \cite{FUNG2020MoECensTrun} for more details.

We now extend the SWLE score function in Equation (\ref{eq:inference:swle}) such that statistical inference is still possible under the above random censoring and truncation mechanisms. We propose that the extended SWLE score function is given by
\begin{align} \label{eq:censtrun:swle}
\mathcal{S}_n(\bm{\Psi};\mathcal{D},\bm{X})
&:=\sum_{i=1}^n\mathcal{S}(\bm{\Psi};\mathcal{D}_i,\bm{x}_i)\nonumber\\
&=\sum_{i=1}^{n}\left(\int_{\mathcal{Y}}f_{\mathcal{T}_i}(u;\bm{x}_i,\bm{\Psi})W(u,\bm{x}_i)du\right)\frac{f^*_{\mathcal{T}_i}(y_i;\bm{x}_i,\bm{\Psi})}{f_{\mathcal{T}_i}(y_i;\bm{x}_i,\bm{\Psi})}\frac{\partial}{\partial\bm{\Psi}}\log f^*_{\mathcal{T}_i}(y_i;\bm{x}_i,\bm{\Psi})1\{y_i\in\mathcal{U}_i\}\nonumber\\
&\quad + \sum_{i=1}^{n}\left(\int_{\mathcal{Y}}f_{\mathcal{T}_i}(u;\bm{x}_i,\bm{\Psi})W(u,\bm{x}_i)du\right)\sum_{m=1}^{M_i}\frac{F^*_{\mathcal{T}_i}(\mathcal{I}_{im};\bm{x}_i,\bm{\Psi})}{F_{\mathcal{T}_i}(\mathcal{I}_{im};\bm{x}_i,\bm{\Psi})}\frac{\partial}{\partial\bm{\Psi}}\log F^*_{\mathcal{T}_i}(\mathcal{I}_{im};\bm{x}_i,\bm{\Psi})1\{y_i\in\mathcal{I}_{im}\}\\
&=\sum_{i=1}^nW(y_i,\bm{x}_i)\frac{\partial}{\partial\bm{\Psi}}\log f^*_{\mathcal{T}_i}(y_i;\bm{x}_i,\bm{\Psi})1\{y_i\in\mathcal{U}_i\}\nonumber\\
&\quad +\sum_{i=1}^n\sum_{m=1}^{M_i}\lambda^{*}(\bm{\Psi};\bm{x}_i)\frac{F^*(\mathcal{I}_{im};\bm{x}_i,\bm{\Psi})}{F(\mathcal{I}_{im};\bm{x}_i,\bm{\Psi})}\frac{\partial}{\partial\bm{\Psi}}\log F^*_{\mathcal{T}_i}(\mathcal{I}_{im};\bm{x}_i,\bm{\Psi})1\{y_i\in\mathcal{I}_{im}\},
\end{align}
where $f_{\mathcal{T}_i}(y_i;\bm{x}_i,\bm{\Psi})$ and $f^*_{\mathcal{T}_i}(y_i;\bm{x}_i,\bm{\Psi})$ are the truncated density functions given by
\begin{equation} \label{eq:censtrun:density}
f_{\mathcal{T}_i}(y_i;\bm{x}_i,\bm{\Psi})=\frac{f(y_i;\bm{x}_i,\bm{\Psi})}{F(\mathcal{T}_i;\bm{x}_i,\bm{\Psi})}1\{y_i\in\mathcal{T}_i\},\qquad f^*_{\mathcal{T}_i}(y_i;\bm{x}_i,\bm{\Psi})=\frac{f^*(y_i;\bm{x}_i,\bm{\Psi})}{F^*(\mathcal{T}_i;\bm{x}_i,\bm{\Psi})}1\{y_i\in\mathcal{T}_i\},
\end{equation}
$F(\cdot;\bm{x}_i,\bm{\Psi})$, $F^*(\cdot;\bm{x}_i,\bm{\Psi})$, $F_{\mathcal{T}_i}(\cdot;\bm{x}_i,\bm{\Psi})$ and $F^*_{\mathcal{T}_i}(\cdot;\bm{x}_i,\bm{\Psi})$ are the distribution functions of $f(\cdot;\bm{x}_i,\bm{\Psi})$, $f^*(\cdot;\bm{x}_i,\bm{\Psi})$, $f_{\mathcal{T}_i}(\cdot;\bm{x}_i,\bm{\Psi})$ and $f^*_{\mathcal{T}_i}(\cdot;\bm{x}_i,\bm{\Psi})$ respectively, i.e., $F(\mathcal{A};\bm{x}_i,\bm{\Psi})=\int_{\mathcal{A}}f(u;\bm{x}_i,\bm{\Psi})du$ for any $\mathcal{A}\subseteq\mathbb{R}$. 

The extended SWLE makes two major modifications compared to the original score function in Equation (\ref{eq:inference:swle}) for complete data. Firstly, the density functions (e.g., $f(y_i;\bm{x}_i,\bm{\Psi})$) are changed to truncated density functions (e.g., $f_{\mathcal{T}_i}(y_i;\bm{x}_i,\bm{\Psi})$), reflecting that $Y_i$ is observed conditioned on $Y_i\in\mathcal{T}_i$. Secondly, Equation (\ref{eq:censtrun:swle}) is segregated into two terms. The first term is very similar to the original score function, reflecting that full information $y_i$ is used for the evaluation of score function if the observation is uncensored. The second term represents the modified score function as the observation falls into the censoring region, where the only information known is the identification of the censoring interval $\mathcal{I}_{im}$ an observation belongs to. In this case, the density functions in the score function are changed to distribution functions evaluated at the censoring interval $\mathcal{I}_{im}$.

The following result theoretically justifies the consistency and asymptotic normality of the extended SWLE under the GLM framework.
\begin{theorem} \label{thm:censtrun:asymp}
Suppose that $\tilde{Y}_i|\bm{x}_i$ (the loss random variable before truncation) independently follows the GLM with density function given by Equation (\ref{eq:glm}) and true model parameters $\bm{\Psi}_0:=(\bm{\beta}_0,\phi_0)$ for $i=1,\ldots,n$. Each loss $i$ is also equipped by censoring and truncation mechanisms $(\mathcal{R}_i,\mathcal{T}_i)$ described above. Assume that $\tilde{Y}_i$ is independent of $(\mathcal{R}_i,\mathcal{T}_i)$ conditioned on $\bm{x}_i$ for every $i=1,\ldots,n$. Also denote $Y_i=\tilde{Y}_i|\tilde{Y}_i\in\mathcal{T}_i$ as the truncated loss random variable. Suppose that the observed information is $\mathcal{D}:=\{\mathcal{D}_i\}_{i=1,\ldots,n}$ described above. If the mild regularity conditions outlined in Section \ref{apx:sec:reg_incom_asymp} of the Appendix are satisfied, then there exists a solution $\hat{\bm{\Psi}}_n:=(\hat{\bm{\beta}}_n,\hat{\phi}_n)$ satisfying the extended SWLE score equations $\mathcal{S}_{n}(\hat{\bm{\Psi}}_n;\mathcal{D},\bm{X})=\bm{0}$ (Equations (\ref{eq:censtrun:swle})) such that \begin{align}
\sqrt{n}(\hat{\bm{\Psi}}_n-\bm{\Psi}_0)\overset{d}{\rightarrow}\mathcal{N}(\bm{0},\bm{\Sigma}),
\end{align}
where $\bm{\Sigma}:=\bm{\Sigma}(\bm{\Psi}_0)=(\Gamma^{-1})\Lambda(\Gamma^{-1})^T$, with $\Gamma$ and $\Lambda$ being $(P+1)\times (P+1)$ matrices given by Section \ref{apx:thm:censtrun:asymp} of the Appendix.
\end{theorem}

\begin{remark}
From Theorem \ref{thm:censtrun:asymp} above and Section \ref{apx:thm:censtrun:asymp} of the Appendix, the covariance matrix $\bm{\Sigma}$ depends on the first and second derivatives of the cdf when the observations are censored and truncated. These terms can still be expressed analytically for the linear model and inverse-Gaussian GLM. For Gamma GLM, these terms can be expressed as incomplete di-gamma and tri-gamma functions, which can be computed using \texttt{\textup{pgamma.deriv}} function within the \texttt{\textup{heavy}} package in \texttt{\textup{R}}. 
\end{remark}

Analogous to Section \ref{sec:diag}, one can construct a Wald-based test statistic to assess whether a specified class of GLM is appropriate for a dataset with censoring and truncation mechanisms. We here denote $\hat{\bm{\Psi}}^{(k)}_n$ as the solution satisfying the extended SWLE equations $\mathcal{S}_{n}(\hat{\bm{\Psi}}_n^{(k)};\mathcal{D},\bm{X})=\bm{0}$ in Equation (\ref{eq:censtrun:swle}) with weight function hyperparameters chosen as $\tilde{\bm{\Psi}}^{(k)}$ for $k=1,\ldots,K$. Also recall that $\hat{\bm{\Psi}}_n^{\text{meta}}=(\hat{\bm{\Psi}}_n^{(1)},\ldots,\hat{\bm{\Psi}}_n^{(K)})$ and $\bm{\Psi}_0^{\text{meta}}=(\bm{\Psi}_0,\ldots,\bm{\Psi}_0)$ defined in Section \ref{sec:diag}. The following theorem holds:

\begin{theorem} \label{thm:censtrun:diag}
Under $H_0$ with true model parameters $\bm{\Psi}_0$, and given that the mild regularity conditions outlined in Section \ref{apx:sec:reg_incom_diag} of the Appendix hold, we have
\begin{align} \label{eq:censtrun:thm_asymp}
\sqrt{n}\left(\hat{\bm{\Psi}}_n^{\text{meta}}-\bm{\Psi}_0^{\text{meta}}\right)\overset{d}{\rightarrow}\mathcal{N}(\bm{0},\bm{\Sigma}^{\text{meta}}),
\end{align}
where $\bm{\Sigma}^{\text{meta}}$ is a $(P+1)K\times(P+1)K$ matrix given by Section \ref{apx:thm:censtrun:diag} of the Appendix, and hence
\begin{align} \label{eq:censtrun:thm_wald}
n\left[\bm{J}\left(\hat{\bm{\Psi}}_n^{\text{meta}}-\bm{\Psi}_0^{\text{meta}}\right)^T\right]^T\left(\bm{J}\bm{\Sigma}^{\text{meta}}\bm{J}^T\right)^{-1}\left[\bm{J}\left(\hat{\bm{\Psi}}_n^{\text{meta}}-\bm{\Psi}_0^{\text{meta}}\right)^T\right]
\overset{d}{\rightarrow}\chi^2_{Q}
\end{align}
for a $Q\times(P+1)K$ design matrix $\bm{J}$.
\end{theorem}

With the above theorem, the Wald-type diagnostic test statistic for censored and truncated data can be constructed as described by Section \ref{sec:diag}.

\section{Simulation studies} \label{sec:sim}

\subsection{Simulation 1: Various GLMs} \label{sec:sim1}
This study aims to empirically verify the asymptotic properties of the proposed SWLE and evaluate its finite-sample performance. In each simulation, we generate $n$ observations $\{(y_i,\bm{x}_i)\}_{i=1,\ldots,n}$ with $P=2$ (for simplicity) so that $\bm{x}_i=(x_{i1},x_{i2})$. We set $x_{i1}=1$ as an intercept term and generate $x_{i2}$ iid from $N(0,1)$. 
The simulation design is as follows:
\begin{enumerate}
\item Data generating model for $Y_i|\bm{x}_i$: We consider Gamma GLM, linear model and inverse-Gaussian GLM with the following link function and parameter settings:
\begin{itemize}
\item Gamma GLM: A log-link $\mu_i=-1/\theta_i=\exp\{\bm{x}_i^T\bm{\beta}\}$ is selected. The parameters are specified as $\bm{\beta}=(1,0.5)^T$ and $\phi=0.5$.
\item Linear model: A linear link $\mu_i=\theta_i=\bm{x}_i^T\bm{\beta}$ is selected. The parameters are specified as $\bm{\beta}=(1,0.5)^T$ and $\phi=0.25$.
\item Inverse-Gaussian GLM: A log-link $\mu_i=(-2\theta_i)^{-1/2}=\exp\{\bm{x}_i^T\bm{\beta}\}$ is selected. The parameters are specified as $\bm{\beta}=(1,0.5)^T$ and $\phi=0.1$.
\end{itemize}
\item Sample size: $n=250,1000,2500,10000,25000$, aligning with the range of sample sizes for insurance loss data, from a few hundred data points for e.g. Secura-Re (\texttt{ReIns} package in \texttt{R}) to 10,000+ data points for French automobile insurance (\texttt{CASdatasets} package in \texttt{R}).
\item Weight functions $W(y_i,\bm{x}_i)$: Selected in accordance to Example \ref{eg:thm}.
\item Weight function hyperparameters $\tilde{\bm{\Psi}}$ and the number of hyperparameter sets $K$ considered for evaluating meta Wald statistic: The hyperparameter $\tilde{\theta}$ or $\tilde{\phi}$ is determined by solving
\begin{align} \label{eq:sim_wgt_ratio}
\frac{E_{Y,\bm{x}}[W(Y,\bm{x})|Y>q_{\alpha}]}{E_{Y,\bm{x}}[W(Y,\bm{x})]}=\delta
\end{align}
for some fixed $0<\alpha<1$ and $0<\delta\leq 1$, where $q_{\alpha}$ is the $\alpha$-percentile of $Y$. An interpretation of the above equation is as follows: We want to choose the hyperparameter such that the average weight assigned to the extreme observations (larger than $q_{\alpha}$) is only $\delta\leq 1$ times as the overall average weight. Smaller $\delta$ means extreme observations are down-weighted more. Note that the above equation is easy to compute (analytically or through simulation) since the true model is known. For example, when $\alpha=0.99$ and $\delta=0.1$ for Gamma GLM, then solving Equation (\ref{eq:sim_wgt_ratio}) we have $(\tilde{\theta},\tilde{\phi})=(6.53,1)$. Note that if $\delta=1$, the weight function will be flat, and hence the proposed SWLE will be equivalent to the MLE. We select $\alpha=0.99$ and consider the following choices of $K$ for model diagnostic purposes:
\begin{itemize}
\item $K=2$: the two sets of hyperparameters are constructed based on $\delta=1,0.001$.
\item $K=3$: the three sets of hyperparameters are constructed based on $\delta=1,0.1,0.001$.
\item $K=5$: the five sets of hyperparameters are constructed based on $\delta=1,0.5,0.1,0.01,0.001$.
\end{itemize}
\item Censoring and truncation: Not considered in this experiment for simplicity.
\end{enumerate}

Each combination of data-generating model, $n$ and $K$ selected above results in a simulated dataset with a sample size of $n$. We replicate (simulate) each combination by $B=500$ times to ensure thorough investigation on the adequateness of SWLE. We then fit each simulated dataset into Gamma GLM, linear model, and inverse-Gaussian GLM, respectively. If the fitted model class matches with the data generating model, then we are fitting a correct model class, and hence we should expect that the rejection rate of the Wald-type diagnostic test presented in Section \ref{sec:diag} is low. Otherwise, the fitted model class is misspecified, and a high rejection rate is expected.

We first examine the case of correct model specification to verify the consistencies of the SWLE fitted parameters. Table \ref{table:sim1_est} depicts the true model parameters $(\beta_1,\beta_2,\phi)$ versus the fitted SWLE parameters $(\hat{\beta}_1,\hat{\beta}_2,\hat{\phi})$ (averaged across the $B=500$ replications), with weight function hyperparameters selected based on Equation (\ref{eq:sim_wgt_ratio}) with $\delta=1$ (this simply reduces to MLE), $\delta=0.1$ and $\delta=0.01$ respectively. The SWLE fitted parameters are very close to the true values for all settings even if the sample size is relatively small ($n=250$), empirically justifying the adequacy of the proposed SWLE in recovering the true model parameters.

We then analyze the performance of the meta Wald-type diagnostic tool presented in Equation (\ref{eq:diag_wald_meta}) of Section \ref{sec:diag} based on SWLE, considering both cases of correct and misspecified fitted models. Table \ref{table:sim1_prob} presents the rejection probabilities at 5\% significance level using the meta Wald statistic across different true (data-generating) models, fitted models, $n$ and $K$. As expected, the rejection rates are almost equal to 1 for most cases when the fitted model class is misspecified. Exceptions are when the sample size is small ($n=250$) enough to hinder the power of the proposed diagnostic test. Further, the rejection probabilities are mostly close to the desired level of 5\% when the model is correctly specified. Exceptions are when $K=5$ with a small sample size $n<2500$ (the rejection probabilities are inflated). An interpretation is that as $K$ grows large, the meta Wald-type statistics become more mathematically complicated, and hence a larger sample size is needed for convergence to the asymptotic results. 

\begin{table}[h]
\centering
\resizebox{\textwidth}{!}{%
\begin{tabular}{llc|ccc}
\hline
\multicolumn{1}{c}{} & \multicolumn{1}{c}{} & true parameters & \multicolumn{3}{c}{mean estimates under SWLE} \\
\multicolumn{1}{c}{} & \multicolumn{1}{c}{} & $(\beta_1,\beta_2,\phi)$ & \multicolumn{3}{c}{$(\hat{\beta}_1,\hat{\beta}_2,\hat{\phi})$} \\ \hline
\multicolumn{1}{c|}{Model} & \multicolumn{1}{c|}{} &  & \multicolumn{1}{c|}{$\delta=1$ (MLE)} & \multicolumn{1}{c|}{$\delta=0.1$} & $\delta=0.001$ \\ \hline
\multicolumn{1}{l|}{} & \multicolumn{1}{l|}{$n=250$} &  & \multicolumn{1}{c|}{(0.999, 0.502, 0.496)} & \multicolumn{1}{c|}{(1.001, 0.504, 0.496)} & (1.004, 0.506, 0.496) \\
\multicolumn{1}{l|}{} & \multicolumn{1}{l|}{$n=1000$} &  & \multicolumn{1}{c|}{(1.000, 0.500, 0.499)} & \multicolumn{1}{c|}{(1.000, 0.500, 0.499)} & (1.001, 0.500, 0.499) \\
\multicolumn{1}{l|}{Gamma} & \multicolumn{1}{l|}{$n=2500$} & (1, 0.5, 0.5) & \multicolumn{1}{c|}{(1.000, 0.501, 0.499)} & \multicolumn{1}{c|}{(1.000, 0.501, 0.499)} & (1.000, 0.502, 0.499) \\
\multicolumn{1}{l|}{} & \multicolumn{1}{l|}{$n=10000$} &  & \multicolumn{1}{c|}{(1.000, 0.500, 0.500)} & \multicolumn{1}{c|}{(1.000, 0.500, 0.500)} & (0.999, 0.500, 0.499) \\
\multicolumn{1}{l|}{} & \multicolumn{1}{l|}{$n=25000$} &  & \multicolumn{1}{c|}{(1.000, 0.500, 0.500)} & \multicolumn{1}{c|}{(1.000, 0.500, 0.500)} & (1.000, 0.500, 0.500) \\ \hline
\multicolumn{1}{l|}{} & \multicolumn{1}{l|}{$n=250$} &  & \multicolumn{1}{c|}{(1.002, 0.500, 0.249)} & \multicolumn{1}{c|}{(1.002, 0.500, 0.249)} & (1.002, 0.504, 0.251) \\
\multicolumn{1}{l|}{} & \multicolumn{1}{l|}{$n=1000$} &  & \multicolumn{1}{c|}{(1.001, 0.500, 0.249)} & \multicolumn{1}{c|}{(1.001, 0.501, 0.249)} & (1.001, 0.502, 0.249) \\
\multicolumn{1}{l|}{Normal} & \multicolumn{1}{l|}{$n=2500$} & (1, 0.5, 0.25) & \multicolumn{1}{c|}{(1.000, 0.500, 0.249)} & \multicolumn{1}{c|}{(1.000, 0.500, 0.249)} & (1.001, 0.501, 0.250) \\
\multicolumn{1}{l|}{} & \multicolumn{1}{l|}{$n=10000$} &  & \multicolumn{1}{c|}{(1.000, 0.500, 0.250)} & \multicolumn{1}{c|}{(1.000, 0.500, 0.250)} & (1.000, 0.500, 0.250) \\
\multicolumn{1}{l|}{} & \multicolumn{1}{l|}{$n=25000$} &  & \multicolumn{1}{c|}{(1.000, 0.500, 0.250)} & \multicolumn{1}{c|}{(1.000, 0.500, 0.250)} & (1.000, 0.500, 0.250) \\ \hline
\multicolumn{1}{l|}{} & \multicolumn{1}{l|}{$n=250$} &  & \multicolumn{1}{c|}{(1.001, 0.501, 0.099)} & \multicolumn{1}{c|}{(1.001, 0.501, 0.100)} & (1.002, 0.503, 0.100) \\
\multicolumn{1}{l|}{} & \multicolumn{1}{l|}{$n=1000$} &  & \multicolumn{1}{c|}{(0.999, 0.500, 0.100)} & \multicolumn{1}{c|}{(1.000, 0.500, 0.100)} & (1.000, 0.501, 0.099) \\
\multicolumn{1}{l|}{Inv-Gauss} & \multicolumn{1}{l|}{$n=2500$} & (1, 0.5, 0.1) & \multicolumn{1}{c|}{(0.999, 0.499, 0.100)} & \multicolumn{1}{c|}{(0.999, 0.499, 0.100)} & (1.000, 0.500, 0.100) \\
\multicolumn{1}{l|}{} & \multicolumn{1}{l|}{$n=10000$} &  & \multicolumn{1}{c|}{(1.000, 0.500, 0.100)} & \multicolumn{1}{c|}{(1.000, 0.500, 0.100)} & (1.000, 0.500, 0.100) \\
\multicolumn{1}{l|}{} & \multicolumn{1}{l|}{$n=25000$} &  & \multicolumn{1}{c|}{(1.000, 0.500, 0.100)} & \multicolumn{1}{c|}{(1.000, 0.500, 0.100)} & (1.000, 0.500, 0.100) \\ \hline
\end{tabular}
}
\caption{[Simulation 1] True versus SWLE fitted parameters with correct model specifications.}
\label{table:sim1_est}
\end{table}

\begin{table}[h]
\centering
\resizebox{\textwidth}{!}{%
\begin{tabular}{lllllllllll}
\hline
 &  & \multicolumn{9}{c}{True model} \\ \cline{3-11} 
 &  & \multicolumn{3}{c}{Gamma} & \multicolumn{3}{c}{Normal} & \multicolumn{3}{c}{Inverse-Gaussian} \\ \hline
\multicolumn{1}{c|}{Fitted} & \multicolumn{1}{l|}{} & $K=2$ & $K=3$ & \multicolumn{1}{l|}{$K=5$} & $K=2$ & $K=3$ & \multicolumn{1}{l|}{$K=5$} & $K=2$ & $K=3$ & $K=5$ \\ \hline
\multicolumn{1}{l|}{} & \multicolumn{1}{l|}{$n=250$} & 0.056 & 0.112 & \multicolumn{1}{l|}{0.292} & 0.654 & 0.580 & \multicolumn{1}{l|}{0.724} & 0.998 & 0.998 & 0.998 \\
\multicolumn{1}{l|}{} & \multicolumn{1}{l|}{$n=1000$} & 0.048 & 0.052 & \multicolumn{1}{l|}{0.134} & 0.998 & 0.998 & \multicolumn{1}{l|}{0.998} & 1.000 & 1.000 & 1.000 \\
\multicolumn{1}{l|}{Gamma} & \multicolumn{1}{l|}{$n=2500$} & 0.034 & 0.042 & \multicolumn{1}{l|}{0.094} & 1.000 & 1.000 & \multicolumn{1}{l|}{1.000} & 1.000 & 1.000 & 1.000 \\
\multicolumn{1}{l|}{} & \multicolumn{1}{l|}{$n=10000$} & 0.046 & 0.048 & \multicolumn{1}{l|}{0.096} & 1.000 & 1.000 & \multicolumn{1}{l|}{1.000} & 1.000 & 1.000 & 1.000 \\
\multicolumn{1}{l|}{} & \multicolumn{1}{l|}{$n=25000$} & 0.048 & 0.062 & \multicolumn{1}{l|}{0.064} & 1.000 & 1.000 & \multicolumn{1}{l|}{1.000} & 1.000 & 1.000 & 1.000 \\ \hline
\multicolumn{1}{l|}{} & \multicolumn{1}{l|}{$n=250$} & 0.664 & 0.648 & \multicolumn{1}{l|}{0.818} & 0.050 & 0.052 & \multicolumn{1}{l|}{0.108} & 0.786 & 0.724 & 0.642 \\
\multicolumn{1}{l|}{} & \multicolumn{1}{l|}{$n=1000$} & 1.000 & 1.000 & \multicolumn{1}{l|}{1.000} & 0.042 & 0.046 & \multicolumn{1}{l|}{0.088} & 1.000 & 1.000 & 0.998 \\
\multicolumn{1}{l|}{Normal} & \multicolumn{1}{l|}{$n=2500$} & 1.000 & 1.000 & \multicolumn{1}{l|}{1.000} & 0.060 & 0.050 & \multicolumn{1}{l|}{0.054} & 1.000 & 1.000 & 1.000 \\
\multicolumn{1}{l|}{} & \multicolumn{1}{l|}{$n=10000$} & 1.000 & 1.000 & \multicolumn{1}{l|}{1.000} & 0.056 & 0.048 & \multicolumn{1}{l|}{0.048} & 1.000 & 1.000 & 1.000 \\
\multicolumn{1}{l|}{} & \multicolumn{1}{l|}{$n=25000$} & 1.000 & 1.000 & \multicolumn{1}{l|}{1.000} & 0.038 & 0.048 & \multicolumn{1}{l|}{0.056} & 1.000 & 1.000 & 1.000 \\ \hline
\multicolumn{1}{l|}{} & \multicolumn{1}{l|}{$n=250$} & 1.000 & 1.000 & \multicolumn{1}{l|}{1.000} & 0.922 & 0.880 & \multicolumn{1}{l|}{0.846} & 0.042 & 0.052 & 0.080 \\
\multicolumn{1}{l|}{} & \multicolumn{1}{l|}{$n=1000$} & 1.000 & 1.000 & \multicolumn{1}{l|}{1.000} & 1.000 & 1.000 & \multicolumn{1}{l|}{1.000} & 0.056 & 0.044 & 0.088 \\
\multicolumn{1}{l|}{Inv-Gauss} & \multicolumn{1}{l|}{$n=2500$} & 1.000 & 1.000 & \multicolumn{1}{l|}{1.000} & 1.000 & 1.000 & \multicolumn{1}{l|}{1.000} & 0.056 & 0.048 & 0.064 \\
\multicolumn{1}{l|}{} & \multicolumn{1}{l|}{$n=10000$} & 1.000 & 1.000 & \multicolumn{1}{l|}{1.000} & 1.000 & 1.000 & \multicolumn{1}{l|}{1.000} & 0.038 & 0.044 & 0.050 \\
\multicolumn{1}{l|}{} & \multicolumn{1}{l|}{$n=25000$} & 1.000 & 1.000 & \multicolumn{1}{l|}{1.000} & 1.000 & 1.000 & \multicolumn{1}{l|}{1.000} & 0.032 & 0.038 & 0.082 \\ \hline
\end{tabular}
}
\caption{[Simulation 1] Meta Wald statistic rejection probabilities at 5\% significance level.}
\label{table:sim1_prob}
\end{table}

\subsection{Simulation 2: Heavy-tail contaminated linear model} \label{sec:sim2}
This study reveals how model contamination leads to unstable MLE estimates and how the proposed SWLE approach detects and addresses the robustness issues. In each simulation, we generate $n=5000$ observations $\{(y_i,\bm{x}_i)\}_{i=1,\ldots,n}$ with $P=2$. $\bm{x}_i$ is generated by the same distribution as the previous study. $Y_i|\bm{x}_i$ is simulated by a contaminated regression model with the following density function:

\begin{align}
f^c(y_i;\bm{x}_i)=(1-\epsilon)f(y_i;\bm{x}_i,\bm{\Psi})+\epsilon\gamma(y_i;\bm{x}_i),
\end{align}
where $f(y_i;\bm{x}_i,\bm{\Psi})$ is chosen as a linear model with parameters $\bm{\Psi}=(\bm{\beta},\phi)=(1,0.5,0.25)$, $\epsilon$ is the contamination probability, and $\gamma(y_i;\bm{x}_i)$ is a contamination density function. We choose $\epsilon=0.1$ to be small so that the linear model is only slightly perturbed. The contamination density $\gamma(y_i;\bm{x}_i)$ is chosen as a scaled and translated Student's t-distribution with 2.5 degrees of freedom, scaled and translated in a way such that the mean is $\mu_i=1+0.5x_{i2}$ and the variance is $\sigma^2=0.25$ (aligning with the linear model $f(y_i;\bm{x}_i,\bm{\Psi})$). Outliers will be more prevalent in the simulated data with such heavy-tailed contamination. We ignore the censoring and truncation effects.

Similar to the previous simulation study, the simulation is replicated by $B=500$ times. For each replication, the resulting simulated dataset is fitted to the linear model, using the SWLE approach and considering $K=5$ sets of weight function hyperparameters constructed based on Equation (\ref{eq:sim_wgt_ratio}) with $\delta=1,0.5,0.1,0.01,0.001$. 

Table \ref{table:sim2_param} presents the estimated parameters (averaged across $B=500$ replications) and the corresponding standard errors (SE) for each of the five weight function hyperparameter sets considered. In the table, $k$ is the hyperparameter set index. For example, when $k=1$, the hyperparameters are selected based on $\delta=1$, leading to a standard MLE approach. $\delta$ is reduced as $k$ increases, resulting in more substantial down-weightings on the outliers. The MLE approach results in an unrobust estimated dispersion parameter $\hat{\phi}$, as evidenced by an abnormally large SE. The outliers severely distort the estimated parameters under the MLE approach. This issue can be effectively mitigated by the proposed SWLE approach: As $k=2$, the SE of $\hat{\phi}$ is reduced significantly from 0.023 to 0.005. On the other hand, the standard errors of any estimated parameters generally increase as $k$ increases. This is natural as more substantial down-weightings (i.e., larger $k$) often imply that more data points are effectively discarded for model estimation purposes, leading to a higher SE. Overall, Table \ref{table:sim2_param} reveals a trade-off between estimation robustness and efficiency when the data-generating model is contaminated. In this case, the choice of $k=2$ or $k=3$ may result in the best fitted model as the stability of estimated parameters is guaranteed without substantially inflating the SE.

We then perform extensive diagnostic tests on the fitted SWLE models. The meta Wald test in Equation (\ref{eq:diag_wald_meta}) shows that the linear model is rejected in 448 out of the 500 replications (89.6\%), suggesting that our proposed SWLE-based diagnostic tool is quite powerful in detecting model contaminations. We further perform individual Wald tests (Equation (\ref{eq:diag_wald_ind}) for each pair of weight function hyperparameter sets $(k,k')$) to carefully examine how the simulated dataset deviates from the fitted linear model. The left panel of Table \ref{table:sim2_matrix} showcases the rejection rates of the individual Wald tests for each pair of hyperparameter sets $(k,k')$. While the rejection rate is very high (0.908) when $(k,k')=(1,2)$, it gradually decreases as $k$ and $k'$ increase. As we note that $k=1$ represents the MLE approach, we may conclude that after reducing the influence of extreme observations (by choosing $k\geq 2$), the simulated data behaves less significantly deviated from the linear model. To showcase an example, we report the individual Wald statistics and the p-values for one specific representative simulation replication in the right panel of Table \ref{table:sim2_matrix}. In this case, we observe that the individual Wald tests fail to reject the linear model, provided that the outliers are already sufficiently down-weighted (with $k,k'\geq 2$). As a result, there is no evidence that the simulated data systematically deviates from the linear model. Instead, the deviation is solely caused by the few outliers caused by model perturbations.

\begin{table}[h]
\centering
\begin{tabular}{crrrrrrr}
\hline
 & \multicolumn{1}{c}{} & \multicolumn{2}{c}{$\hat{\beta}_1$} & \multicolumn{2}{c}{$\hat{\beta}_2$} & \multicolumn{2}{c}{$\hat{\phi}$} \\ \hline
$k$ & \multicolumn{1}{c}{$\delta$} & \multicolumn{1}{c}{Mean} & \multicolumn{1}{c}{SE} & \multicolumn{1}{c}{Mean} & \multicolumn{1}{c}{SE} & \multicolumn{1}{c}{Mean} & \multicolumn{1}{c}{SE} \\ \hline
1 & 1 & 1.000 & 0.007 & 0.500 & 0.007 & 0.248 & 0.023 \\
2 & 0.5 & 1.000 & 0.007 & 0.500 & 0.007 & 0.237 & 0.005 \\
3 & 0.1 & 1.000 & 0.008 & 0.500 & 0.008 & 0.233 & 0.006 \\
4 & 0.01 & 1.000 & 0.009 & 0.500 & 0.011 & 0.230 & 0.007 \\
5 & 0.001 & 1.000 & 0.010 & 0.500 & 0.014 & 0.228 & 0.008 \\ \hline
\end{tabular}
\caption{[Simulation 2] Mean estimated parameters and the corresponding standard errors under SWLE across five different sets of weight function hyperparameters.}
\label{table:sim2_param}
\end{table}

\begin{table}[h]
\centering
\begin{tabular}{cllllllclllll}
\cline{1-6} \cline{8-13}
\multicolumn{6}{c}{Individual Wald statistic rejection rate} & \multicolumn{1}{c}{} & \multicolumn{6}{c}{Individual Wald statistic \& p-value} \\ \cline{1-6} \cline{8-13} 
$k~\backslash ~k'$ & \multicolumn{1}{c}{1} & \multicolumn{1}{c}{2} & \multicolumn{1}{c}{3} & \multicolumn{1}{c}{4} & \multicolumn{1}{c}{5} & \multicolumn{1}{c}{} & $k~\backslash ~k'$ & \multicolumn{1}{c}{1} & \multicolumn{1}{c}{2} & \multicolumn{1}{c}{3} & \multicolumn{1}{c}{4} & \multicolumn{1}{c}{5} \\ \cline{1-6} \cline{8-13} 
1 & \multicolumn{1}{c}{---} & 0.908 & 0.864 & 0.778 & 0.678 &  & 1 & \multicolumn{1}{c}{---} & 0.000 & 0.000 & 0.000 & 0.003 \\
2 &  & \multicolumn{1}{c}{---} & 0.536 & 0.38 & 0.264 &  & 2 & 157.549 & \multicolumn{1}{c}{---} & 0.070 & 0.255 & 0.409 \\
3 &  &  & \multicolumn{1}{c}{---} & 0.202 & 0.146 &  & 3 & 42.145 & 7.060 & \multicolumn{1}{c}{---} & 0.617 & 0.735 \\
4 &  &  &  & \multicolumn{1}{c}{---} & 0.090 &  & 4 & 20.743 & 4.065 & 1.789 & \multicolumn{1}{c}{---} & 0.850 \\
5 &  &  &  &  & \multicolumn{1}{c}{---} &  & 5 & 13.799 & 2.889 & 1.277 & 0.797 & \multicolumn{1}{c}{---} \\ \cline{1-6} \cline{8-13} 
\end{tabular}
\caption{[Simulation 2] Left panel: Rejection rates of the individual Wald test in Equation (\ref{eq:diag_wald_ind}) for each pair of weight function hyperparameter sets. Right panel: Individual Wald statistics (bottom left triangle) and the corresponding $p$-values (top right triangle) under a representative simulation replication.}
\label{table:sim2_matrix}
\end{table}

\subsection{Simulation 3: GLM with varying dispersion} \label{sec:sim3}
This study analyzes the case when the data generating model systematically deviates from the GLM. We generate $B=500$ replications of simulated samples $\{(y_i,\bm{x}_i)\}_{i=1,\ldots,n}$ with $n=5000$, and $\bm{x}_i$ is generated by the same scheme as the previous studies. $Y_i|\bm{x}_i$ is simulated by a Gamma GLM with varying dispersion. Its density is given by Equation (\ref{eq:glm}) except that the dispersion parameter $\phi$ also depends on $\bm{x}_i$, linked by $\phi=\exp\{\bm{x}_i^T\bm{\alpha}\}$, where $\bm{\alpha}$ represents the regression parameters. We choose a log-link as the mean function, and set $\bm{\beta}=(1,0.5)^T$. We consider the following two cases for $\bm{\alpha}$:
\begin{itemize}
\item Case I: We choose $\bm{\alpha}=(\log 0.5,0)$ so that the dispersion index $\phi=0.5$ does not depend on $\bm{x}_i$. In this case, the simulated model is reduced to the standard Gamma GLM and is exactly the same as that considered in Simulation study 1.
\item Case II: We choose $\bm{\alpha}=(\log 0.5, 0.25)$ so that the covariate $x_{i2}$ has a substantial positive impact on the dispersion of the Gamma distribution.
\end{itemize}

We consider the effect of left truncation and right censoring. 
The left truncation point $T_i$ is sampled as $T_i=0$ or $0.5$ with equal probability, so that any loss $Y_i$ smaller than $T_i$ will not be observed. The right censoring point $C_i$ is sampled as $C_i=10$ or $20$ with equal probability, so that any loss $Y_i$ greater than $C_i$ will be observed as $C_i$ only. In accordance to Section \ref{sec:censtrun}, we have $M_i=1$, the truncation interval $\mathcal{T}_i=(T_i,\infty)$, uncensored interval $\mathcal{U}_i=(T_i,C_i]$ and censored interval $\mathcal{C}_i=\mathcal{I}_{i1}=(C_i,\infty)$.

The resulting modified observed data (for each replication) will be fitted to a Gamma GLM using the proposed SWLE approach with the extended SWLE score function given by Equation (\ref{eq:censtrun:swle}). Therefore, Case I corresponds to a correct model specification, while the model is systematically misspecified in Case II. In accordance to the previous studies, we also consider $K=5$ weight function hyperparameter sets, with the indices $k=1,2,3,4,5$ corresponding to $\delta=1,0.5,0.1,0.01,0.001$ respectively.

Table \ref{table:sim3_param} exhibits the estimated parameters and SE across the $K=5$ hyperparameter sets and two cases for the data generating models. As expected, the estimated parameters are very close to the true parameters in Case I (correctly specified model). In Case II, the regression parameter $\hat{\beta}_2$ decreases most substantially as $k$ increases among all three parameters $(\hat{\beta}_1,\hat{\beta}_2,\hat{\phi})$. As $k$ increases, weights are more concentrated on smaller claims, and hence $\hat{\beta}_2$ better reflects the influence of the covariate $x_{i2}$ on the smaller claims. Note that in Case II, $x_{i2}$ has a positive effect on the dispersion of the distribution. Higher dispersion would decrease the lower quartile yet increase the upper quartile of the distribution. Hence, as $k$ is large, it would bring a negative impact of $x_{i2}$ on the claims, offsetting the positive effect of $x_{i2}$ on the mean function. This explains why $\hat{\beta}_2$ decreases as $k$ increases.

We further perform individual Wald tests in both cases, and the two rejection rate matrices (at 5\% significance level) are presented in Table \ref{table:sim3_matrix}. As expected, the rejection rates in Case I are close to 0.05 (the significance level). Therefore, the asymptotic theories in Theorems \ref{thm:censtrun:asymp} and \ref{thm:censtrun:diag} empirically work well for censored and truncated data. In Case II, the rejection rates are close to 1 for all pairs of $(k,k')$. This reflects a systematic deviation of the data-generating model in both body and tail distributional parts from the Gamma GLM. The above results are different from those in Simulation study 2, where the model misspecification is reflected on the few outliers only so that the rejection rate can be greatly reduced as we choose $k,k'\geq 2$.

\begin{table}[h]
\centering
\begin{tabular}{cllllllllllll}
\hline
 & \multicolumn{6}{c}{Case I   (Correct)} & \multicolumn{6}{c}{Case II   (Misspecified)} \\ \cline{2-13} 
 & \multicolumn{2}{c}{$\hat{\beta}_1$} & \multicolumn{2}{c}{$\hat{\beta}_2$} & \multicolumn{2}{c}{$\hat{\phi}$} & \multicolumn{2}{c}{$\hat{\beta}_1$} & \multicolumn{2}{c}{$\hat{\beta}_2$} & \multicolumn{2}{c}{$\hat{\phi}$} \\ \hline
$k$ & \multicolumn{1}{c}{Mean} & \multicolumn{1}{c}{SE} & \multicolumn{1}{c}{Mean} & \multicolumn{1}{c}{SE} & \multicolumn{1}{c}{Mean} & \multicolumn{1}{c}{SE} & \multicolumn{1}{c}{Mean} & \multicolumn{1}{c}{SE} & \multicolumn{1}{c}{Mean} & \multicolumn{1}{c}{SE} & \multicolumn{1}{c}{Mean} & \multicolumn{1}{c}{SE} \\ \hline
1 & 1.000 & 0.010 & 0.500 & 0.010 & 0.500 & 0.011 & 0.989 & 0.011 & 0.512 & 0.010 & 0.520 & 0.012 \\
2 & 1.000 & 0.011 & 0.500 & 0.011 & 0.500 & 0.011 & 0.986 & 0.011 & 0.501 & 0.010 & 0.515 & 0.012 \\
3 & 1.000 & 0.011 & 0.500 & 0.012 & 0.500 & 0.012 & 0.975 & 0.011 & 0.477 & 0.011 & 0.506 & 0.012 \\
4 & 1.000 & 0.014 & 0.500 & 0.015 & 0.500 & 0.013 & 0.959 & 0.013 & 0.447 & 0.013 & 0.496 & 0.013 \\
5 & 1.000 & 0.019 & 0.500 & 0.019 & 0.500 & 0.015 & 0.942 & 0.016 & 0.421 & 0.016 & 0.489 & 0.015 \\ \hline
\end{tabular}
\caption{[Simulation 3] Mean estimated parameters and standard errors under SWLE across five different weight function hyperparameter sets and two data generating models.}
\label{table:sim3_param}
\end{table}

\begin{table}[h]
\centering
\begin{tabular}{cllllllclllll}
\cline{1-6} \cline{8-13}
\multicolumn{6}{c}{Individual Wald statistic rejection rate} & \multicolumn{1}{c}{} & \multicolumn{6}{c}{Individual Wald statistic rejection rate} \\
\multicolumn{6}{c}{Case I   (Correct)} & \multicolumn{1}{c}{} & \multicolumn{6}{c}{Case II (Misspecified)} \\ \cline{1-6} \cline{8-13} 
$k~\backslash ~k'$ & \multicolumn{1}{c}{1} & \multicolumn{1}{c}{2} & \multicolumn{1}{c}{3} & \multicolumn{1}{c}{4} & \multicolumn{1}{c}{5} & \multicolumn{1}{c}{} & $k~\backslash ~k'$ & \multicolumn{1}{c}{1} & \multicolumn{1}{c}{2} & \multicolumn{1}{c}{3} & \multicolumn{1}{c}{4} & \multicolumn{1}{c}{5} \\ \cline{1-6} \cline{8-13} 
1 & \multicolumn{1}{c}{---} & 0.056 & 0.044 & 0.054 & 0.056 &  & 1 & \multicolumn{1}{c}{---} & 0.994 & 0.998 & 1.000 & 1.000 \\
2 &  & \multicolumn{1}{c}{---} & 0.052 & 0.058 & 0.056 &  & 2 &  & \multicolumn{1}{c}{---} & 1.000 & 1.000 & 1.000 \\
3 &  &  & \multicolumn{1}{c}{---} & 0.046 & 0.038 &  & 3 &  &  & \multicolumn{1}{c}{---} & 1.000 & 1.000 \\
4 &  &  &  & \multicolumn{1}{c}{---} & 0.048 &  & 4 &  &  &  & \multicolumn{1}{c}{---} & 1.000 \\
5 &  &  &  &  & \multicolumn{1}{c}{---} &  & 5 &  &  &  &  & \multicolumn{1}{c}{---} \\ \cline{1-6} \cline{8-13} 
\end{tabular}
\caption{[Simulation 3] Rejection rate matrices (at 5\% significance level) of the individual Wald test for correctly specified (Case I, left panel) and misspecified (Case II, right panel) data generating models.}
\label{table:sim3_matrix}
\end{table}

\section{Real insurance data analysis} \label{sec:data}
This section showcases the applications of the proposed SWLE fitting and diagnostic methods to two real insurance datasets: US indemnity losses and the European automobile insurance dataset. The data-generating model is unknown for a real dataset, and the data is also censored and truncated. Hence, it is challenging to determine the weight function hyperparameters by computing and solving Equation (\ref{eq:sim_wgt_ratio}) directly. We propose approximating Equation (\ref{eq:sim_wgt_ratio}) semi-analytically as follows. To begin with, defining $\mathcal{Q}_{\alpha}=(q_{\alpha},\infty)\cap\mathcal{Y}$, the numerator and denominator in the left-hand side of the equation can respectively be approximated and analytically expressed as (we refer readers to Section \ref{apx:thm:exp_w_data} of the Appendix for the derivations):
\begin{align} \label{eq:data_wgt_ratio_num}
E_{\mathcal{D},\bm{x}}[W(Y,\bm{x})|Y>q_{\alpha}]
=E_{\mathcal{D},\bm{x}}\left[\lambda^{*}(\bm{\Psi};\bm{x})\frac{F^*(\mathcal{Q}_{\alpha}\cap\mathcal{T};\bm{x},\bm{\Psi})}{F(\mathcal{Q}_{\alpha}\cap\mathcal{T};\bm{x},\bm{\Psi})}\right],
\end{align}
\begin{align} \label{eq:data_wgt_ratio_den}
E_{\mathcal{D},\bm{x}}[W(Y,\bm{x})]
=E_{\mathcal{D},\bm{x}}\left[\lambda^{*}(\bm{\Psi};\bm{x})\frac{F^*(\mathcal{T};\bm{x},\bm{\Psi})}{F(\mathcal{T};\bm{x},\bm{\Psi})}\right].
\end{align}

After obtaining an MLE of parameters $\hat{\bm{\Psi}}$, Equations (\ref{eq:data_wgt_ratio_num}) and (\ref{eq:data_wgt_ratio_den}) are approximated by
\begin{align}
\widehat{E_{\mathcal{D},\bm{x}}}[W(Y,\bm{x})|Y>q_{\alpha}]
=\frac{1}{n}\sum_{i=1}^{n}\lambda^{*}(\hat{\bm{\Psi}};\bm{x}_i)\frac{F^*(\mathcal{Q}_{\alpha}\cap\mathcal{T}_i;\bm{x}_i,\hat{\bm{\Psi}})}{F(\mathcal{Q}_{\alpha}\cap\mathcal{T}_i;\bm{x}_i,\hat{\bm{\Psi}})},
\end{align}
\begin{align}
\widehat{E_{\mathcal{D},\bm{x}}}[W(Y,\bm{x})]
=\frac{1}{n}\sum_{i=1}^{n}\lambda^{*}(\hat{\bm{\Psi}};\bm{x}_i)\frac{F^*(\mathcal{T}_i;\bm{x}_i,\hat{\bm{\Psi}})}{F(\mathcal{T}_i;\bm{x}_i,\hat{\bm{\Psi}})}.
\end{align}

Finally, we solve the following equation, which serves as a semi-analytical approximation to Equation (\ref{eq:sim_wgt_ratio}), to determine the appropriate weight function hyperparameters:
\begin{align} \label{eq:data:wgt_ratio}
\frac{\widehat{E_{\mathcal{D},\bm{x}}}[W(Y,\bm{x})|Y>q_{\alpha}]}{\widehat{E_{\mathcal{D},\bm{x}}}[W(Y,\bm{x})]}=\delta.
\end{align}

\subsection{US indemnity losses}
Consider a publicly available dataset of $n=1500$ indemnity losses in the US (\texttt{copula} package in \texttt{R}) extensively studied in several actuarial papers, including \cite{punzo2018compound} and \cite{poudyal2021robust}. Among the 1500 losses, 1352 of them are each accompanied by a maximum benefit ranging from US\$5,000 to US\$7,500,000. The indemnity losses exceeding their maximum benefits will only be recorded as their maximum benefits, so the loss dataset is right-censored. There is no deductible, so all losses are observed (no truncation effects). Each loss accompanies no explanatory variables. As a preliminary analysis, we first present in Figure \ref{fig:real1_prelim} the density plot and normal Q-Q plot for the log-transformed indemnity losses, ignoring the censoring effect. The log-normal model decently captures the indemnity losses, except for a few outliers on the left (minor losses). These outliers may distort the estimated parameters of the log-normal model, which is unreasonable because such immaterial losses should bring little or no impact on the portfolio risk characteristics. 

\begin{figure}[!h]
\begin{center}
\includegraphics[width=0.8\linewidth]{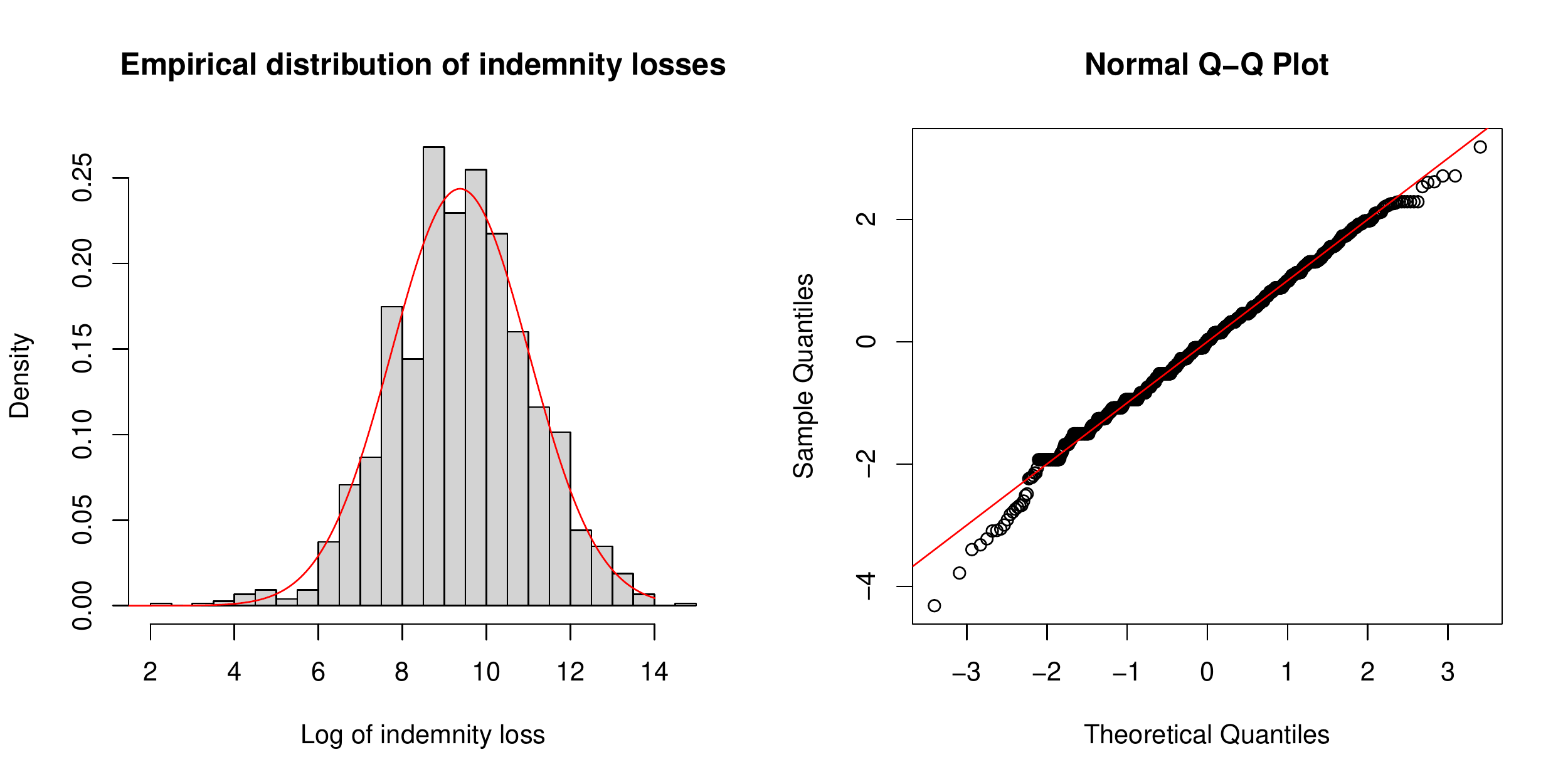}
\end{center}
\vspace{-0.5cm}
\caption{Density plot (left panel) and normal Q-Q plot (right panel) of the log indemnity losses.}
\label{fig:real1_prelim}
\end{figure}

We first define $Y_i$ as the log-transformed indemnity loss. We then construct the censoring-truncation mechanism following Section \ref{sec:censtrun}: The right censoring point $C_i$ is given by the log-transformed maximum benefit. For the remaining 148 losses with no limits, we have $C_i=\infty$. Therefore, the random uncensoring and censoring intervals are given by $\mathcal{U}_i=(-\infty,C_i]$ and $\mathcal{C}_i=\mathcal{I}_{iM_i}=(C_i,\infty)$ with $M_i=1$. Since the data is untruncated, the truncation interval is $\mathcal{T}_i=(-\infty,\infty)$.

To perform robust modelling and diagnostic analysis on the indemnity loss using the log-normal model, the transformed loss data is fitted to a linear model with $P=1$ (no covariates) by solving the extended SWLE score function given by Equation (\ref{eq:censtrun:swle}). We consider $K=5$ weight function hyperparameter sets, indexed by $k=1,2,3,4,5$ and determined by solving Equation (\ref{eq:data:wgt_ratio}) with $\alpha=0.99$ and $\delta=1,0.5,0.1,0.01,0.001$.

The left panel of Table \ref{table:real1_param} depicts the fitted parameters $\hat{\beta}_1$ and $\hat{\phi}$ (equivalent to the mean $\hat{\mu}$ and variance $\hat{\sigma^2}$ parameters of the normal distribution) with the corresponding SE. While the estimated mean parameter $\hat{\mu}$ is stable across $k=1$ to $k=5$, the estimated variance $\hat{\sigma^2}$ drops substantially (from 2.777 to 2.691) as the outliers are down-weighted (i.e., from $k=1$ to $k=2$). Such a drop is significant enough for both individual (right panel of Table \ref{table:real1_param}) and individual parameter-specific (right panel of Table \ref{table:real1_matrix}) Wald tests to reject the log-normal model, choosing $(k,k')=(1,2)$ and a significance level of 5\%. On the other hand, after sufficiently down-weighting the outliers, both estimated parameters $(\hat{\mu},\hat{\sigma^2})$ are stable across various chosen sets of weight function hyperparameters (from $k=2$ to $k=5$). The individual Wald tests fail to reject the log-normal model when we choose $k,k'\geq 2$.

Therefore, the following conclusions are made: The log-normal distribution is overall a suitable model for fitting the indemnity loss data because the individual Wald tests accept the log-normal model in most cases. However, without weighting the observations, the fitted MLE dispersion (variance) parameter (by choosing $k=1$) can be substantially distorted by a few outliers on the left tail.
To mitigate this robustness issue, we should consider the proposed SWLE to fit the indemnity loss data because the estimated parameters become reliable after sufficiently de-emphasizing the outliers' influences. The choice of $k=2$ seems quite desirable as robust estimated parameters are obtained with high efficiencies (the SE of estimated parameters under $k=2$ are very close to those under the MLE with $k=1$, see the left panel of Table \ref{table:real1_param}). The analysis and conclusion here are pretty similar to those obtained in Section \ref{sec:sim2}, suggesting that the indemnity losses may follow a contaminated log-normal distribution.

We also fit the indemnity losses to Gamma and Inverse-Gaussian models using the SWLE approach. The resulting Wald tests strongly reject both models. Therefore, neither of these distributions are appropriate. We refrain from bogging down all details for conciseness purpose.

\begin{table}[h]
\centering
\begin{tabular}{clllllclllll}
\cline{1-5} \cline{7-12}
 & \multicolumn{2}{c}{$\hat{\beta}_1=\hat{\mu}$} & \multicolumn{2}{c}{$\hat{\phi}=\hat{\sigma^2}$} & \multicolumn{1}{c}{} & \multicolumn{6}{c}{Individual Wald statistic \& p-value} \\ \cline{1-5} \cline{7-12} 
$k$ & \multicolumn{1}{c}{Fitted} & \multicolumn{1}{c}{SE} & \multicolumn{1}{c}{Fitted} & \multicolumn{1}{c}{SE} & \multicolumn{1}{c}{} & $k~\backslash ~k'$ & \multicolumn{1}{c}{1} & \multicolumn{1}{c}{2} & \multicolumn{1}{c}{3} & \multicolumn{1}{c}{4} & \multicolumn{1}{c}{5} \\ \cline{1-5} \cline{7-12} 
1 & 9.392 & 0.043 & 2.777 & 0.103 &  & 1 & \multicolumn{1}{c}{---} & 0.018 & 0.255 & 0.591 & 0.776 \\
2 & 9.396 & 0.044 & 2.691 & 0.108 &  & 2 & 8.082 & \multicolumn{1}{c}{---} & 0.442 & 0.638 & 0.809 \\
3 & 9.380 & 0.050 & 2.641 & 0.136 &  & 3 & 2.734 & 1.631 & \multicolumn{1}{c}{---} & 0.705 & 0.899 \\
4 & 9.365 & 0.060 & 2.653 & 0.192 &  & 4 & 1.052 & 0.898 & 0.699 & \multicolumn{1}{c}{---} & 0.998 \\
5 & 9.365 & 0.071 & 2.659 & 0.257 &  & 5 & 0.508 & 0.423 & 0.214 & 0.005 & \multicolumn{1}{c}{---} \\ \cline{1-5} \cline{7-12} 
\end{tabular}
\caption{[US indemnity losses] Left panel: Fitted log-normal parameters and standard errors. Right panel: Individual Wald statistics (bottom left triangle) and the corresponding $p$-values (top right triangle).}
\label{table:real1_param}
\end{table}

\begin{table}[h]
\centering
\begin{tabular}{cllllllclllll}
\cline{1-6} \cline{8-13}
\multicolumn{6}{c}{Individual Wald statistic for $\hat{\beta}_1$ \& p-value} & \multicolumn{1}{c}{} & \multicolumn{6}{c}{Individual Wald statistic for $\hat{\phi}$ \& p-value} \\ \cline{1-6} \cline{8-13} 
$k~\backslash ~k'$ & \multicolumn{1}{c}{1} & \multicolumn{1}{c}{2} & \multicolumn{1}{c}{3} & \multicolumn{1}{c}{4} & \multicolumn{1}{c}{5} & \multicolumn{1}{c}{} & $k~\backslash ~k'$ & \multicolumn{1}{c}{1} & \multicolumn{1}{c}{2} & \multicolumn{1}{c}{3} & \multicolumn{1}{c}{4} & \multicolumn{1}{c}{5} \\ \cline{1-6} \cline{8-13} 
1 & \multicolumn{1}{c}{---} & 0.736 & 0.625 & 0.507 & 0.617 &  & 1 & \multicolumn{1}{c}{---} & 0.004 & 0.118 & 0.440 & 0.614 \\
2 & 0.113 & \multicolumn{1}{c}{---} & 0.348 & 0.366 & 0.527 &  & 2 & 8.075 & \multicolumn{1}{c}{---} & 0.397 & 0.780 & 0.880 \\
3 & 0.239 & 0.879 & \multicolumn{1}{c}{---} & 0.410 & 0.653 &  & 3 & 2.438 & 0.717 & \multicolumn{1}{c}{---} & 0.884 & 0.913 \\
4 & 0.441 & 0.816 & 0.678 & \multicolumn{1}{c}{---} & 0.993 &  & 4 & 0.596 & 0.078 & 0.021 & \multicolumn{1}{c}{---} & 0.945 \\
5 & 0.250 & 0.400 & 0.202 & 0.000 & \multicolumn{1}{c}{---} &  & 5 & 0.255 & 0.023 & 0.012 & 0.005 & \multicolumn{1}{c}{---} \\ \cline{1-6} \cline{8-13} 
\end{tabular}
\caption{[US indemnity losses] Rejection rate matrices (at 5\% significance level) of the individual parameter-specific Wald test for $\hat{\beta}_1$ and $\hat{\phi}$ respectively.}
\label{table:real1_matrix}
\end{table}

\subsection{European automobile insurance data} \label{sec:data:euro}
Consider a European automobile insurance dataset with $n=10,032$ car damage claim losses during 2016. This dataset is also analyzed by \cite{FUNG2020MoECensTrun}. The empirical (log-transformed) loss distribution is depicted by Figure \ref{fig:real2_disn}. Each claim is supplemented by policyholder information denoted by $x_{i2}$ to $x_{i11}$ described in Table \ref{table:real2_cov} of the Appendix, a policy limit (right censoring point) ranging from 900 to 183,610 Euros, and a deductible (left truncation point) ranging from 0 to 1,000 Euros. As a preliminary analysis, we fit the loss amounts to the Gamma, log-normal, and inverse Gaussian distributions without considering the effects of the covariates, deductibles, and policy limits. The goodness-of-fit is assessed by the three Q-Q plots in Figure \ref{fig:real2_qq} of the Appendix. Both Gamma and inverse Gaussian distributions fit the loss data poorly. The log-normal distribution decently fits the body part of the empirical distribution, but the right tail seems slightly under-extrapolated. Note that with the inclusion of the covariates' effects and consideration of data incompleteness, it may be possible to improve the goodness-of-fit for the right tail. Hence, it is reasonable to consider a log-normal GLM as a baseline benchmark model using the proposed SWLE approach to assess the appropriateness of the log-normal model and recommend suitable model improvements. The MLE parameters and the corresponding SE are listed in the second and third columns of Table \ref{table:real2_param}.

\begin{figure}[!h]
\begin{center}
\includegraphics[width=0.9\linewidth]{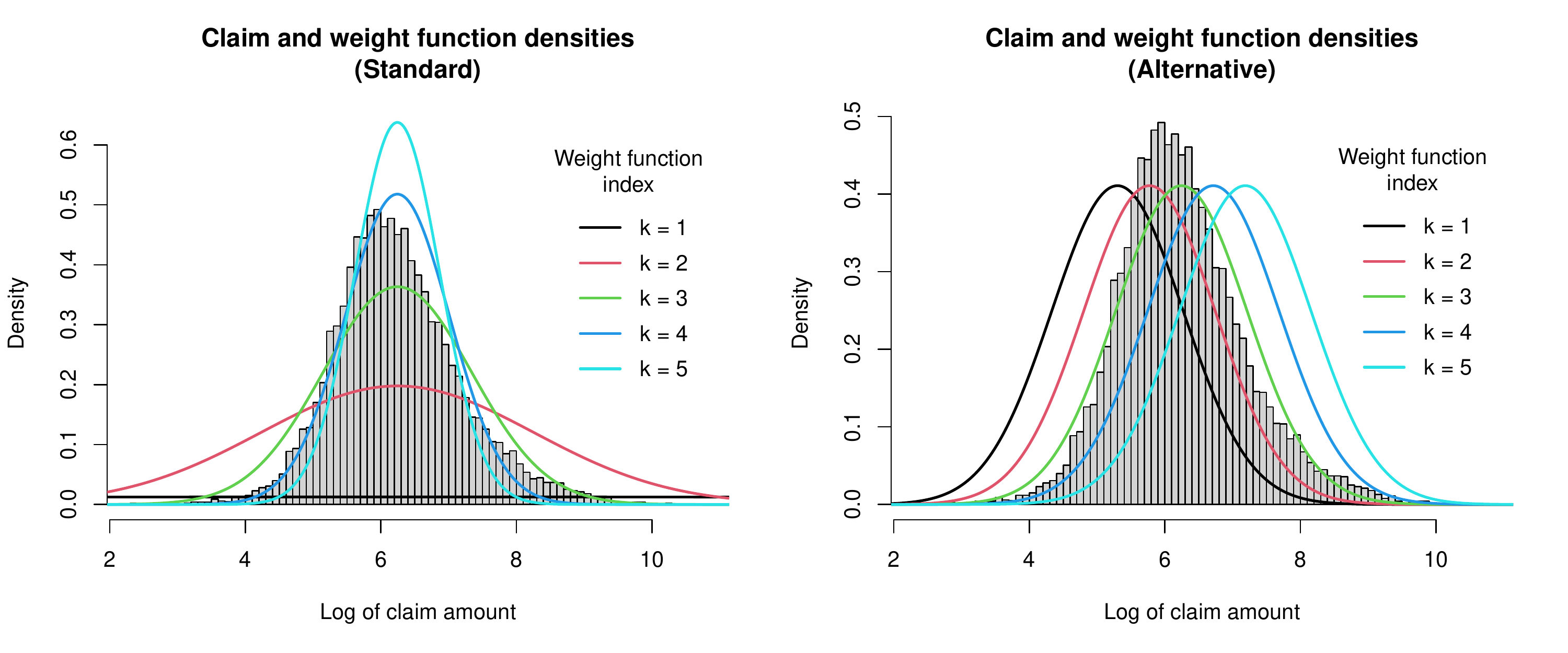}
\end{center}
\vspace{-0.5cm}
\caption{[European automobile claims] Empirical distribution of the log-transformed European car damages claim amount and weight function densities with standard (left panel) and alternative (right panel) hyperparameter settings.}
\label{fig:real2_disn}
\end{figure}

Considering the log-normal regression model, we first define $Y_i$ as the log-transformed car damage loss and $\bm{x}_i$ as the covariate vector with length $P=11$. Setting $C_i$ and $T_i$ respectively as the log-transformed policy limit and deductible, we have $M_i=1$, $\mathcal{U}_i=(-\infty,C_i]$, $\mathcal{C}_i=\mathcal{I}_{iM_i}=(C_i,\infty)$, and $\mathcal{T}_i=(T_i,\infty)$ according to Section \ref{sec:censtrun}.

Same as the previous studies, we first consider $K=5$ ``standard" sets of weight function hyperparameters chosen in accordance to Example \ref{eg:thm} by solving Equation (\ref{eq:data:wgt_ratio}) with $\alpha=0.99$ and $\delta=1,0.5,0.1,0.01,0.001$. The five resulting weight functions are plotted against the log-transformed loss in the left panel of Figure \ref{fig:real2_disn}. To thoroughly examine the impacts of the SWLE weight functions on the estimated parameters and draw legitimate conclusions, we introduce a quantity called standardized parameter deviance residual, defined as
\begin{align} \label{eq:data:residual}
\hat{r}_p^{(k)}=\frac{\hat{\Psi}_p^{(k)}-\hat{\Psi}_p^{(k_0)}}{SE(\hat{\Psi}_p^{(k_0)})},\qquad k=1,\ldots,K~~\text{and}~~p=1,\ldots,P+1,
\end{align}
where $\hat{\Psi}_p^{(k)}$ is the estimated $p$-th parameter of $\bm{\Psi}$, using the SWLE approach with the $k$-th set of weight function hyperparameters selected. $k_0\in\{1,\ldots,K\}$ is the weight function index selected as the benchmark, and $SE(\hat{\Psi}_p^{(k_0)})$ is the SE of $\hat{\Psi}_p^{(k_0)}$. In this case, we choose $k_0=1$ because it represents the MLE approach. The proposed residual statistic reflects the sensitivity of the estimated parameters to the choice of weight function hyperparameters. Suppose $\hat{r}_p^{(k)}$ shows a systematic trend (increasing or decreasing) as $k$ increases and $\hat{r}_p^{(k)}$ significantly differs from zero for $k\neq k_0$. In that case, the fitted model class (i.e., linear model) systematically deviates from the empirical dataset, and hence considerations of alternative model classes are necessary.

Figure \ref{fig:real2_param1} plots $\hat{r}_p^{(k)}$ against $k$ for each parameter $p=1,\ldots,12$ with the ``standard" hyperparameter setting. The 95\% confidence intervals of $\hat{r}_p^{(k)}$ are also constructed, appearing as the grey shallows in the figure. Below are some of the observations and recommendations for model improvements:
\begin{itemize}
\item The residuals for the dispersion parameter $\phi$ significantly decrease as $k$ increases from 1 to 5. As the tail observations are more severely down-weighted, the estimated dispersion parameter is reduced. This suggests that the tails implied by the loss dataset are too heavy that the estimated dispersion parameter under the MLE approach is inflated. In other words, the log-normal model still under-extrapolates the tail-heaviness of the empirical distribution even after incorporating the effects of covariates, censoring, and truncation, so one should fit a heavier-tailed model.
\item While the MLE suggests that the car age negatively impacts the loss amounts ($\hat{\beta}_3<0$), Figure \ref{fig:real2_param1} shows that the residuals for $\beta_3$ (the regression coefficient of car age) are significantly positive when $k>1$. As the weights are more centralized to the body part of the loss distribution, the estimated coefficient becomes less negative. This implies that the car age affects the body part of the distribution less negatively than the tail part. Oppositely, the residuals for $\beta_4$, $\beta_6$ and $\beta_8$ are significantly negative as $k>1$, meaning that the influences of these variables are more negative (or less positive) to the body part of the loss distribution than to the tail part. As a result, one should consider modeling the heterogeneity of covariate influence to various parts of the loss distribution.
\end{itemize}

The above ``standard" weight function setting weights the observations symmetrically, i.e., the losses from both tails are under-weighted. Therefore, the study above does not tell whether the model misfit comes from the left or right tail. Hence, one may also consider an ``alternative" weight function setting, which allocates asymmetric weights to the observations from the left and right tails. Considering also $K=5$ sets of hyperparameters, we set the weight function hyperparameters as $\tilde{\beta}_1^{(k)}=\hat{\mu}_Y+0.5(k-3)\hat{\sigma}_Y$, $\tilde{\beta}_p^{(k)}=0$ for $p=2,3,\ldots,P$ and $\tilde{\phi}^{(k)}=\hat{\sigma}_Y^2$, where $\hat{\mu}_Y$ and $\hat{\sigma}_Y$ are respectively the empirical mean and standard deviation of the log-transformed losses, so that the resulting weight functions are plotted against the log-transformed loss in the right panel of Figure \ref{fig:real2_disn}. As $k$ increases, more weights are assigned to larger losses. We choose $k_0=3$ in Equation (\ref{eq:data:residual}) when a maximum weight is assigned to an average loss. Figure \ref{fig:real2_param2} plots $\hat{r}_p^{(k)}$ against $k$ for each parameter using the ``alternative" weight function hyperparameter setting. The observations are as follows:
\begin{itemize}
\item The residual for the dispersion parameter $\phi$ is significantly positive when $k>3$ (i.e., larger weights on the larger losses) and significantly negative when $k<3$. This means that the log-normal model under-estimates the heaviness of the right tail yet over-extrapolates the left tail.
\item The residuals for $\beta_4$ to $\beta_8$ are significantly positive when $k>3$ and negative when $k<3$. This reflects that the influences of these variables are more positive (or less negative) to the upper quartiles of the loss distribution and vice versa to the lower quartiles. This also echos with the results obtained by Simulation study 3, where the fitted regression parameters show a clear trend when more and more weights are assigned to larger losses. Therefore, a possible model improvement is to allow for a varying dispersion in the regression model.
\end{itemize}

For completeness, we conduct the parameter-specific meta Wald tests (Equation (\ref{eq:diag_wald_param_meta})) under both ``standard" and ``alternative" hyperparameter settings. The resulting Wald statistics and p-values for each parameter are displayed in the four rightmost columns of Table \ref{table:real2_param}. Not surprisingly, the p-values are very small for many parameters, and hence the log-normal regression model is strongly rejected. This confirms the necessity of considering the suggested model improvements to the log-normal model.

\begin{figure}[!h]
\begin{center}
\includegraphics[width=\linewidth]{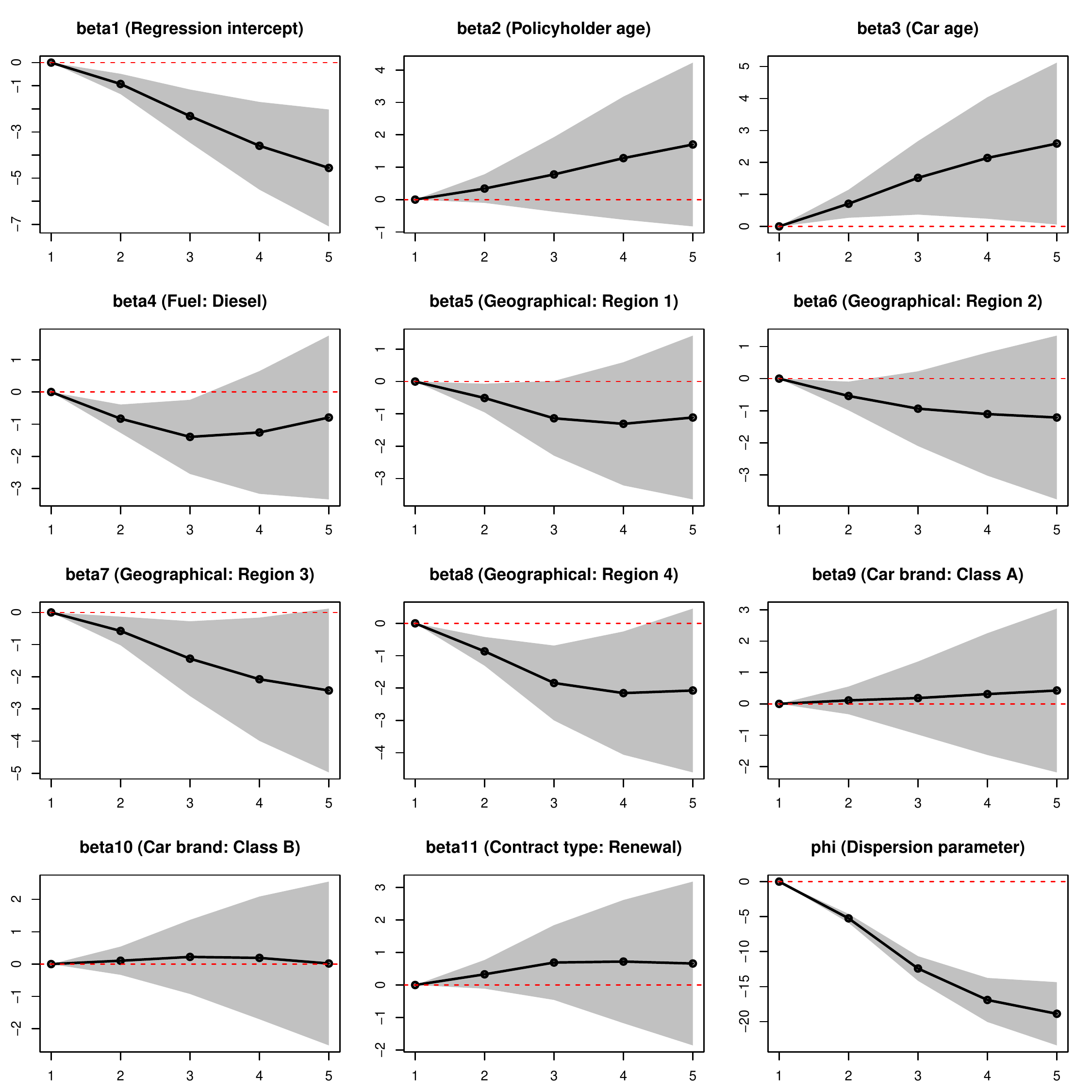}
\end{center}
\vspace{-0.5cm}
\caption{[European automobile claims] Standardized parameter deviance residuals (vertical axis) versus weight function hyperparameter index $k$ (horizontal axis) with standard hyperparameter setting.}
\label{fig:real2_param1}
\end{figure}

\begin{figure}[!h]
\begin{center}
\includegraphics[width=\linewidth]{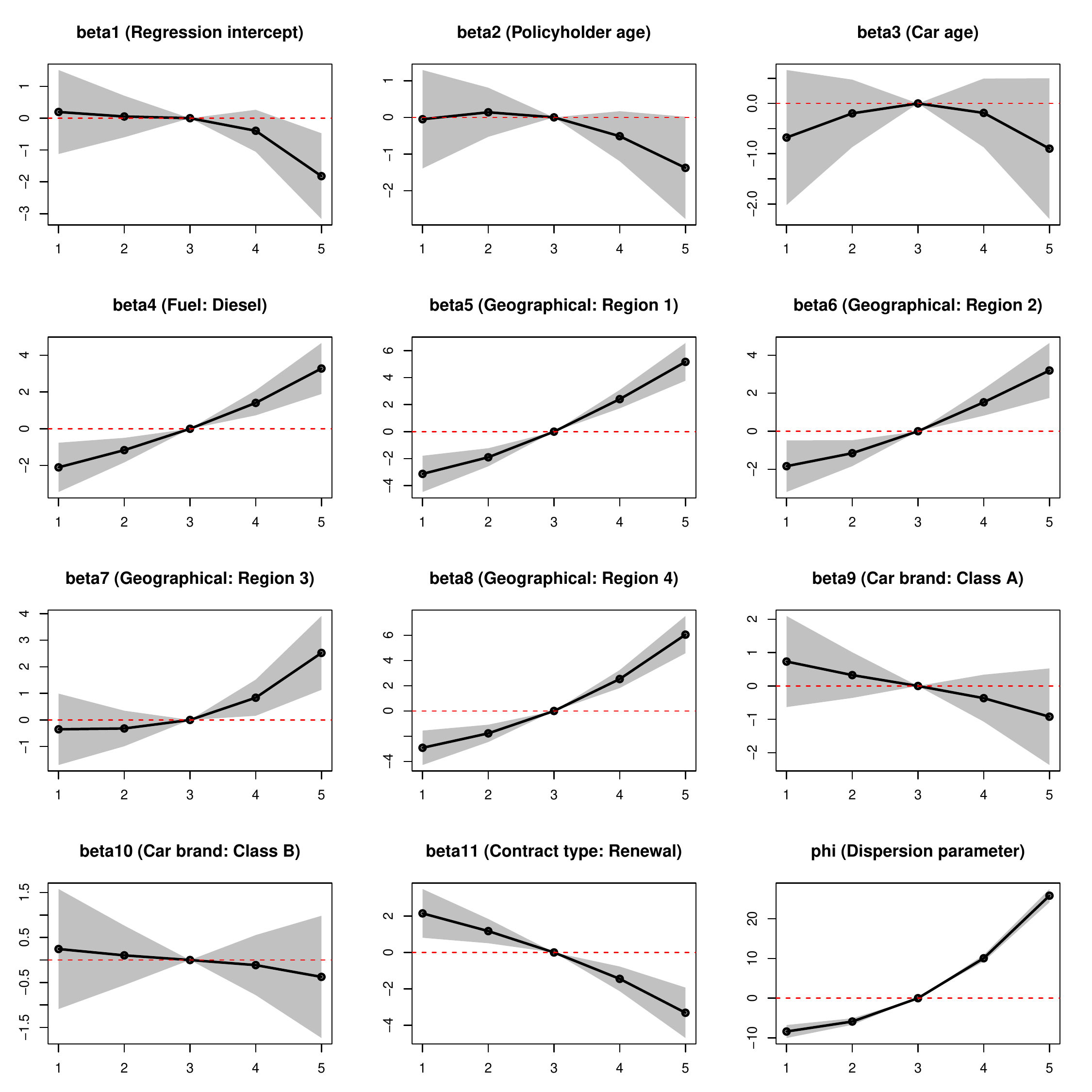}
\end{center}
\vspace{-0.5cm}
\caption{[European automobile claims] Standardized parameter deviance residuals (vertical axis) versus weight function hyperparameter index $k$ (horizontal axis) with alternative hyperparameter setting.}
\label{fig:real2_param2}
\end{figure}

\begin{table}[h]
\centering
\begin{tabular}{crrrrrr}
\hline
 & \multicolumn{2}{c}{Parameter Estimates} & \multicolumn{4}{c}{Parameter specific meta Wald test} \\ \cline{2-7} 
 & \multicolumn{2}{c}{MLE} & \multicolumn{2}{c}{Standard} & \multicolumn{2}{c}{Alternative} \\ \cline{2-7} 
 & \multicolumn{1}{c}{Estimates} & \multicolumn{1}{c}{SE} & \multicolumn{1}{c}{Wald Stat} & \multicolumn{1}{c}{p-value} & \multicolumn{1}{c}{Wald Stat} & \multicolumn{1}{c}{p-value} \\ \hline
$\hat{\beta}_1$ & 6.504 & 0.048 & 19.766 & 0.001 & 113.637 & 0.000 \\
$\hat{\beta}_2$ & -0.003 & 0.001 & 8.041 & 0.090 & 7.026 & 0.135 \\
$\hat{\beta}_3$ & -0.014 & 0.003 & 17.013 & 0.002 & 15.803 & 0.003 \\
$\hat{\beta}_4$ & 0.153 & 0.020 & 39.060 & 0.000 & 69.599 & 0.000 \\
$\hat{\beta}_5$ & 0.040 & 0.032 & 8.238 & 0.083 & 62.907 & 0.000 \\
$\hat{\beta}_6$ & -0.070 & 0.031 & 17.710 & 0.001 & 20.479 & 0.000 \\
$\hat{\beta}_7$ & -0.011 & 0.032 & 6.807 & 0.146 & 64.203 & 0.000 \\
$\hat{\beta}_8$ & -0.070 & 0.035 & 19.679 & 0.001 & 119.188 & 0.000 \\
$\hat{\beta}_9$ & -0.350 & 0.025 & 3.813 & 0.432 & 7.791 & 0.100 \\
$\hat{\beta}_{10}$ & -0.176 & 0.023 & 1.284 & 0.864 & 8.320 & 0.081 \\
$\hat{\beta}_{11}$ & -0.020 & 0.023 & 7.301 & 0.121 & 42.344 & 0.000 \\
$\hat{\phi}$ & 0.902 & 0.013 & 285.230 & 0.000 & 2291.140 & 0.000 \\ \hline
\end{tabular}
\caption{[European automobile claims] Second and third columns: The MLE parameter estimates and the corresponding SE. Rightmost four columns: The Wald statistics and p-values for the parameter specific meta Wald test under standard and alternative settings of weight function hyperparameters.}
\label{table:real2_param}
\end{table}

\section{Concluding remarks} \label{sec:conclude}
This paper introduces a score-based weighted likelihood estimation (SWLE), which incorporates weights to reduce the impact of the outliers, to estimate the parameters of the GLM robustly. With a specially designed weight function, closed-form expressions are obtained for both the score function and asymptotic covariance matrix, making it computationally appealing to estimate parameters and determine parameter uncertainties. The robustness of the SWLE is also theoretically justified by a bounded influence function (IF). Apart from robust estimations, the SWLE also serves as a diagnostic tool to quantitatively assess the overall appropriateness of fitting the GLM. We further extend the SWLE to cater to random censored and truncated regression data prevalent in the insurance losses subjected to coverage modifications. The proposed tool is exemplified on three simulation studies and two real insurance datasets, revealing the usefulness of the SWLE in the following three aspects:
\begin{itemize}
\item If the data-generating model is the GLM (Simulation study 1), the SWLE will provide consistent estimations;
\item If the data-generating model is the GLM contaminated by a few outliers (Simulation study 2 and US indemnity loss data), the SWLE will provide a more reliable estimate of parameters as compared to the MLE;
\item If the data-generating model deviates systematically from the GLM (Simulation study 3 and European automobile insurance data), the SWLE Wald test will detect the model misspecifications with very high power and suggest model improvements.
\end{itemize} 

The SWLE is applicable not only to the GLM but also to other more complex model classes, including the generalized additive models (GAM) for non-linear regression links and finite mixture models for distributional multimodalities. Therefore, it is worthwhile to explore alternative weight functions such that the SWLE is computationally appealing for broader model classes. Another potential research direction is to extend the SWLE to cater to multivariate and longitudinal data. This is useful in insurance practice because insurance companies often contain multiple business lines (multivariate losses), and policyholders often have multiple years of claim history (longitudinal data). Finally, as discussed in Remark \ref{rmk:robust}, the current study does not consider the effects of covariate outliers. To do so, Assumption (i) of Theorem \ref{thm:if} needs to be removed, and hence extra conditions on the weight function $W(y_i,\bm{x}_i)$ will be needed. We leverage this study to the future research direction.

\bibliographystyle{abbrvnat}
\bibliography{reference}

\begin{thebibliography}{25}
\providecommand{\natexlab}[1]{#1}
\providecommand{\url}[1]{\texttt{#1}}
\expandafter\ifx\csname urlstyle\endcsname\relax
  \providecommand{\doi}[1]{doi: #1}\else
  \providecommand{\doi}{doi: \begingroup \urlstyle{rm}\Url}\fi

\bibitem[Aeberhard et~al.(2014)Aeberhard, Cantoni, and
  Heritier]{aeberhard2014robust}
W.~H. Aeberhard, E.~Cantoni, and S.~Heritier.
\newblock Robust inference in the negative binomial regression model with an
  application to falls data.
\newblock \emph{Biometrics}, 70\penalty0 (4):\penalty0 920--931, 2014.

\bibitem[Aeberhard et~al.(2021)Aeberhard, Cantoni, Marra, and
  Radice]{aeberhard2021robust}
W.~H. Aeberhard, E.~Cantoni, G.~Marra, and R.~Radice.
\newblock Robust fitting for generalized additive models for location, scale
  and shape.
\newblock \emph{Statistics and Computing}, 31\penalty0 (1):\penalty0 1--16,
  2021.

\bibitem[Blostein and Miljkovic(2019)]{BLOSTEIN201935}
M.~Blostein and T.~Miljkovic.
\newblock On modeling left-truncated loss data using mixtures of distributions.
\newblock \emph{Insurance: Mathematics and Economics}, 85:\penalty0 35 -- 46,
  2019.

\bibitem[Brazauskas and Serfling(2000)]{brazauskas2000robust}
V.~Brazauskas and R.~Serfling.
\newblock Robust and efficient estimation of the tail index of a
  single-parameter {P}areto distribution.
\newblock \emph{North American Actuarial Journal}, 4\penalty0 (4):\penalty0
  12--27, 2000.

\bibitem[Brazauskas and Serfling(2003)]{brazauskas2003favorable}
V.~Brazauskas and R.~Serfling.
\newblock Favorable estimators for fitting {P}areto models: A study using
  goodness-of-fit measures with actual data.
\newblock \emph{ASTIN Bulletin: The Journal of the IAA}, 33\penalty0
  (2):\penalty0 365--381, 2003.

\bibitem[Cantoni and Ronchetti(2001)]{cantoni2001robust}
E.~Cantoni and E.~Ronchetti.
\newblock Robust inference for generalized linear models.
\newblock \emph{Journal of the American Statistical Association}, 96\penalty0
  (455):\penalty0 1022--1030, 2001.

\bibitem[Cooray and Ananda(2005)]{cooray2005modeling}
K.~Cooray and M.~M. Ananda.
\newblock Modeling actuarial data with a composite lognormal-{P}areto model.
\newblock \emph{Scandinavian Actuarial Journal}, 2005\penalty0 (5):\penalty0
  321--334, 2005.

\bibitem[Delong et~al.(2021)Delong, Lindholm, and
  W{\"u}thrich]{delong2021gamma}
{\L}.~Delong, M.~Lindholm, and M.~V. W{\"u}thrich.
\newblock Gamma mixture density networks and their application to modelling
  insurance claim amounts.
\newblock \emph{Insurance: Mathematics and Economics}, 101:\penalty0 240--261,
  2021.

\bibitem[Embrechts et~al.(1999)Embrechts, Resnick, and
  Samorodnitsky]{embrechts1999extreme}
P.~Embrechts, S.~I. Resnick, and G.~Samorodnitsky.
\newblock Extreme value theory as a risk management tool.
\newblock \emph{North American Actuarial Journal}, 3\penalty0 (2):\penalty0
  30--41, 1999.

\bibitem[Fung(2021)]{fung2021maximum}
T.~C. Fung.
\newblock Maximum weighted likelihood estimator for robust heavy-tail modelling
  of finite mixture models.
\newblock \emph{arXiv preprint arXiv:2108.01356}, 2021.

\bibitem[Fung et~al.(2021)Fung, Tzougas, and Wuthrich]{fung2021mixture}
T.~C. Fung, G.~Tzougas, and M.~Wuthrich.
\newblock Mixture composite regression models with multi-type feature
  selection.
\newblock \emph{arXiv preprint arXiv:2103.07200}, 2021.

\bibitem[Fung et~al.(2022)Fung, Badescu, and Lin]{FUNG2020MoECensTrun}
T.~C. Fung, A.~L. Badescu, and X.~S. Lin.
\newblock Fitting censored and truncated regression data using the mixture of
  experts models.
\newblock \emph{North American Actuarial Journal}, 2022.
\newblock forthcoming.

\bibitem[Ghosh and Basu(2016)]{ghosh2016robust}
A.~Ghosh and A.~Basu.
\newblock Robust estimation in generalized linear models: the density power
  divergence approach.
\newblock \emph{Test}, 25\penalty0 (2):\penalty0 269--290, 2016.

\bibitem[Hampel(1974)]{hampel1974influence}
F.~R. Hampel.
\newblock The influence curve and its role in robust estimation.
\newblock \emph{Journal of the American Statistical Association}, 69\penalty0
  (346):\penalty0 383--393, 1974.

\bibitem[Miljkovic and Grün(2016)]{MILJKOVIC2016387}
T.~Miljkovic and B.~Grün.
\newblock Modeling loss data using mixtures of distributions.
\newblock \emph{Insurance: Mathematics and Economics}, 70:\penalty0 387 -- 396,
  2016.
\newblock ISSN 0167-6687.

\bibitem[Nelder and Wedderburn(1972)]{nelder1972generalized}
J.~A. Nelder and R.~W. Wedderburn.
\newblock Generalized linear models.
\newblock \emph{Journal of the Royal Statistical Society: Series A (General)},
  135\penalty0 (3):\penalty0 370--384, 1972.

\bibitem[Poudyal(2021{\natexlab{a}})]{poudyal2021robust}
C.~Poudyal.
\newblock Robust estimation of loss models for lognormal insurance payment
  severity data.
\newblock \emph{ASTIN Bulletin: The Journal of the IAA}, 51\penalty0
  (2):\penalty0 475--507, 2021{\natexlab{a}}.

\bibitem[Poudyal(2021{\natexlab{b}})]{poudyal2021truncated}
C.~Poudyal.
\newblock Truncated, censored, and actuarial payment-type moments for robust
  fitting of a single-parameter {P}areto distribution.
\newblock \emph{Journal of Computational and Applied Mathematics},
  388:\penalty0 113310, 2021{\natexlab{b}}.

\bibitem[Punzo et~al.(2018)Punzo, Bagnato, and Maruotti]{punzo2018compound}
A.~Punzo, L.~Bagnato, and A.~Maruotti.
\newblock Compound unimodal distributions for insurance losses.
\newblock \emph{Insurance: Mathematics and Economics}, 81:\penalty0 95--107,
  2018.

\bibitem[Serfling(2002)]{serfling2002efficient}
R.~Serfling.
\newblock Efficient and robust fitting of lognormal distributions.
\newblock \emph{North American Actuarial Journal}, 6\penalty0 (4):\penalty0
  95--109, 2002.

\bibitem[Tzougas and Karlis(2020)]{tzougas2020algorithm}
G.~Tzougas and D.~Karlis.
\newblock An em algorithm for fitting a new class of mixed exponential
  regression models with varying dispersion.
\newblock \emph{Astin Bulletin}, 2020.

\bibitem[Valdora and Yohai(2014)]{valdora2014robust}
M.~Valdora and V.~J. Yohai.
\newblock Robust estimators for generalized linear models.
\newblock \emph{Journal of Statistical Planning and Inference}, 146:\penalty0
  31--48, 2014.

\bibitem[Van~der Vaart(2000)]{van2000asymptotic}
A.~W. Van~der Vaart.
\newblock \emph{Asymptotic statistics}, volume~3.
\newblock Cambridge university press, 2000.

\bibitem[Wong et~al.(2014)Wong, Yao, and Lee]{wong2014robust}
R.~K. Wong, F.~Yao, and T.~C. Lee.
\newblock Robust estimation for generalized additive models.
\newblock \emph{Journal of Computational and Graphical Statistics}, 23\penalty0
  (1):\penalty0 270--289, 2014.

\bibitem[Zhao et~al.(2018)Zhao, Brazauskas, and Ghorai]{zhao2018robust}
Q.~Zhao, V.~Brazauskas, and J.~Ghorai.
\newblock Robust and efficient fitting of severity models and the method of
  winsorized moments.
\newblock \emph{ASTIN Bulletin: The Journal of the IAA}, 48\penalty0
  (1):\penalty0 275--309, 2018.

\end{thebibliography}

\pagebreak
\begin{appendices}
\section{Regularity conditions}

\subsection{Regularity conditions for Theorems \ref{thm:asymp} and \ref{thm:if} (complete data)} \label{apx:sec:reg_com_asymp}
Recall that the (individual) SWLE score function for complete data is given by $\mathcal{S}(\bm{\Psi};y,\bm{x})$ in Equation (\ref{eq:thm:score_glm}). We denote $\Omega$ as the parameter space of $\bm{\Psi}$. The regularity conditions for Theorems \ref{thm:asymp} and \ref{thm:if} are:

\begin{enumerate}
\item The functions $A(\cdot)$, $g(\cdot)$, $b(\cdot)$ and $\xi(\cdot)$ are three times continuously differentiable.
\item $E_{Y,\bm{x}}\left[\|\mathcal{S}(\bm{\Psi};Y,\bm{x})\|^2\right]<\infty$ for $\bm{\Psi}\in\Omega$.
\item $E_{Y,\bm{x}}\left[\frac{\partial}{\partial\bm{\Psi}}\mathcal{S}(\bm{\Psi};Y,\bm{x})^T\right]$ exists and is finite for $\bm{\Psi}\in\Omega$.
\item $\left|\frac{\partial^2}{\partial\psi_{p_1}\partial\psi_{p_2}}[\mathcal{S}(\bm{\Psi};Y,\bm{x})]_{p_3}\right|$ is dominated by a fixed integrable function of $(Y,\bm{x})$ for $p_1,p_2,p_3=1,\ldots,P+1$, where $\psi_p$ is the $p$-th element of $\bm{\Psi}$ and $[\mathcal{S}(\bm{\Psi};Y,\bm{x})]_{p_3}$ is the $p_3$-th element of the score function vector.
\end{enumerate}

\subsection{Regularity conditions for Theorems \ref{thm:diag:asymp} and \ref{thm:diag:chisq} (complete data)} \label{apx:sec:reg_com_diag}
We define the meta individual score function $\mathcal{S}^{\text{meta}}(\bm{\Psi};Y,\bm{x})$ as
\begin{align}
\mathcal{S}^{\text{meta}}(\bm{\Psi}^{\text{meta}};Y,\bm{x})=
\begin{pmatrix}
\mathcal{S}^{(1)}(\bm{\Psi}^{(1)};Y,\bm{x})\\
\vdots\\
\mathcal{S}^{(K)}(\bm{\Psi}^{(K)};Y,\bm{x})
\end{pmatrix},
\end{align}
with $\mathcal{S}^{(k)}(\bm{\Psi};Y,\bm{x})$ being the individual score function (Equation (\ref{eq:thm:score_glm})) evaluated with hyperparameters $\tilde{\bm{\Psi}}^{(k)}$. Again, we let $\Omega$ be the common parameter space of $\bm{\Psi}^{(k)}$, $k=1,\ldots,K$. The regularity conditions for Theorems \ref{thm:diag:asymp} and \ref{thm:diag:chisq} are as follows for $k_1,k_2,k_3=1,\ldots,K$:
\begin{enumerate}
\item The functions $A(\cdot)$, $g(\cdot)$, $b(\cdot)$ and $\xi(\cdot)$ are three times continuously differentiable.
\item $E_{Y,\bm{x}}\left[\|\mathcal{S}^{\text{meta}}(\bm{\Psi}^{\text{meta}};Y,\bm{x})\|^2\right]<\infty$ for $\bm{\Psi}^{\text{meta}}\in\Omega^K$.
\item $E_{Y,\bm{x}}\left[\frac{\partial}{\partial\bm{\Psi}^{(k_1)}}\mathcal{S}^{(k_2)}(\bm{\Psi}^{(k_2)};Y,\bm{x})^T\right]$ exists and is finite for $\bm{\Psi}^{(k_1)},\bm{\Psi}^{(k_2)}\in\Omega$.
\item $\left|\frac{\partial^2}{\partial\psi^{(k_1)}_{p_1}\partial\psi^{(k_2)}_{p_2}}[\mathcal{S}^{(k_3)}(\bm{\Psi}^{(k_3)};Y,\bm{x})]_{p_3}\right|$ is dominated by a fixed integrable function of $(Y,\bm{x})$ for $p_1,p_2,p_3=1,\ldots,P+1$, where $\psi_p^{(k)}$ is the $p$-th element of $\bm{\Psi}^{(k)}$ and $[\mathcal{S}^{(k_3)}(\bm{\Psi}^{(k_3)};Y,\bm{x})]_{p_3}$ is the $p_3$-th element of the individual score function vector.
\end{enumerate}

\subsection{Regularity conditions for Theorem \ref{thm:censtrun:asymp} (incomplete data)} \label{apx:sec:reg_incom_asymp}
Recall that the (individual) SWLE score function for incomplete data is given by $\mathcal{S}(\bm{\Psi};\mathcal{D},\bm{x})$ in Equation (\ref{eq:censtrun:swle}). We denote $\Omega$ as the parameter space of $\bm{\Psi}$. Similar to the previous subsections, the regularity conditions for Theorem \ref{thm:censtrun:asymp} are:

\begin{enumerate}
\item The functions $A(\cdot)$, $g(\cdot)$, $b(\cdot)$ and $\xi(\cdot)$ are three times continuously differentiable.
\item $E_{\mathcal{D},\bm{x}}\left[\|\mathcal{S}(\bm{\Psi};\mathcal{D},\bm{x})\|^2\right]<\infty$ for $\bm{\Psi}\in\Omega$.
\item $E_{\mathcal{D},\bm{x}}\left[\frac{\partial}{\partial\bm{\Psi}}\mathcal{S}(\bm{\Psi};\mathcal{D},\bm{x})^T\right]$ exists and is finite for $\bm{\Psi}\in\Omega$.
\item $\left|\frac{\partial^2}{\partial\psi_{p_1}\partial\psi_{p_2}}[\mathcal{S}(\bm{\Psi};\mathcal{D},\bm{x})]_{p_3}\right|$ is dominated by a fixed integrable function of $(\mathcal{D},\bm{x})$ for $p_1,p_2,p_3=1,\ldots,P+1$, where $\psi_p$ is the $p$-th element of $\bm{\Psi}$ and $[\mathcal{S}(\bm{\Psi};\mathcal{D},\bm{x})]_{p_3}$ is the $p_3$-th element of the score function vector.
\end{enumerate}

\subsection{Regularity conditions for Theorem \ref{thm:censtrun:diag} (incomplete data)} \label{apx:sec:reg_incom_diag}
We define the meta individual score function $\mathcal{S}^{\text{meta}}(\bm{\Psi};\mathcal{D},\bm{x})$ for incomplete data as
\begin{align}
\mathcal{S}^{\text{meta}}(\bm{\Psi}^{\text{meta}};\mathcal{D},\bm{x})=
\begin{pmatrix}
\mathcal{S}^{(1)}(\bm{\Psi}^{(1)};\mathcal{D},\bm{x})\\
\vdots\\
\mathcal{S}^{(K)}(\bm{\Psi}^{(K)};\mathcal{D},\bm{x})
\end{pmatrix},
\end{align}
with $\mathcal{S}^{(k)}(\bm{\Psi};\mathcal{D},\bm{x})$ being the individual score function (Equation (\ref{eq:censtrun:swle})) evaluated with hyperparameters $\tilde{\bm{\Psi}}^{(k)}$. Again, we let $\Omega$ be the common parameter space of $\bm{\Psi}^{(k)}$, $k=1,\ldots,K$. The regularity conditions for Theorem \ref{thm:censtrun:diag} are as follows for $k_1,k_2,k_3=1,\ldots,K$:
\begin{enumerate}
\item The functions $A(\cdot)$, $g(\cdot)$, $b(\cdot)$ and $\xi(\cdot)$ are three times continuously differentiable.
\item $E_{\mathcal{D},\bm{x}}\left[\|\mathcal{S}^{\text{meta}}(\bm{\Psi}^{\text{meta}};\mathcal{D},\bm{x})\|^2\right]<\infty$ for $\bm{\Psi}^{\text{meta}}\in\Omega^K$.
\item $E_{\mathcal{D},\bm{x}}\left[\frac{\partial}{\partial\bm{\Psi}^{(k_1)}}\mathcal{S}^{(k_2)}(\bm{\Psi}^{(k_2)};\mathcal{D},\bm{x})^T\right]$ exists and is finite for $\bm{\Psi}^{(k_1)},\bm{\Psi}^{(k_2)}\in\Omega$.
\item $\left|\frac{\partial^2}{\partial\psi^{(k_1)}_{p_1}\partial\psi^{(k_2)}_{p_2}}[\mathcal{S}^{(k_3)}(\bm{\Psi}^{(k_3)};\mathcal{D},\bm{x})]_{p_3}\right|$ is dominated by a fixed integrable function of $(\mathcal{D},\bm{x})$ for $p_1,p_2,p_3=1,\ldots,P+1$, where $\psi_p^{(k)}$ is the $p$-th element of $\bm{\Psi}^{(k)}$ and $[\mathcal{S}^{(k_3)}(\bm{\Psi}^{(k_3)};\mathcal{D},\bm{x})]_{p_3}$ is the $p_3$-th element of the individual score function vector.
\end{enumerate}

\section{Covariance matrix of the SWLE in Theorem \ref{thm:asymp}} \label{apx:sec:cov_asymp}
Theorem \ref{thm:asymp} states that the SWLE $\hat{\bm{\Psi}}_n$ satisfies
\begin{align}
\sqrt{n}(\hat{\bm{\Psi}}_n-\bm{\Psi}_0)\overset{d}{\rightarrow}\mathcal{N}(\bm{0},\bm{\Sigma}),
\end{align}
where $\bm{\Sigma}:=\bm{\Sigma}(\bm{\Psi}_0)=(\Gamma^{-1})\Lambda(\Gamma^{-1})^T$, with $\Gamma$ and $\Lambda$ being $(P+1)\times (P+1)$ matrices given by
\begin{align}
\Gamma:=\Gamma(\bm{\Psi}_0)=
E_{Y,\bm{x}}\left[\frac{\partial}{\partial\bm{\Psi}}\mathcal{S}(\bm{\Psi};Y,\bm{x})^T\right]\Bigg|_{\bm{\Psi}=\bm{\Psi}_0}=
\begin{pmatrix}
\Gamma_{\theta\theta}(\bm{\Psi}_0) & \Gamma_{\theta\phi}(\bm{\Psi}_0) \\
\Gamma_{\phi\theta}(\bm{\Psi}_0) & \Gamma_{\phi\phi}(\bm{\Psi}_0)
\end{pmatrix},
\end{align}
where
\begin{align}
\Gamma_{\theta\theta}(\bm{\Psi})
:=E_{\bm{x}}\left[W_{\theta\theta}(\bm{\Psi},\bm{x})\bm{x}\bm{x}^T\right],
\end{align}
\begin{align}
\Gamma_{\theta\phi}(\bm{\Psi})=\Gamma_{\phi\theta}(\bm{\Psi})^T
:=E_{\bm{x}}\left[W_{\theta\phi}(\bm{\Psi},\bm{x})\bm{x}\right],
\end{align}
\begin{align}
\Gamma_{\phi\phi}(\bm{\Psi})
:=E_{\bm{x}}\left[W_{\phi\phi}(\bm{\Psi},\bm{x})\right],
\end{align}
and hence
\begin{align} \label{apx:eq:w_tt}
W_{\theta\theta}(\bm{\Psi},\bm{x})=-\lambda^{*}(\bm{\Psi};\bm{x})\frac{\phi^{*}}{\phi^2}A''(\theta^*)\left(\xi'(\bm{x}^T\bm{\beta})\right)^2,
\end{align}
\begin{align} \label{apx:eq:w_tp}
W_{\theta\phi}(\bm{\Psi},\bm{x})=-\lambda^{*}(\bm{\Psi};\bm{x})\frac{\phi^{*2}}{\phi^3}A''(\theta^*)\xi'(\bm{x}^T\bm{\beta})\left(\left(c-\frac{1}{\tilde{\phi}}\right)\theta+\frac{\tilde{\theta}}{\tilde{\phi}}\right),
\end{align}
\begin{align} \label{apx:eq:w_pp}
W_{\phi\phi}(\bm{\Psi},\bm{x})=\lambda^{*}(\bm{\Psi};\bm{x})\frac{\phi^{*4}}{\phi^4}\left\{-\frac{1}{\phi^*}A''(\theta^*)\left(\left(c-\frac{1}{\tilde{\phi}}\right)\theta+\frac{\tilde{\theta}}{\tilde{\phi}}\right)^2+\left(\frac{2}{\phi^*}b'(\phi^*)+b''(\phi^*)\right)\right\},
\end{align}
and
\begin{align} 
\Lambda:=\Lambda(\bm{\Psi}_0)=
E_{Y,\bm{x}}\left[\mathcal{S}(\bm{\Psi};Y,\bm{x})\mathcal{S}(\bm{\Psi};Y,\bm{x})^T\right]\Bigg|_{\bm{\Psi}=\bm{\Psi}_0}=
\begin{pmatrix}
\Lambda_{\theta\theta}(\bm{\Psi}_0) & \Lambda_{\theta\phi}(\bm{\Psi}_0) \\
\Lambda_{\phi\theta}(\bm{\Psi}_0) & \Lambda_{\phi\phi}(\bm{\Psi}_0)
\end{pmatrix},
\end{align}
where
\begin{align}
\Lambda_{\theta\theta}(\bm{\Psi})
:=E_{\bm{x}}\left[V_{\theta\theta}(\bm{\Psi},\bm{x})\bm{x}\bm{x}^T\right],
\end{align}
\begin{align}
\Lambda_{\theta\phi}(\bm{\Psi})
=\Lambda_{\phi\theta}(\bm{\Psi})^T
&:=E_{\bm{x}}\left[V_{\theta\phi}(\bm{\Psi},\bm{x})\bm{x}\right],
\end{align}
\begin{align}
\Lambda_{\phi\phi}(\bm{\Psi})
&:=E_{\bm{x}}\left[V_{\phi\phi}(\bm{\Psi},\bm{x})\right],
\end{align}
and hence
\begin{align} \label{apx:eq:v_tt}
V_{\theta\theta}(\bm{\Psi},\bm{x})=\lambda^{**}(\bm{\Psi};\bm{x})B_{\theta\theta}(\bm{\Psi};\bm{x})\frac{\phi^{*2}}{\phi^2}\xi'(\bm{x}^T\bm{\beta})^2,
\end{align}
\begin{align} \label{apx:eq:v_tp}
V_{\theta\phi}(\bm{\Psi},\bm{x})=\lambda^{**}(\bm{\Psi};\bm{x})\frac{\phi^{*3}}{\phi^3}\left[B_{\theta\theta}(\bm{\Psi};\bm{x})\left(\left(c-\frac{1}{\tilde{\phi}}\right)\theta+\frac{\tilde{\theta}}{\tilde{\phi}}\right)+B_{\theta\phi}(\bm{\Psi};\bm{x})\right]\xi'(\bm{x}^T\bm{\beta}),
\end{align}
\begin{align} \label{apx:eq:v_pp}
V_{\phi\phi}(\bm{\Psi},\bm{x})=\lambda^{**}(\bm{\Psi};\bm{x})\frac{\phi^{*4}}{\phi^4}\left[B_{\theta\theta}(\bm{\Psi};\bm{x})\left(\left(c-\frac{1}{\tilde{\phi}}\right)\theta+\frac{\tilde{\theta}}{\tilde{\phi}}\right)^2+2B_{\theta\phi}(\bm{\Psi};\bm{x})\left(\left(c-\frac{1}{\tilde{\phi}}\right)\theta+\frac{\tilde{\theta}}{\tilde{\phi}}\right)+B_{\phi\phi}(\bm{\Psi};\bm{x})\right],
\end{align}
with
\begin{align}
B_{\theta\theta}(\bm{\Psi};\bm{x})
=\frac{1}{\phi^{*2}}\left[\phi^{**}A''(\theta^{**})+\left(A'(\theta^{**})-A'(\theta^{*})\right)^2\right],
\end{align}
\begin{align}
B_{\theta\phi}(\bm{\Psi};\bm{x})
=-\frac{1}{\phi^{*3}}\Big\{(\theta^*-\theta^{**})\phi^{**}A''(\theta^{**})
-\left(A'(\theta^*)-A'(\theta^{**})\right)\big[
&(\theta^*-\theta^{**})A'(\theta^{**})-\left(A(\theta^*)-A(\theta^{**})\right) \nonumber\\
&~-\left(\phi^{*2}b'(\phi^*)-\phi^{**2}b'(\phi^{**})\right)\big]\Big\},
\end{align}
\begin{align}
B_{\phi\phi}(\bm{\Psi};\bm{x})
&=-\frac{\phi^{**4}}{\phi^{*4}}\left(\frac{2}{\phi^{**}}b'(\phi^{**})+b''(\phi^{**})\right)\nonumber\\
&\hspace{1em}+\frac{1}{\phi^{*4}}\Big\{(\theta^*-\theta^{**})^2\phi^{**}A''(\theta^{**})
+\Big[(\theta^*-\theta^{**})A'(\theta^{**})-\left(A(\theta^*)-A(\theta^{**})\right) \nonumber\\
&\hspace{18em}-\left(\phi^{*2}b'(\phi^*)-\phi^{**2}b'(\phi^{**})\right)\Big]^2\Big\}.
\end{align}
Here, we have $\theta=\theta(\bm{x},\bm{\beta})=\xi(\bm{x}^T\bm{\beta})$, $\phi^*=(\phi^{-1}+\tilde{\phi}^{-1}-c)^{-1}$, $\theta^*=(\theta/\phi+\tilde{\theta}/\tilde{\phi})\phi^*$, $\phi^{**}=(\phi^{-1}+2\tilde{\phi}^{-1}-2c)^{-1}$ and $\theta^{**}=(\theta/\phi+2\tilde{\theta}/\tilde{\phi})\phi^{**}$. Also, $\lambda^{*}(\bm{\Psi};\bm{x})$ and $\lambda^{**}(\bm{\Psi};\bm{x})$ are the bias adjustment terms defined as $\lambda^*(\bm{\Psi};\bm{x})=\exp\left\{{A(\theta^*)}/{\phi^*}-b(\phi^*)-{A(\theta)}/{\phi}+b(\phi)\right\}$ and $\lambda^{**}(\bm{\Psi};\bm{x})=\exp\left\{{A(\theta^{**})}/{\phi^{**}}-b(\phi^{**})-{A(\theta)}/{\phi}+b(\phi)\right\}$.

\section{Covariance meta matrix for the SWLE diagnostic test in Theorem \ref{thm:diag:asymp}} \label{apx:sec:cov_diag}
Recall Theorem \ref{thm:diag:asymp} that the estimated parameters $\hat{\bm{\Psi}}_n^{\text{meta}}$ satisfy
\begin{align}
\sqrt{n}\left(\hat{\bm{\Psi}}_n^{\text{meta}}-\bm{\Psi}_0^{\text{meta}}\right)\overset{d}{\rightarrow}\mathcal{N}(\bm{0},\bm{\Sigma}^{\text{meta}}),
\end{align}
where the covariance meta matrix is given by
\begin{align} \label{apx:eq:diag:sigma}
{\tiny
\Sigma^{\text{meta}}=
\begin{pmatrix}
\left(\left[\Gamma^{(1)}\right]^{-1}\right)\Lambda^{(1,1)}\left(\left[\Gamma^{(1)}\right]^{-1}\right)^T & \left(\left[\Gamma^{(1)}\right]^{-1}\right)\Lambda^{(1,2)}\left(\left[\Gamma^{(2)}\right]^{-1}\right)^T & 
\dots & 
\left(\left[\Gamma^{(1)}\right]^{-1}\right)\Lambda^{(1,K)}\left(\left[\Gamma^{(K)}\right]^{-1}\right)^T \\
\left(\left[\Gamma^{(2)}\right]^{-1}\right)\Lambda^{(2,1)}\left(\left[\Gamma^{(1)}\right]^{-1}\right)^T & 
\left(\left[\Gamma^{(2)}\right]^{-1}\right)\Lambda^{(2,2)}\left(\left[\Gamma^{(2)}\right]^{-1}\right)^T & 
\dots & 
\left(\left[\Gamma^{(2)}\right]^{-1}\right)\Lambda^{(2,K)}\left(\left[\Gamma^{(K)}\right]^{-1}\right)^T \\
\vdots & \vdots & \ddots & \vdots \\
\left(\left[\Gamma^{(K)}\right]^{-1}\right)\Lambda^{(K,1)}\left(\left[\Gamma^{(1)}\right]^{-1}\right)^T & 
\left(\left[\Gamma^{(K)}\right]^{-1}\right)\Lambda^{(K,2)}\left(\left[\Gamma^{(2)}\right]^{-1}\right)^T & 
\dots & 
\left(\left[\Gamma^{(K)}\right]^{-1}\right)\Lambda^{(K,K)}\left(\left[\Gamma^{(K)}\right]^{-1}\right)^T
\end{pmatrix},
}%
\end{align}
with $\Gamma^{(k)}:=\Gamma^{(k)}(\bm{\Psi}_0)$ and $\Lambda^{(k,k)}:=\Lambda^{(k,k)}(\bm{\Psi}_0)$ respectively being $\Gamma(\bm{\Psi}_0)$ and $\Lambda(\bm{\Psi}_0)$ in Equation (\ref{eq:asymp:matrix}) of Theorem \ref{thm:asymp}, evaluated at weight function hyperparameters $\tilde{\bm{\Psi}}^{(k)}$. For $k\neq k'$, $\Lambda^{(k,k')}$ is a $(P+1)\times (P+1)$ matrix given by
\begin{align} \label{eq:diag:lambda}
\Lambda^{(k,k')}:=\Lambda^{(k,k')}(\bm{\Psi}_0)=
\begin{pmatrix}
\Lambda_{\theta\theta}^{(k,k')}(\bm{\Psi}_0) & \Lambda_{\theta\phi}^{(k,k')}(\bm{\Psi}_0) \\
\Lambda_{\phi\theta}^{(k,k')}(\bm{\Psi}_0) & \Lambda_{\phi\phi}^{(k,k')}(\bm{\Psi}_0)
\end{pmatrix},
\end{align}
where
\begin{align}
\Lambda_{\theta\theta}^{(k,k')}(\bm{\Psi})
=E_{\bm{x}}\left[\lambda^{(k,k')}(\bm{\Psi};\bm{x})B_{\theta\theta}^{(k,k')}(\bm{\Psi};\bm{x})\frac{\phi^{(k)}\phi^{(k')}}{\phi^2}\xi'(\bm{x}^T\bm{\beta})^2\bm{x}\bm{x}^T\right]
:=E_{\bm{x}}\left[V_{\theta\theta}^{(k,k')}(\bm{\Psi},\bm{x})\bm{x}\bm{x}^T\right],
\end{align}
\begin{align}
\Lambda_{\theta\phi}^{(k,k')}(\bm{\Psi})
&=E_{\bm{x}}\left[\lambda^{(k,k')}(\bm{\Psi};\bm{x})\frac{\phi^{(k)}\phi^{(k')2}}{\phi^3}\left[B^{(k,k')}_{\theta\theta}(\bm{\Psi};\bm{x})\left(\left(c-\frac{1}{\tilde{\phi}^{(k')}}\right)\theta+\frac{\tilde{\theta}^{(k')}}{\tilde{\phi}^{(k')}}\right)+B^{(k,k')}_{\theta\phi}(\bm{\Psi};\bm{x})\right]\xi'(\bm{x}^T\bm{\beta})\bm{x}\right]\nonumber\\
&:=E_{\bm{x}}\left[V^{(k,k')}_{\theta\phi}(\bm{\Psi},\bm{x})\bm{x}\right],
\end{align}
\begin{align}
\Lambda_{\phi\theta}^{(k,k')}(\bm{\Psi})=\Lambda_{\theta\phi}^{(k',k)}(\bm{\Psi})^T,
\end{align}
\begin{align}
\Lambda_{\phi\phi}(\bm{\Psi})
&=E_{\bm{x}}\Bigg[\lambda^{(k,k')}(\bm{\Psi};\bm{x})\frac{\phi^{(k)2}\phi^{(k')2}}{\phi^4}\Bigg[B^{(k,k')}_{\theta\theta}(\bm{\Psi};\bm{x})\left(\left(c-\frac{1}{\tilde{\phi}^{(k)}}\right)\theta+\frac{\tilde{\theta}^{(k)}}{\tilde{\phi}^{(k)}}\right)\left(\left(c-\frac{1}{\tilde{\phi}^{(k')}}\right)\theta+\frac{\tilde{\theta}^{(k')}}{\tilde{\phi}^{(k')}}\right)\nonumber\\
&\hspace{3em} +B^{(k,k')}_{\theta\phi}(\bm{\Psi};\bm{x})\left(\left(c-\frac{1}{\tilde{\phi}^{(k')}}\right)\theta+\frac{\tilde{\theta}^{(k')}}{\tilde{\phi}^{(k')}}\right)+B^{(k,k')}_{\phi\theta}(\bm{\Psi};\bm{x})\left(\left(c-\frac{1}{\tilde{\phi}^{(k)}}\right)\theta+\frac{\tilde{\theta}^{(k)}}{\tilde{\phi}^{(k)}}\right)
+B^{(k,k')}_{\phi\phi}(\bm{\Psi};\bm{x})\Bigg]\Bigg]\nonumber\\
&:=E_{\bm{x}}\left[V^{(k,k')}_{\phi\phi}(\bm{\Psi},\bm{x})\right],
\end{align}
with
\begin{align}
B^{(k,k')}_{\theta\theta}(\bm{\Psi};\bm{x})
=\frac{1}{\phi^{(k)}\phi^{(k')}}\left[\phi^{(k,k')}A''(\theta^{(k,k')})+\left(A'(\theta^{(k,k')})-A'(\theta^{(k)})\right)\left(A'(\theta^{(k,k')})-A'(\theta^{(k')})\right)\right],
\end{align}
\begin{align}
B^{(k,k')}_{\theta\phi}(\bm{\Psi};\bm{x})
&=-\frac{1}{\phi^{(k)}\phi^{(k')2}}\bigg\{(\theta^{(k')}-\theta^{(k,k')})\phi^{(k,k')}A''(\theta^{(k,k')})\nonumber\\
&\hspace{7em}-\left(A'(\theta^{(k)})-A'(\theta^{(k,k')})\right)
\Big[(\theta^{(k')}-\theta^{(k,k')})A'(\theta^{(k,k')})-\left(A(\theta^{(k')})-A(\theta^{(k,k')})\right) \nonumber\\
&\hspace{20em}-\left(\phi^{(k')2}b'(\phi^{(k')})-\phi^{(k,k')2}b'(\phi^{(k,k')})\right)\Big]\bigg\},
\end{align}
\begin{align}
B^{(k,k')}_{\phi\theta}(\bm{\Psi};\bm{x})=B^{(k',k)}_{\theta\phi}(\bm{\Psi};\bm{x}),
\end{align}
\begin{align}
B^{(k,k')}_{\phi\phi}(\bm{\Psi};\bm{x})
&=-\frac{\phi^{(k,k')4}}{\phi^{(k)2}\phi^{(k')2}}\left(\frac{2}{\phi^{(k,k')}}b'(\phi^{(k,k')})+b''(\phi^{(k,k')})\right)\nonumber\\
&\hspace{1em}+\frac{1}{\phi^{(k)2}\phi^{(k')2}}\bigg\{(\theta^{(k)}-\theta^{(k,k')})(\theta^{(k')}-\theta^{(k,k')})\phi^{**}A''(\theta^{(k,k')})\nonumber\\
&\hspace{1em}+\Big[(\theta^{(k)}-\theta^{(k,k')})A'(\theta^{(k,k')})-\left(A(\theta^{(k)})-A(\theta^{(k,k')})\right) -\left(\phi^{(k)2}b'(\phi^{(k)})-\phi^{(k,k')2}b'(\phi^{(k,k')})\right)\Big]\nonumber\\
&\hspace{2em}\times\Big[(\theta^{(k')}-\theta^{(k,k')})A'(\theta^{(k,k')})-\left(A(\theta^{(k')})-A(\theta^{(k,k')})\right) -\left(\phi^{(k')2}b'(\phi^{(k')})-\phi^{(k,k')2}b'(\phi^{(k,k')})\right)\Big]\bigg\}.
\end{align}
Here, $\lambda^{(k,k')}(\bm{\Psi};\bm{x})=\exp\{A(\theta^{(k,k')})/\phi^{(k,k')}-b(\phi^{(k,k')})-A(\theta)/\phi+b(\phi)\}$ is the bias adjustment term, $\phi^{(k)}=(\phi^{-1}+\tilde{\phi}^{(k)-1}-c)^{-1}$, $\theta^{(k)}=(\theta/\phi+\tilde{\theta}^{(k)}/\tilde{\phi}^{(k)})\phi^{(k)}$, $\phi^{(k,k')}=(\phi^{-1}+\tilde{\phi}^{(k)-1}+\tilde{\phi}^{(k')-1}-2c)^{-1}$ and $\theta^{(k,k')}=(\theta/\phi+\tilde{\theta}^{(k)}/\tilde{\phi}^{(k)}+\tilde{\theta}^{(k')}/\tilde{\phi}^{(k')})\phi^{(k,k')}$.

\section{Proofs}
\subsection{Proof of Lemma \ref{lem:density} and Corollary \ref{cor:thm:density}} \label{apx:lemma:wgt}
From Equations (\ref{eq:glm}), (\ref{eq:glm_C}) and (\ref{eq:thm:weight}), we have
\begin{align}
f(y_i;\bm{x}_i,\bm{\Psi})W(y_i,\bm{x}_i)
&\propto\exp\left\{\frac{\theta_iy_i-A(\theta_i)}{\phi}+\left(\frac{1}{\phi}-c\right)g(y_i)+a(y_i)+b(\phi)+\frac{\tilde{\theta}_iy_i}{\tilde{\phi}}+\left(\frac{1}{\tilde{\phi}}-c\right)g(y_i)\right\}\nonumber\\
&=\exp\left\{\frac{A(\theta^*_i)}{\phi^*}-b(\phi^*)-\frac{A(\theta_i)}{\phi}+b(\phi)\right\}\times\exp\left\{\frac{\theta_i^*y_i-A(\theta_i^*)}{\phi^*}+C(y_i,\phi^*)\right\}
\end{align}
with $\phi^*=(\phi^{-1}+\tilde{\phi}^{-1}-c)^{-1}$ and $\theta_i^*=(\theta_i/\phi+\tilde{\theta}_i/\tilde{\phi})\phi^*$. Integrating the above expression, we have
\begin{align}
\lambda_i^*(\bm{\Psi};\bm{x}_i)
&=\int_{\mathcal{Y}}f(u;\bm{x}_i,\bm{\Psi})W(u,\bm{x}_i)du\nonumber\\
&\propto\exp\left\{\frac{A(\theta^*_i)}{\phi^*}-b(\phi^*)-\frac{A(\theta_i)}{\phi}+b(\phi)\right\}\times\int_{\mathcal{Y}}\exp\left\{\frac{\theta_i^*u-A(\theta_i^*)}{\phi^*}+C(u,\phi^*)\right\}du\nonumber\\
&=\exp\left\{\frac{A(\theta^*_i)}{\phi^*}-b(\phi^*)-\frac{A(\theta_i)}{\phi}+b(\phi)\right\}
\end{align}
and
\begin{align}
f^*(y_i;\bm{x}_i,\bm{\Psi})
&=\frac{f(y_i;\bm{x}_i,\bm{\Psi})W(y_i,\bm{x}_i)}{\int_{\mathcal{Y}}f(u;\bm{x}_i,\bm{\Psi})W(u,\bm{x}_i)du}
=\exp\left\{\frac{\theta_i^*y_i-A(\theta_i^*)}{\phi^*}+C(y_i,\phi^*)\right\}.
\end{align}

As a result, Lemma \ref{lem:density} holds. Statement 1 of Corollary \ref{cor:thm:density} holds immediately by the definition of $\theta_i^{*}$. For statement 2, given that a canonical link is selected for the GLM, we can write
\begin{align}
\theta^*_i=\left(\frac{\bm{x}_i^T\bm{\beta}}{\phi}+\frac{\bm{x}_i^T\tilde{\bm{\beta}}}{\tilde{\phi}}\right)\phi^*=\bm{x}_i^T\left(\frac{\bm{\beta}}{\phi}+\frac{\tilde{\bm{\beta}}}{\tilde{\phi}}\right)\phi^*
:=\bm{x}_i^T\bm{\beta}^*.
\end{align}

The GLM link is still canonical after transformation, with the transformed regression coefficients given by $\bm{\beta}^*=(\bm{\beta}/\phi+\tilde{\bm{\beta}}/\tilde{\phi})\phi^*$. The result then follows.

\subsection{Proof of Theorem \ref{thm:if} and Corollary \ref{cor:if}}
The influence function (IF) in Equation (\ref{eq:robust:if}) is evaluated by (huber 1981 robust statistics???) as (we here omit the subscript $i$ for $y_i$ and $\bm{x}_i$):
\begin{align} \label{apx:eq:if}
\text{IF}(\bm{\Psi}_0;F,\Delta)=-\Gamma^{-1}\int_{\mathcal{Y}\times\mathcal{X}}\mathcal{S}(\bm{\Psi}_0;y,\bm{x})d\Delta(y,\bm{x}),
\end{align}
where $\Gamma=E_{Y,\bm{x}}\left[\frac{\partial}{\partial\bm{\Psi}}\mathcal{S}(\bm{\Psi}_0;Y,\bm{x})^T\Big|_{\bm{\Psi}=\bm{\Psi}_0}\right]$ is a $(P+1)\times(P+1)$ Hessian matrix and $E_{Y,\bm{x}}[\cdot]$ is an expectation taken on $(Y,\bm{x})$ assuming that $Y|\bm{x}$ follows the unperturbed GLM parameterized by $\bm{\Psi}_0$. The regularity condition 3 in Section \ref{apx:sec:reg_com_asymp} implies that $\Gamma$ exists and is finite. To prove that Equation (\ref{eq:robust:if_bound}) holds (i.e. Equation (\ref{apx:eq:if}) above is bounded), it suffices to show that the score function $\mathcal{S}(\bm{\Psi}_0;y,\bm{x})$ is bounded for $y\in\mathcal{Y}$ and $\bm{x}\in\bar{\mathcal{X}}$. Given that regularity condition 1 in Section \ref{apx:sec:reg_com_asymp} holds and $\bar{\mathcal{X}}$ is compact, we can conclude that $\theta$, $\theta^*$, $A(\theta^*)$, $A'(\theta^*)$, $\xi'(\bm{x}^T\bm{\beta})$ and $\bm{x}$ are all bounded. Observing Equations (\ref{eq:thm:score_glm_t}) and (\ref{eq:thm:score_glm_p}), it suffices to show that $W(y,\bm{x})$, $W(y,\bm{x})y$ and $W(y,\bm{x})g(y)$ are bounded for $y\in\mathcal{Y}$ and $\bm{x}\in\bar{\mathcal{X}}$, which is implied by Assumption (ii) of Theorem \ref{thm:if}. Hence, the result of Theorem \ref{thm:if} follows.

For Corollary \ref{cor:if}, it suffices to show that for every $\bm{x}\in\mathcal{X}$, all $W(y,\bm{x})$, $W(y,\bm{x})y$ and $W(y,\bm{x})g(y)$ converges to 0 as $y\rightarrow\inf\{\mathcal{Y}\}$ or $y\rightarrow\sup\{\mathcal{Y}\}$. These can all be easily verified for the case when $F(y;\bm{x},\bm{\Psi}_0)$ follows Gamma GLM, linear model or inverse-Gaussian GLM, the weight function is in the form of Equation (\ref{eq:thm:weight}) and the hyperparameters are selected according to Example \ref{eg:thm}.

\subsection{A technical lemma} \label{apx:lemma}
\begin{lemma} \label{apx:lemma:D}
Let $f(y;\bm{\Psi})$ be an exponential dispersion density function given by
\begin{equation}
f(y;\bm{\Psi})=\exp\left\{\frac{\theta y-A(\theta)}{\phi}+\left(\frac{1}{\phi}+c\right)g(y)+a(y)+b(\phi)\right\}.
\end{equation}
Denote $\mathcal{U}\subset\mathbb{R}$ as the uncensoring region, $F(y;\bm{\Psi})$ as the corresponding distribution function and $F(\mathcal{U};\bm{\Psi})=\int_{\mathcal{U}}f(y;\bm{\Psi})dy$. Further denote the following expressions:
\begin{equation}
D_{\theta}(\mathcal{U};\bm{\Psi}):=\int_{\mathcal{U}}(y-A'(\theta))f(y;\bm{\Psi})dy,
\end{equation}
\begin{equation}
D_{\phi}(\mathcal{U};\bm{\Psi}):=\int_{\mathcal{U}}(\theta y-A(\theta)+g(y)-\phi^2b'(\phi))f(y;\bm{\Psi})dy
\end{equation}
\begin{equation}
D_{\theta\theta}(\mathcal{U};\bm{\Psi}):=\int_{\mathcal{U}}(y-A'(\theta))^2f(y;\bm{\Psi})dy,
\end{equation}
\begin{equation}
D_{\phi\phi}(\mathcal{U};\bm{\Psi}):=\int_{\mathcal{U}}(\theta y-A(\theta)+g(y)-\phi^2b'(\phi))^2f(y;\bm{\Psi})dy
\end{equation}
\begin{equation}
D_{\theta\phi}(\mathcal{U};\bm{\Psi}):=\int_{\mathcal{U}}(y-A'(\theta))(\theta y-A(\theta)+g(y)-\phi^2b'(\phi))f(y;\bm{\Psi})dy
\end{equation}

Then, all five expressions above can be analytically simplified to Equations (\ref{apx:lemma:Dt}) to (\ref{apx:lemma:Dtp}) presented in the proof below respectively.
\end{lemma}

\begin{proof}
Given that $F(\mathcal{U};\bm{\Psi})=\int_{\mathcal{U}}f(y;\bm{\Psi})dy$, we take derivative with respect to various parameters on both sides of the equations and the following results follow:

\begin{enumerate}
    \item Differentiating w.r.t. $\theta$:
    \begin{equation*}
    \frac{\partial}{\partial\theta}F(\mathcal{U};\bm{\Psi})=\int_{\mathcal{U}}\frac{1}{\phi}(y-A'(\theta))f(y;\bm{\Psi})dy=\frac{1}{\phi}D_{\theta}(\mathcal{U};\bm{\Psi})
    \end{equation*}
    \begin{equation} \label{apx:lemma:Dt}
    \Rightarrow ~ D_{\theta}(\mathcal{U};\bm{\Psi})=\phi\frac{\partial}{\partial\theta}F(\mathcal{U};\bm{\Psi}).
    \end{equation}
    \item Differentiating w.r.t. $\phi$:
    \begin{equation*}
    \frac{\partial}{\partial\phi}F(\mathcal{U};\bm{\Psi})=\int_{\mathcal{U}}\left\{-\frac{1}{\phi^2}[\theta y-A(\theta)+g(y)]+b'(\phi)\right\}f(y;\bm{\Psi})dy=-\frac{1}{\phi^2}D_{\phi}(\mathcal{U};\bm{\Psi})
    \end{equation*}
    \begin{equation} \label{apx:lemma:Dp}
    \Rightarrow ~ D_{\phi}(\mathcal{U};\bm{\Psi})=-\phi^2\frac{\partial}{\partial\phi}F(\mathcal{U};\bm{\Psi}).
    \end{equation}
    \item Differentiating w.r.t. $\theta$ twice:
    \begin{equation*}
    \frac{\partial^2}{\partial\theta^2}F(\mathcal{U};\bm{\Psi})=\int_{\mathcal{U}}\left\{\frac{1}{\phi^2}(y-A'(\theta))^2-\frac{1}{\phi}A''(\theta)\right\}f(y;\bm{\Psi})dy=\frac{1}{\phi^2}D_{\theta\theta}(\mathcal{U};\bm{\Psi})-\frac{1}{\phi}A''(\theta)F(\mathcal{U};\bm{\Psi})
    \end{equation*}
    \begin{equation} \label{apx:lemma:Dtt}
    \Rightarrow ~ D_{\theta\theta}(\mathcal{U};\bm{\Psi})=\phi A''(\theta)F(\mathcal{U};\bm{\Psi})+\phi^2\frac{\partial^2}{\partial\theta^2}F(\mathcal{U};\bm{\Psi}).
    \end{equation}
    \item Differentiating w.r.t. $\phi$ twice:
    \begin{align*}
    \frac{\partial^2}{\partial\phi^2}F(\mathcal{U};\bm{\Psi})
    &=\int_{\mathcal{U}}\left[\left\{\frac{2}{\phi^3}[\theta y-A(\theta)+g(y)]+b''(\phi)\right\}+\left\{-\frac{1}{\phi^2}[\theta y-A(\theta)+g(y)]+b'(\phi)\right\}^2\right]f(y;\bm{\Psi})dy\\
    &=\frac{2}{\phi^3}D_{\phi}(\mathcal{U};\bm{\Psi})+\frac{2}{\phi}b'(\phi)F(\mathcal{U};\bm{\Psi})+b''(\phi)F(\mathcal{U};\bm{\Psi})+\frac{1}{\phi^4}D_{\phi\phi}(\mathcal{U};\bm{\Psi})
    \end{align*}
    \begin{equation} \label{apx:lemma:Dpp}
    \Rightarrow ~ D_{\phi\phi}(\mathcal{U};\bm{\Psi})=-\left[2\phi^3b'(\phi)F(\mathcal{U};\bm{\Psi})+\phi^4b''(\phi)\right]F(\mathcal{U};\bm{\Psi})+2\phi^3\frac{\partial}{\partial\phi}F(\mathcal{U};\bm{\Psi})+\phi^4\frac{\partial^2}{\partial\phi^2}F(\mathcal{U};\bm{\Psi}).
    \end{equation}
    \item Differentiating w.r.t. $\theta$ and $\phi$:
    \begin{align*}
    \frac{\partial^2}{\partial\theta\partial\phi}F(\mathcal{U};\bm{\Psi})
    &=\int_{\mathcal{U}}\left[-\frac{1}{\phi^2}(y-A'(\theta))\right]f(y;\bm{\Psi})dy\\
    &\hspace{1em}+\int_{\mathcal{U}}\left[-\frac{1}{\phi^2}(\theta y-A(\theta)+g(y))+b'(\phi)\right]\times\left[\frac{1}{\phi}(y-A'(\theta))\right]f(y;\bm{\Psi})dy\\
    &=-\frac{1}{\phi^2}D_{\theta}(\mathcal{U};\bm{\Psi})-\frac{1}{\phi^3}D_{\theta\phi}(\mathcal{U};\bm{\Psi})
    \end{align*}
    \begin{equation} \label{apx:lemma:Dtp}
    \Rightarrow ~ D_{\theta\phi}(\mathcal{U};\bm{\Psi})
    =-\phi^2\frac{\partial}{\partial\theta}F(\mathcal{U};\bm{\Psi})-\phi^3\frac{\partial^2}{\partial\theta\partial\phi}F(\mathcal{U};\bm{\Psi}).
    \end{equation}
\end{enumerate}
\end{proof}

\subsection{Proof of Theorem \ref{thm:censtrun:asymp}} \label{apx:thm:censtrun:asymp}
Since SWLE is within a class of M-estimators, consistency and asymptotic normality can be proved by using Theorems 5.41 and 5.42 of \cite{van2000asymptotic}. The proof of consistency requires that $\bm{\Psi}_0$ is the solution of the expected individual score function $E_{\mathcal{D},\bm{x}}[\mathcal{S}(\bm{\Psi};\mathcal{D},\bm{x})]=\bm{0}$, where $E_{\mathcal{D},\bm{x}}[\cdot]$ is the expectation taken on $\mathcal{D}$ and $\bm{x}$. Note that we have made a slight abuse of notations, because $\mathcal{D}:=(\mathcal{R},\mathcal{T},y1\{y\in\mathcal{U}\},\{1\{y\in\mathcal{I}_{m}\}\}_{m=1,\ldots,M})$ here represents the individual observed information with subscript $i$ dropped, in opposed to the main text of this paper where $\mathcal{D}$ represents the observed information across all losses. Taking a double expectation conditioned on the censoring and truncation mechanisms $(\mathcal{R},\mathcal{T})$ and covariates $\bm{x}$, it suffices to show that
\begin{equation} \label{apx:asymp:consis}
E_Y[\mathcal{S}(\bm{\Psi};\mathcal{D},\bm{x})|\mathcal{R},\mathcal{T},\bm{x}]:=\int_{\mathcal{Y}}\mathcal{S}(\bm{\Psi};\mathcal{D},\bm{x})f_{\mathcal{T}}(y;\bm{x},\bm{\Psi})dy=\bm{0}
\end{equation}
Evaluating the above expression, we have:
\begin{align}
E_Y[\mathcal{S}(\bm{\Psi};\mathcal{D},\bm{x})|\mathcal{R},\mathcal{T},\bm{x}]
&=\int_{\mathcal{U}}W(y;\bm{x})\left\{\frac{\partial}{\partial\bm{\Psi}}\log f^*_{\mathcal{T}}(y;\bm{x},\bm{\Psi})\right\}f_{\mathcal{T}}(y;\bm{x},\bm{\Psi})dy\nonumber\\
&\quad +\left(\int_{\mathcal{Y}}f_{\mathcal{T}}(u;\bm{x},\bm{\Psi})W(u,\bm{x})du\right)\sum_{m=1}^{M}\frac{F^*_{\mathcal{T}}(\mathcal{I}_{m};\bm{x},\bm{\Psi})}{F_{\mathcal{T}}(\mathcal{I}_{m};\bm{x},\bm{\Psi})}\left\{\frac{\partial}{\partial\bm{\Psi}}\log F^*_{\mathcal{T}}(\mathcal{I}_{m};\bm{x},\bm{\Psi})\right\}F_{\mathcal{T}}(\mathcal{I}_{m};\bm{x},\bm{\Psi})\nonumber\\
&=\left(\int_{\mathcal{Y}}f_{\mathcal{T}}(u;\bm{x},\bm{\Psi})W(u,\bm{x})du\right)\int_{\mathcal{U}}f^*_{\mathcal{T}}(y;\bm{x},\bm{\Psi})\frac{\partial}{\partial\bm{\Psi}}\log f^*_{\mathcal{T}}(y;\bm{x},\bm{\Psi})dy\nonumber\\
&\quad +\left(\int_{\mathcal{Y}}f_{\mathcal{T}}(u;\bm{x},\bm{\Psi})W(u,\bm{x})du\right)\sum_{m=1}^{M}F^*_{\mathcal{T}}(\mathcal{I}_{m};\bm{x},\bm{\Psi})\frac{\partial}{\partial\bm{\Psi}}\log F^*_{\mathcal{T}}(\mathcal{I}_{m};\bm{x},\bm{\Psi})\nonumber\\
&=\left(\int_{\mathcal{Y}}f_{\mathcal{T}}(u;\bm{x},\bm{\Psi})W(u,\bm{x})du\right)
\left\{\int_{\mathcal{U}}\frac{\partial}{\partial\bm{\Psi}}f^*_{\mathcal{T}}(y;\bm{x},\bm{\Psi})dy+\sum_{m=1}^{M}\frac{\partial}{\partial\bm{\Psi}}\int_{\mathcal{I}_m}f^*_{\mathcal{T}}(y;\bm{x},\bm{\Psi})dy\right\}\nonumber\\
&=\left(\int_{\mathcal{Y}}f_{\mathcal{T}}(u;\bm{x},\bm{\Psi})W(u,\bm{x})du\right)\frac{\partial}{\partial\bm{\Psi}}\int_{\mathcal{T}}f^*_{\mathcal{T}}(y;\bm{x},\bm{\Psi})dy=\bm{0},
\end{align}
which proves the consistency results. For asymptotic normality, we denote the following notations before proceeding. First, define the transformed density function
\begin{align}
f^{**}(y;\bm{x},\bm{\Psi})=\exp\left\{\frac{\theta^{**}y-A(\theta^{**})}{\phi^{**}}+C(y,\phi^{**})\right\},
\end{align}
and the corresponding truncated density function $f^{**}_{\mathcal{T}}(y;\bm{x},\bm{\Psi})$ represented in the form of Equation (\ref{eq:censtrun:density}). The corresponding distributions are then denoted as $F^{**}(y;\bm{x},\bm{\Psi})$ or $F^{**}_{\mathcal{T}}(y;\bm{x},\bm{\Psi})$. Also, we denote $D^{**}_{\theta}(\mathcal{U};\bm{x},\bm{\Psi})$, $D^{**}_{\phi}(\mathcal{U};\bm{x},\bm{\Psi})$, $D^{**}_{\theta\theta}(\mathcal{U};\bm{x},\bm{\Psi})$, $D^{**}_{\theta\phi}(\mathcal{U};\bm{x},\bm{\Psi})$ and $D^{**}_{\phi\phi}(\mathcal{U};\bm{x},\bm{\Psi})$ as $D_{\theta}(\mathcal{U};\bm{\Psi})$, $D_{\phi}(\mathcal{U};\bm{\Psi})$, $D_{\theta\theta}(\mathcal{U};\bm{\Psi})$, $D_{\theta\phi}(\mathcal{U};\bm{\Psi})$ and $D_{\phi\phi}(\mathcal{U};\bm{\Psi})$ in Equations (\ref{apx:lemma:Dt}) to (\ref{apx:lemma:Dtp}) evaluated at $\theta^{**}$ and $\phi^{**}$ with covariates $\bm{x}$.

With the regularity conditions satisfied, Theorem 5.41 of \cite{van2000asymptotic} shows that $\sqrt{n}(\hat{\bm{\Phi}}_n-\bm{\Phi}_0)\overset{d}{\rightarrow}\mathcal{N}(\bm{0},\bm{\Sigma})$ with 
$\bm{\Sigma}=\Gamma^{-1}\Lambda(\Gamma^{-1})^T$, where 
\begin{align} \label{eq:apx:gamma}
\Gamma:=\Gamma(\bm{\Psi}_0)=
E_{\mathcal{D},\bm{x}}\left[\frac{\partial}{\partial\bm{\Psi}}\mathcal{S}(\bm{\Psi};\mathcal{D},\bm{x})^T\right]\Bigg|_{\bm{\Psi}=\bm{\Psi}_0}
&=E_{\mathcal{D},\bm{x}}\left[E_{Y}\left[\frac{\partial}{\partial\bm{\Psi}}\mathcal{S}(\bm{\Psi};\mathcal{D},\bm{x})^T|\mathcal{R},\mathcal{T},\bm{x}\right]\right]\Bigg|_{\bm{\Psi}=\bm{\Psi}_0}\nonumber\\
&:=E_{\mathcal{D},\bm{x}}\left[\tilde{\Gamma}(\bm{\Psi}_0;\mathcal{R},\mathcal{T},\bm{x})\right]
\end{align}
and
\begin{align} \label{eq:apx:lambda}
\Lambda:=\Lambda(\bm{\Psi}_0)=
E_{\mathcal{D},\bm{x}}\left[\mathcal{S}(\bm{\Psi};\mathcal{D},\bm{x})\mathcal{S}(\bm{\Psi};\mathcal{D},\bm{x})^T\right]\Bigg|_{\bm{\Psi}=\bm{\Psi}_0}
&=E_{\mathcal{D},\bm{x}}\left[E_{Y}\left[\mathcal{S}(\bm{\Psi};\mathcal{D},\bm{x})\mathcal{S}(\bm{\Psi};\mathcal{D},\bm{x})^T|\mathcal{R},\mathcal{T},\bm{x}\right]\right]\Bigg|_{\bm{\Psi}=\bm{\Psi}_0}\nonumber\\
&:=E_{\mathcal{D},\bm{x}}\left[\tilde{\Lambda}(\bm{\Psi}_0;\mathcal{R},\mathcal{T},\bm{x})\right].
\end{align}

We now derive $\tilde{\Gamma}(\bm{\Psi};\mathcal{R},\mathcal{T},\bm{x})$ and $\tilde{\Lambda}(\bm{\Psi};\mathcal{R},\mathcal{T},\bm{x})$ as follows. First, we write
\begin{align}
\tilde{\Gamma}(\bm{\Psi};\mathcal{R},\mathcal{T},\bm{x}):=
\begin{pmatrix}
\tilde{\Gamma}_{\theta\theta}(\bm{\Psi};\mathcal{R},\mathcal{T},\bm{x}) & \tilde{\Gamma}_{\theta\phi}(\bm{\Psi};\mathcal{R},\mathcal{T},\bm{x})  \\
\tilde{\Gamma}_{\phi\theta}(\bm{\Psi};\mathcal{R},\mathcal{T},\bm{x})  & \tilde{\Gamma}_{\phi\phi}(\bm{\Psi};\mathcal{R},\mathcal{T},\bm{x}),
\end{pmatrix}
\end{align}
where the four elements above are expressed as follows.

\begin{align} \label{eq:apx:gamma_tt}
\tilde{\Gamma}_{\theta\theta}(\bm{\Psi};\mathcal{R},\mathcal{T},\bm{x})
&=\int_{\mathcal{U}}W(y;\bm{x})\left\{\frac{\partial^2}{\partial\bm{\beta}\partial\bm{\beta}^T}\log f^*_{\mathcal{T}}(y;\bm{x},\bm{\Psi})\right\}f_{\mathcal{T}}(y;\bm{x},\bm{\Psi})dy\nonumber\\
&\quad +\sum_{m=1}^{M}\left\{\frac{\partial}{\partial\bm{\beta}}\left\{\lambda^*(\bm{\Psi};\bm{x})\frac{F^*(\mathcal{I}_{m};\bm{x},\bm{\Psi})}{F(\mathcal{I}_{m};\bm{x},\bm{\Psi})}\frac{\partial}{\partial\bm{\beta}^T} \log F^*_{\mathcal{T}}(\mathcal{I}_{m};\bm{x},\bm{\Psi})\right\}\right\}F_{\mathcal{T}}(\mathcal{I}_{m};\bm{x},\bm{\Psi})\nonumber\\
&=\int_{\mathcal{U}}W(y;\bm{x})\left\{\frac{\partial^2}{\partial\theta^{*2}}\log f^*_{\mathcal{T}}(y;\bm{x},\bm{\Psi})\right\}f_{\mathcal{T}}(y;\bm{x},\bm{\Psi})dy\left(\frac{\partial\theta^*}{\partial\bm{\beta}}\right)\left(\frac{\partial\theta^*}{\partial\bm{\beta}}\right)^T\nonumber\\
&\quad +\int_{\mathcal{U}}W(y;\bm{x})\left\{\frac{\partial}{\partial\theta^{*}}\log f^*_{\mathcal{T}}(y;\bm{x},\bm{\Psi})\right\}f_{\mathcal{T}}(y;\bm{x},\bm{\Psi})dy\left(\frac{\partial\theta^{*2}}{\partial\bm{\beta}\partial\bm{\beta}^T}\right)\nonumber\\
&\quad +\sum_{m=1}^{M}\left\{\frac{\partial}{\partial\theta^*}\left\{\lambda^*(\bm{\Psi};\bm{x})\frac{F^*(\mathcal{I}_{m};\bm{x},\bm{\Psi})}{F(\mathcal{I}_{m};\bm{x},\bm{\Psi})}\frac{\partial}{\partial\theta^*} \log F^*_{\mathcal{T}}(\mathcal{I}_{m};\bm{x},\bm{\Psi})\right\}\right\}F_{\mathcal{T}}(\mathcal{I}_{m};\bm{x},\bm{\Psi})\left(\frac{\partial\theta^*}{\partial\bm{\beta}}\right)\left(\frac{\partial\theta^*}{\partial\bm{\beta}}\right)^T\nonumber\\
&\quad +\sum_{m=1}^{M}\left\{\lambda^*(\bm{\Psi};\bm{x})\frac{F^*(\mathcal{I}_{m};\bm{x},\bm{\Psi})}{F(\mathcal{I}_{m};\bm{x},\bm{\Psi})}\frac{\partial}{\partial\theta^*} \log F^*_{\mathcal{T}}(\mathcal{I}_{m};\bm{x},\bm{\Psi})\right\}F_{\mathcal{T}}(\mathcal{I}_{m};\bm{x},\bm{\Psi})\left(\frac{\partial\theta^{*2}}{\partial\bm{\beta}\partial\bm{\beta}^T}\right)\nonumber\\
&:=\left(P_{\theta\theta}(\bm{\Psi};\mathcal{R},\mathcal{T},\bm{x})+Q_{\theta\theta}(\bm{\Psi};\mathcal{R},\mathcal{T},\bm{x})\right)\left(\frac{\partial\theta^*}{\partial\bm{\beta}}\right)\left(\frac{\partial\theta^*}{\partial\bm{\beta}}\right)^T
:=W_{\theta\theta}(\bm{\Psi};\mathcal{R},\mathcal{T},\bm{x})\bm{x}\bm{x}^T,
\end{align}
with
\begin{align}
W_{\theta\theta}(\bm{\Psi};\mathcal{R},\mathcal{T},\bm{x})=\frac{\phi^{*2}}{\phi^2}\left(\xi'(\bm{x}^T\bm{\beta})\right)^2\left(P_{\theta\theta}(\bm{\Psi};\mathcal{R},\mathcal{T},\bm{x})+Q_{\theta\theta}(\bm{\Psi};\mathcal{R},\mathcal{T},\bm{x})\right),
\end{align}
\begin{align}
P_{\theta\theta}(\bm{\Psi};\mathcal{R},\mathcal{T},\bm{x})
&=\int_{\mathcal{U}}W(y;\bm{x})\left\{\frac{\partial^2}{\partial\theta^{*2}}\log f^*_{\mathcal{T}}(y;\bm{x},\bm{\Psi})\right\}f_{\mathcal{T}}(y;\bm{x},\bm{\Psi})dy\nonumber\\
&=-\frac{\lambda^*(\bm{\Psi};\bm{x})}{F(\mathcal{T};\bm{x},\bm{\Psi})}\left[\frac{1}{\phi^*}A''(\theta^*)+\frac{\partial^2}{\partial\theta^{*2}}\log F^*(\mathcal{T};\bm{x},\bm{\Psi})\right]F^*(\mathcal{U};\bm{x},\bm{\Psi}),
\end{align}
\begin{align}
Q_{\theta\theta}(\bm{\Psi};\mathcal{R},\mathcal{T},\bm{x})
&=\sum_{m=1}^{M}\left\{\frac{\partial}{\partial\theta^*}\left\{\lambda^*(\bm{\Psi};\bm{x})\frac{F^*(\mathcal{I}_{m};\bm{x},\bm{\Psi})}{F(\mathcal{I}_{m};\bm{x},\bm{\Psi})}\frac{\partial}{\partial\theta^*} \log F^*_{\mathcal{T}}(\mathcal{I}_{m};\bm{x},\bm{\Psi})\right\}\right\}F_{\mathcal{T}}(\mathcal{I}_{m};\bm{x},\bm{\Psi}).
\end{align}

Note that in Equation (\ref{eq:apx:gamma_tt}) the sum of the second and forth terms of the second equality is zero, because it represents the expected score function multiplied by $\partial\theta^{*2}/\partial\bm{\beta}\partial\bm{\beta}^T$. (Need some explanations on how to evaluate (7.20)). Similarly, the other elements can be expressed as

\begin{align}
\tilde{\Gamma}_{\theta\phi}(\bm{\Psi};\mathcal{R},\mathcal{T},\bm{x})
&=\left(P_{\theta\theta}(\bm{\Psi};\mathcal{R},\mathcal{T},\bm{x})+Q_{\theta\theta}(\bm{\Psi};\mathcal{R},\mathcal{T},\bm{x})\right)\left(\frac{\partial\theta^*}{\partial\phi}\right)\left(\frac{\partial\theta^*}{\partial\bm{\beta}}\right)\nonumber\\
&\quad + \left(P_{\theta\phi}(\bm{\Psi};\mathcal{R},\mathcal{T},\bm{x})+Q_{\theta\phi}(\bm{\Psi};\mathcal{R},\mathcal{T},\bm{x})\right)\left(\frac{\partial\phi^*}{\partial\phi}\right)\left(\frac{\partial\theta^*}{\partial\bm{\beta}}\right)\nonumber\\
&:=W_{\theta\phi}(\bm{\Psi};\mathcal{R},\mathcal{T},\bm{x})\bm{x},
\end{align}
with
\begin{align}
W_{\theta\phi}(\bm{\Psi};\mathcal{R},\mathcal{T},\bm{x})
&=\frac{\phi^{*3}}{\phi^3}\xi'(\bm{x}^T\bm{\beta})\Bigg[\left(\left(c-\frac{1}{\tilde{\phi}}\right)\theta+\frac{\tilde{\theta}}{\tilde{\phi}}\right)\left(P_{\theta\theta}(\bm{\Psi};\mathcal{R},\mathcal{T},\bm{x})+Q_{\theta\theta}(\bm{\Psi};\mathcal{R},\mathcal{T},\bm{x})\right)\nonumber\\
&\hspace{15em}+\left(P_{\theta\phi}(\bm{\Psi};\mathcal{R},\mathcal{T},\bm{x})+Q_{\theta\phi}(\bm{\Psi};\mathcal{R},\mathcal{T},\bm{x})\right)\Bigg],
\end{align}
\begin{align}
P_{\theta\phi}(\bm{\Psi};\mathcal{R},\mathcal{T},\bm{x})
&=\int_{\mathcal{U}}W(y;\bm{x})\left\{\frac{\partial^2}{\partial\theta^{*}\partial\phi^*}\log f^*_{\mathcal{T}}(y;\bm{x},\bm{\Psi})\right\}f_{\mathcal{T}}(y;\bm{x},\bm{\Psi})dy\nonumber\\
&=-\frac{\lambda^*(\bm{\Psi};\bm{x})}{F(\mathcal{T};\bm{x},\bm{\Psi})}\left[\frac{1}{\phi^{*2}}\int_{\mathcal{U}}(y-A'(\theta^*))f^*(y;\bm{x},\bm{\Psi})dy+\left(\frac{\partial^2}{\partial\theta^{*}\partial\phi^*}\log F^*(\mathcal{T};\bm{x},\bm{\Psi})\right)F^*(\mathcal{U};\bm{x},\bm{\Psi})\right]\nonumber\\
&=-\frac{\lambda^*(\bm{\Psi};\bm{x})}{F(\mathcal{T};\bm{x},\bm{\Psi})}\left[\frac{1}{\phi^{*2}}D^*_{\theta}(\mathcal{U};\bm{x},\bm{\Psi})
+\left(\frac{\partial^2}{\partial\theta^{*}\partial\phi^*}\log F^*(\mathcal{T};\bm{x},\bm{\Psi})\right)F^*(\mathcal{U};\bm{x},\bm{\Psi})\right]
\end{align}
\begin{align}
Q_{\theta\phi}(\bm{\Psi};\mathcal{R},\mathcal{T},\bm{x})
&=\sum_{m=1}^{M}\left\{\frac{\partial}{\partial\theta^*}\left\{\lambda^*(\bm{\Psi};\bm{x})\frac{F^*(\mathcal{I}_{m};\bm{x},\bm{\Psi})}{F(\mathcal{I}_{m};\bm{x},\bm{\Psi})}\frac{\partial}{\partial\phi^*} \log F^*_{\mathcal{T}}(\mathcal{I}_{m};\bm{x},\bm{\Psi})\right\}\right\}F_{\mathcal{T}}(\mathcal{I}_{m};\bm{x},\bm{\Psi}),
\end{align}
where $D^*_{\theta}(\mathcal{U};\bm{x},\bm{\Psi})$ is simply $D_{\theta}(\mathcal{U};\bm{\Psi})$ in Equation (\ref{apx:lemma:Dt}) of Lemma \ref{apx:lemma:D} evaluated at $\theta^*$, $\phi^*$ and $\bm{x}$.

\begin{align}
\tilde{\Gamma}_{\phi\theta}(\bm{\Psi};\mathcal{R},\mathcal{T},\bm{x})
&:=W_{\phi\theta}(\bm{\Psi};\mathcal{R},\mathcal{T},\bm{x})\bm{x}^T,
\end{align}
with
\begin{align}
W_{\phi\theta}(\bm{\Psi};\mathcal{R},\mathcal{T},\bm{x})
&=\frac{\phi^{*3}}{\phi^3}\xi'(\bm{x}^T\bm{\beta})\Bigg[\left(\left(c-\frac{1}{\tilde{\phi}}\right)\theta+\frac{\tilde{\theta}}{\tilde{\phi}}\right)\left(P_{\theta\theta}(\bm{\Psi};\mathcal{R},\mathcal{T},\bm{x})+Q_{\theta\theta}(\bm{\Psi};\mathcal{R},\mathcal{T},\bm{x})\right)\nonumber\\
&\hspace{15em}+\left(P_{\phi\theta}(\bm{\Psi};\mathcal{R},\mathcal{T},\bm{x})+Q_{\phi\theta}(\bm{\Psi};\mathcal{R},\mathcal{T},\bm{x})\right)\Bigg],
\end{align}
\begin{align}
P_{\phi\theta}(\bm{\Psi};\mathcal{R},\mathcal{T},\bm{x})=P_{\theta\phi}(\bm{\Psi};\mathcal{R},\mathcal{T},\bm{x}),
\end{align}
\begin{align}
Q_{\phi\theta}(\bm{\Psi};\mathcal{R},\mathcal{T},\bm{x})
&=\sum_{m=1}^{M}\left\{\frac{\partial}{\partial\phi^*}\left\{\lambda^*(\bm{\Psi};\bm{x})\frac{F^*(\mathcal{I}_{m};\bm{x},\bm{\Psi})}{F(\mathcal{I}_{m};\bm{x},\bm{\Psi})}\frac{\partial}{\partial\theta^*} \log F^*_{\mathcal{T}}(\mathcal{I}_{m};\bm{x},\bm{\Psi})\right\}\right\}F_{\mathcal{T}}(\mathcal{I}_{m};\bm{x},\bm{\Psi}).
\end{align}

\begin{align}
\tilde{\Gamma}_{\phi\phi}(\bm{\Psi};\mathcal{R},\mathcal{T},\bm{x})
&:=W_{\phi\phi}(\bm{\Psi};\mathcal{R},\mathcal{T},\bm{x}),
\end{align}
with
\begin{align}
W_{\phi\phi}(\bm{\Psi};\mathcal{R},\mathcal{T},\bm{x})
&=\frac{\phi^{*4}}{\phi^4}\Bigg[\left(\left(c-\frac{1}{\tilde{\phi}}\right)\theta+\frac{\tilde{\theta}}{\tilde{\phi}}\right)^2\left(P_{\theta\theta}(\bm{\Psi};\mathcal{R},\mathcal{T},\bm{x})+Q_{\theta\theta}(\bm{\Psi};\mathcal{R},\mathcal{T},\bm{x})\right)\nonumber\\
&\hspace{3em}+\left(\left(c-\frac{1}{\tilde{\phi}}\right)\theta+\frac{\tilde{\theta}}{\tilde{\phi}}\right)\Big[\left(P_{\theta\phi}(\bm{\Psi};\mathcal{R},\mathcal{T},\bm{x})+Q_{\theta\phi}(\bm{\Psi};\mathcal{R},\mathcal{T},\bm{x})\right)\nonumber\\
&\hspace{13em}+\left(P_{\phi\theta}(\bm{\Psi};\mathcal{R},\mathcal{T},\bm{x})+Q_{\phi\theta}(\bm{\Psi};\mathcal{R},\mathcal{T},\bm{x})\right)\Big]\nonumber\\
&\hspace{3em}+\left(P_{\phi\phi}(\bm{\Psi};\mathcal{R},\mathcal{T},\bm{x})+Q_{\phi\phi}(\bm{\Psi};\mathcal{R},\mathcal{T},\bm{x})\right)\Bigg],
\end{align}

\begin{align}
P_{\phi\phi}(\bm{\Psi};\mathcal{R},\mathcal{T},\bm{x})
&=\int_{\mathcal{U}}W(y;\bm{x})\left\{\frac{\partial^2}{\partial\phi^{*2}}\log f^*_{\mathcal{T}}(y;\bm{x},\bm{\Psi})\right\}f_{\mathcal{T}}(y;\bm{x},\bm{\Psi})dy\nonumber\\
&=\frac{\lambda^*(\bm{\Psi};\bm{x})}{F(\mathcal{T};\bm{x},\bm{\Psi})}\bigg[\int_{\mathcal{U}}\left[\frac{2}{\phi^{*3}}(\theta^*y-A(\theta^*)+g(y))+b''(\phi^*)\right]f^*(y;\bm{x},\bm{\Psi})dy\nonumber\\
&\hspace{8em}-\left(\frac{\partial^2}{\partial\phi^{*2}}\log F^*(\mathcal{T};\bm{x},\bm{\Psi})\right)F^*(\mathcal{U};\bm{x},\bm{\Psi})\bigg]\nonumber\\
&=\frac{\lambda^*(\bm{\Psi};\bm{x})}{F(\mathcal{T};\bm{x},\bm{\Psi})}\bigg[\frac{2}{\phi^{*3}}D^*_{\phi}(\mathcal{U};\bm{\Psi},\bm{x})+\left[\frac{2}{\phi^*}b'(\phi^*)+b''(\phi^*)-\frac{\partial^2}{\partial\phi^{*2}}\log F^*(\mathcal{T};\bm{x},\bm{\Psi})\right]F^*(\mathcal{U};\bm{x},\bm{\Psi})\bigg]
\end{align}
\begin{align}
Q_{\phi\phi}(\bm{\Psi};\mathcal{R},\mathcal{T},\bm{x})
&=\sum_{m=1}^{M}\left\{\frac{\partial}{\partial\phi^*}\left\{\lambda^*(\bm{\Psi};\bm{x})\frac{F^*(\mathcal{I}_{m};\bm{x},\bm{\Psi})}{F(\mathcal{I}_{m};\bm{x},\bm{\Psi})}\frac{\partial}{\partial\phi^*} \log F^*_{\mathcal{T}}(\mathcal{I}_{m};\bm{x},\bm{\Psi})\right\}\right\}F_{\mathcal{T}}(\mathcal{I}_{m};\bm{x},\bm{\Psi}),
\end{align}
where $D^*_{\phi}(\mathcal{U};\bm{x},\bm{\Psi})$ is simply $D_{\phi}(\mathcal{U};\bm{\Psi})$ in Equation (\ref{apx:lemma:Dp}) of Lemma \ref{apx:lemma:D} evaluated at $\theta^*$, $\phi^*$ and $\bm{x}$. Second, we write
\begin{align}
\tilde{\Lambda}(\bm{\Psi};\mathcal{R},\mathcal{T},\bm{x}):=
\begin{pmatrix}
\tilde{\Lambda}_{\theta\theta}(\bm{\Psi};\mathcal{R},\mathcal{T},\bm{x}) & \tilde{\Lambda}_{\theta\phi}(\bm{\Psi};\mathcal{R},\mathcal{T},\bm{x})  \\
\tilde{\Lambda}_{\phi\theta}(\bm{\Psi};\mathcal{R},\mathcal{T},\bm{x})  & \tilde{\Lambda}_{\phi\phi}(\bm{\Psi};\mathcal{R},\mathcal{T},\bm{x})
\end{pmatrix},
\end{align}
where the four elements above are expressed as follows.
\begin{align}
\tilde{\Lambda}_{\theta\theta}(\bm{\Psi};\mathcal{R},\mathcal{T},\bm{x})
&=\int_{\mathcal{U}}W(y;\bm{x})^2\left(\frac{\partial}{\partial\bm{\beta}}\log f_{\mathcal{T}}^{*}(y;\bm{x},\bm{\Psi})\right)\left(\frac{\partial}{\partial\bm{\beta}}\log f_{\mathcal{T}}^{*}(y;\bm{x},\bm{\Psi})\right)^Tf_{\mathcal{T}}(y;\bm{x},\bm{\Psi})dy\nonumber\\
&\qquad +\sum_{m=1}^{M}\frac{(\lambda^*(\bm{\Psi};\bm{x}))^2}{F(\mathcal{I}_{m};\bm{x},\bm{\Psi})^2}\left(\frac{\partial}{\partial\bm{\beta}} F^*_{\mathcal{T}}(\mathcal{I}_{m};\bm{x},\bm{\Psi})\right)\left(\frac{\partial}{\partial\bm{\beta}} F^*_{\mathcal{T}}(\mathcal{I}_{m};\bm{x},\bm{\Psi})\right)^TF_{\mathcal{T}}(\mathcal{I}_{m};\bm{x},\bm{\Psi})\nonumber\\
&=\int_{\mathcal{U}}W(y;\bm{x})^2\left(\frac{\partial}{\partial\theta^*}\log f_{\mathcal{T}}^{*}(y;\bm{x},\bm{\Psi})\right)^2f_{\mathcal{T}}(y;\bm{x},\bm{\Psi})dy\left(\frac{\partial\theta^*}{\partial\bm{\beta}}\right)\left(\frac{\partial\theta^*}{\partial\bm{\beta}}\right)^T\nonumber\\
&\qquad +\sum_{m=1}^{M}\left(\lambda^*(\bm{\Psi};\bm{x})\frac{F^*(\mathcal{I}_{m};\bm{x},\bm{\Psi})}{F(\mathcal{I}_{m};\bm{x},\bm{\Psi})}\right)^2\left(\frac{\partial}{\partial\theta^*} \log F^*_{\mathcal{T}}(\mathcal{I}_{m};\bm{x},\bm{\Psi})\right)^2F_{\mathcal{T}}(\mathcal{I}_{m};\bm{x},\bm{\Psi})\left(\frac{\partial\theta^*}{\partial\bm{\beta}}\right)\left(\frac{\partial\theta^*}{\partial\bm{\beta}}\right)^T\nonumber\\
&:=(R_{\theta\theta}(\bm{\Psi};\mathcal{R},\mathcal{T},\bm{x})+S_{\theta\theta}(\bm{\Psi};\mathcal{R},\mathcal{T},\bm{x}))\left(\frac{\partial\theta^*}{\partial\bm{\beta}}\right)\left(\frac{\partial\theta^*}{\partial\bm{\beta}}\right)^T
:=V_{\theta\theta}(\bm{\Psi};\mathcal{R},\mathcal{T},\bm{x})\bm{x}\bm{x}^T,
\end{align}
with
\begin{align} \label{apx:asymp:v_tt}
V_{\theta\theta}(\bm{\Psi};\mathcal{R},\mathcal{T},\bm{x})=\frac{\phi^{*2}}{\phi^2}\left(\xi'(\bm{x}^T\bm{\beta})\right)^2\left(R_{\theta\theta}(\bm{\Psi};\mathcal{R},\mathcal{T},\bm{x})+S_{\theta\theta}(\bm{\Psi};\mathcal{R},\mathcal{T},\bm{x})\right),
\end{align}
\begin{align}
R_{\theta\theta}(\bm{\Psi};\mathcal{R},\mathcal{T},\bm{x})
&=\int_{\mathcal{U}}W(y;\bm{x})^2\left(\frac{\partial}{\partial\theta^*}\log f_{\mathcal{T}}^{*}(y;\bm{x},\bm{\Psi})\right)^2f_{\mathcal{T}}(y;\bm{x},\bm{\Psi})dy\nonumber\\
&=\frac{\lambda^{**}(\bm{\Psi};\bm{x})}{F(\mathcal{T};\bm{x},\bm{\Psi})}\int_{\mathcal{U}}\left\{\frac{1}{\phi^{*}}(y-A'(\theta^*))-\frac{\partial}{\partial\theta^*}\log F^*(\mathcal{T};\bm{x},\bm{\Psi})\right\}^2f^{**}(y;\bm{x},\bm{\Psi})dy\nonumber\\
&=\frac{\lambda^{**}(\bm{\Psi};\bm{x})}{F(\mathcal{T};\bm{x},\bm{\Psi})}\frac{1}{\phi^{*2}}\int_{\mathcal{U}}\left[(y-A'(\theta^{**}))+U_{\theta\theta}(\bm{\Psi};\mathcal{R},\mathcal{T},\bm{x})\right]^2f^{**}(y;\bm{x},\bm{\Psi})dy\nonumber\\
&=\frac{\lambda^{**}(\bm{\Psi};\bm{x})}{F(\mathcal{T};\bm{x},\bm{\Psi})}\frac{1}{\phi^{*2}}
\left[D^{**}_{\theta\theta}(\mathcal{U};\bm{x},\bm{\Psi})+2U_{\theta\theta}(\bm{\Psi};\mathcal{R},\mathcal{T},\bm{x})D^{**}_{\theta}(\mathcal{U};\bm{x},\bm{\Psi})+U_{\theta\theta}(\bm{\Psi};\mathcal{R},\mathcal{T},\bm{x})^2F^{**}(\mathcal{U};\bm{x},\bm{\Psi})\right]
\end{align}
such that $D^{**}_{\theta\theta}(\mathcal{U};\bm{x},\bm{\Psi})$ and $D^{**}_{\theta}(\mathcal{U};\bm{x},\bm{\Psi})$ are $D_{\theta\theta}(\mathcal{U};\bm{\Psi})$ and $D_{\theta}(\mathcal{U};\bm{\Psi})$ in Equations (\ref{apx:lemma:Dtt}) and (\ref{apx:lemma:Dt}) of Lemma \ref{apx:lemma:D} evaluated at $\theta^{**}$ and $\phi^{**}$ with covariates $\bm{x}$, $U_{\theta\theta}(\bm{\Psi};\mathcal{R},\mathcal{T},\bm{x})$ above is given by
\begin{align}
U_{\theta\theta}(\bm{\Psi};\mathcal{R},\mathcal{T},\bm{x})
:=-\left(A'(\theta^*)-A'(\theta^{**})\right)-\phi^*\frac{\partial}{\partial\theta^*}\log F^*(\mathcal{T};\bm{x},\bm{\Psi}),
\end{align}
and $S_{\theta\theta}(\bm{\Psi};\mathcal{R},\mathcal{T},\bm{x})$ in Equation (\ref{apx:asymp:v_tt}) is given by
\begin{align}
S_{\theta\theta}(\bm{\Psi};\mathcal{R},\mathcal{T},\bm{x})
&=\sum_{m=1}^{M}\left(\lambda^*(\bm{\Psi};\bm{x})\frac{F^*(\mathcal{I}_{m};\bm{x},\bm{\Psi})}{F(\mathcal{I}_{m};\bm{x},\bm{\Psi})}\right)^2\left(\frac{\partial}{\partial\theta^*} \log F^*_{\mathcal{T}}(\mathcal{I}_{m};\bm{x},\bm{\Psi})\right)^2F_{\mathcal{T}}(\mathcal{I}_{m};\bm{x},\bm{\Psi}).
\end{align}

Similarly, the other elements are expressed as
\begin{align}
\tilde{\Lambda}_{\theta\phi}(\bm{\Psi};\mathcal{R},\mathcal{T},\bm{x})
&=\left(R_{\theta\theta}(\bm{\Psi};\mathcal{R},\mathcal{T},\bm{x})+S_{\theta\theta}(\bm{\Psi};\mathcal{R},\mathcal{T},\bm{x})\right)\left(\frac{\partial\theta^*}{\partial\phi}\right)\left(\frac{\partial\theta^*}{\partial\bm{\beta}}\right)\nonumber\\
&\quad + \left(R_{\theta\phi}(\bm{\Psi};\mathcal{R},\mathcal{T},\bm{x})+S_{\theta\phi}(\bm{\Psi};\mathcal{R},\mathcal{T},\bm{x})\right)\left(\frac{\partial\phi^*}{\partial\phi}\right)\left(\frac{\partial\theta^*}{\partial\bm{\beta}}\right)\nonumber\\
&:=V_{\theta\phi}(\bm{\Psi};\mathcal{R},\mathcal{T},\bm{x})\bm{x},
\end{align}
with
\begin{align}
V_{\theta\phi}(\bm{\Psi};\mathcal{R},\mathcal{T},\bm{x})
&=\frac{\phi^{*3}}{\phi^3}\xi'(\bm{x}^T\bm{\beta})\Bigg[\left(\left(c-\frac{1}{\tilde{\phi}}\right)\theta+\frac{\tilde{\theta}}{\tilde{\phi}}\right)\left(R_{\theta\theta}(\bm{\Psi};\mathcal{R},\mathcal{T},\bm{x})+S_{\theta\theta}(\bm{\Psi};\mathcal{R},\mathcal{T},\bm{x})\right)\nonumber\\
&\hspace{15em}+\left(R_{\theta\phi}(\bm{\Psi};\mathcal{R},\mathcal{T},\bm{x})+S_{\theta\phi}(\bm{\Psi};\mathcal{R},\mathcal{T},\bm{x})\right)\Bigg],
\end{align}
\begin{align}
R_{\theta\phi}(\bm{\Psi};\mathcal{R},\mathcal{T},\bm{x})
&=\int_{\mathcal{U}}W(y;\bm{x})^2\left(\frac{\partial}{\partial\theta^*}\log f_{\mathcal{T}}^{*}(y;\bm{x},\bm{\Psi})\right)\left(\frac{\partial}{\partial\phi^*}\log f_{\mathcal{T}}^{*}(y;\bm{x},\bm{\Psi})\right)f_{\mathcal{T}}(y;\bm{x},\bm{\Psi})dy\nonumber\\
&=-\frac{\lambda^{**}(\bm{\Psi};\bm{x})}{F(\mathcal{T};\bm{x},\bm{\Psi})}\frac{1}{\phi^{*3}}\int_{\mathcal{U}}\left[(y-A'(\theta^{**}))+U_{\theta\theta}(\bm{\Psi};\mathcal{R},\mathcal{T},\bm{x})\right]\nonumber\\
&\hspace{10em}\times[(\theta^{**}y-A(\theta^{**})+g(y)-\phi^{**2}b'(\phi^{**}))+(\theta^*-\theta^{**})(y-A'(\theta^{**}))\nonumber\\
&\hspace{12em}+U_{\phi\phi}(\bm{\Psi};\mathcal{R},\mathcal{T},\bm{x})] f^{**}(y;\bm{x},\bm{\Psi})dy\nonumber\\
&=-\frac{\lambda^{**}(\bm{\Psi};\bm{x})}{F(\mathcal{T};\bm{x},\bm{\Psi})}\frac{1}{\phi^{*3}}
\Big\{D^{**}_{\theta\phi}(\mathcal{U};\bm{x},\bm{\Psi})
+U_{\theta\theta}(\bm{\Psi};\mathcal{R},\mathcal{T},\bm{x})D^{**}_{\phi}(\mathcal{U};\bm{x},\bm{\Psi})
+(\theta^*-\theta^{**})D^{**}_{\theta\theta}(\mathcal{U};\bm{x},\bm{\Psi})\nonumber\\
&\hspace{10em}+[(\theta^*-\theta^{**})U_{\theta\theta}(\bm{\Psi};\mathcal{R},\mathcal{T},\bm{x})+U_{\phi\phi}(\bm{\Psi};\mathcal{R},\mathcal{T},\bm{x})]D^{**}_{\theta}(\mathcal{U};\bm{x},\bm{\Psi})\nonumber\\
&\hspace{10em}+U_{\theta\theta}(\bm{\Psi};\mathcal{R},\mathcal{T},\bm{x})U_{\phi\phi}(\bm{\Psi};\mathcal{R},\mathcal{T},\bm{x})F^{**}(\mathcal{U};\bm{x},\bm{\Psi})\Big\},
\end{align}
\begin{align}
S_{\theta\phi}(\bm{\Psi};\mathcal{R},\mathcal{T},\bm{x})
&=\sum_{m=1}^{M}\left(\lambda^*(\bm{\Psi};\bm{x})\frac{F^*(\mathcal{I}_{m};\bm{x},\bm{\Psi})}{F(\mathcal{I}_{m};\bm{x},\bm{\Psi})}\right)^2\left(\frac{\partial}{\partial\theta^*} \log F^*_{\mathcal{T}}(\mathcal{I}_{m};\bm{x},\bm{\Psi})\right)\left(\frac{\partial}{\partial\phi^*} \log F^*_{\mathcal{T}}(\mathcal{I}_{m};\bm{x},\bm{\Psi})\right)F_{\mathcal{T}}(\mathcal{I}_{m};\bm{x},\bm{\Psi}),
\end{align}
where $U_{\phi\phi}(\bm{\Psi};\mathcal{R},\mathcal{T},\bm{x})$ is evaluated as
\begin{align}
U_{\phi\phi}(\bm{\Psi};\mathcal{R},\mathcal{T},\bm{x})
=(\theta^*-\theta^{**})A'(\theta^{**})-(A(\theta^*)-A(\theta^{**}))-(\phi^{*2}b'(\phi^*)-\phi^{**2}b'(\phi^{**}))+\phi^{*2}\frac{\partial}{\partial\phi^*}\log F^*(\mathcal{T};\bm{x},\bm{\Psi}),
\end{align}
and $D^{**}_{\theta\phi}(\mathcal{U};\bm{x},\bm{\Psi})$ and $D^{**}_{\phi}(\mathcal{U};\bm{x},\bm{\Psi})$ are $D_{\theta\phi}(\mathcal{U};\bm{\Psi})$ and $D_{\phi}(\mathcal{U};\bm{\Psi})$ in Equations (\ref{apx:lemma:Dtp}) and (\ref{apx:lemma:Dp}) of Lemma \ref{apx:lemma:D} evaluated at $\theta^{**}$ and $\phi^{**}$ with covariates $\bm{x}$.

\begin{align}
\tilde{\Lambda}_{\phi\theta}(\bm{\Psi};\mathcal{R},\mathcal{T},\bm{x})=\tilde{\Lambda}_{\theta\phi}(\bm{\Psi};\mathcal{R},\mathcal{T},\bm{x})^T.
\end{align}

\begin{align}
\tilde{\Lambda}_{\phi\phi}(\bm{\Psi};\mathcal{R},\mathcal{T},\bm{x})
&:=V_{\phi\phi}(\bm{\Psi};\mathcal{R},\mathcal{T},\bm{x}),
\end{align}
with
\begin{align}
V_{\phi\phi}(\bm{\Psi};\mathcal{R},\mathcal{T},\bm{x})
&=\frac{\phi^{*4}}{\phi^4}\Bigg[\left(\left(c-\frac{1}{\tilde{\phi}}\right)\theta+\frac{\tilde{\theta}}{\tilde{\phi}}\right)^2\left(R_{\theta\theta}(\bm{\Psi};\mathcal{R},\mathcal{T},\bm{x})+S_{\theta\theta}(\bm{\Psi};\mathcal{R},\mathcal{T},\bm{x})\right)\nonumber\\
&\hspace{3em}+2\left(\left(c-\frac{1}{\tilde{\phi}}\right)\theta+\frac{\tilde{\theta}}{\tilde{\phi}}\right)\left(R_{\theta\phi}(\bm{\Psi};\mathcal{R},\mathcal{T},\bm{x})+S_{\theta\phi}(\bm{\Psi};\mathcal{R},\mathcal{T},\bm{x})\right)\nonumber\\
&\hspace{11em}+\left(R_{\phi\phi}(\bm{\Psi};\mathcal{R},\mathcal{T},\bm{x})+S_{\phi\phi}(\bm{\Psi};\mathcal{R},\mathcal{T},\bm{x})\right)\Bigg],
\end{align}
\begin{align}
R_{\phi\phi}(\bm{\Psi};\mathcal{R},\mathcal{T},\bm{x})
&=\int_{\mathcal{U}}W(y;\bm{x})^2\left(\frac{\partial}{\partial\phi^*}\log f_{\mathcal{T}}^{*}(y;\bm{x},\bm{\Psi})\right)^2f_{\mathcal{T}}(y;\bm{x},\bm{\Psi})dy\nonumber\\
&=\frac{\lambda^{**}(\bm{\Psi};\bm{x})}{F(\mathcal{T};\bm{x},\bm{\Psi})}\frac{1}{\phi^{*4}}\int_{\mathcal{U}}\Big[(\theta^{**}y-A(\theta^{**})+g(y)-\phi^{**2}b'(\phi^{**}))+(\theta^*-\theta^{**})(y-A'(\theta^{**}))\nonumber\\
&\hspace{12em}+U_{\phi\phi}(\bm{\Psi};\mathcal{R},\mathcal{T},\bm{x})\Big]^2 f^{**}(y;\bm{x},\bm{\Psi})dy\nonumber\\
&=\frac{\lambda^{**}(\bm{\Psi};\bm{x})}{F(\mathcal{T};\bm{x},\bm{\Psi})}\frac{1}{\phi^{*4}}
\Big\{D^{**}_{\phi\phi}(\mathcal{U};\bm{x},\bm{\Psi})
+(\theta^*-\theta^{**})^2D^{**}_{\theta\theta}(\mathcal{U};\bm{x},\bm{\Psi})
+2(\theta^*-\theta^{**})D^{**}_{\theta\phi}(\mathcal{U};\bm{x},\bm{\Psi})\nonumber\\
&\hspace{8em}+2U_{\phi\phi}(\bm{\Psi};\mathcal{R},\mathcal{T},\bm{x})D^{**}_{\phi}(\mathcal{U};\bm{x},\bm{\Psi})
+2(\theta^*-\theta^{**})U_{\phi\phi}(\bm{\Psi};\mathcal{R},\mathcal{T},\bm{x})D^{**}_{\theta}(\mathcal{U};\bm{x},\bm{\Psi})\nonumber\\
&\hspace{8em}+U_{\phi\phi}(\bm{\Psi};\mathcal{R},\mathcal{T},\bm{x})^2F^{**}(\mathcal{U};\bm{x},\bm{\Psi})\Big\},
\end{align}
\begin{align} \label{apx:asymp:s_pp}
S_{\phi\phi}(\bm{\Psi};\mathcal{R},\mathcal{T},\bm{x})
&=\sum_{m=1}^{M}\left(\lambda^*(\bm{\Psi};\bm{x})\frac{F^*(\mathcal{I}_{m};\bm{x},\bm{\Psi})}{F(\mathcal{I}_{m};\bm{x},\bm{\Psi})}\right)^2\left(\frac{\partial}{\partial\phi^*} \log F^*_{\mathcal{T}}(\mathcal{I}_{m};\bm{x},\bm{\Psi})\right)^2F_{\mathcal{T}}(\mathcal{I}_{m};\bm{x},\bm{\Psi}),
\end{align}
where $D^{**}_{\phi\phi}(\mathcal{U};\bm{x},\bm{\Psi})$ is $D_{\phi\phi}(\mathcal{U};\bm{\Psi})$ in Equation (\ref{apx:lemma:Dpp}) of Lemma \ref{apx:lemma:D} evaluated at $\theta^{**}$ and $\phi^{**}$ with covariates $\bm{x}$.

\subsection{Proof of Theorem \ref{thm:asymp}} \label{apx:thm:asymp}
Note that Theorem \ref{thm:asymp} is a special case of Theorem \ref{thm:censtrun:asymp} with $M=0$, $\mathcal{U}=\mathcal{T}=\mathbb{R}$ and $\mathcal{D}=Y$. As a result, the terms defined in the previous subsection related to censored data $Q_{\theta\theta}(\bm{\Psi};\mathcal{R},\mathcal{T},\bm{x})$, $Q_{\theta\phi}(\bm{\Psi};\mathcal{R},\mathcal{T},\bm{x})$, $Q_{\phi\phi}(\bm{\Psi};\mathcal{R},\mathcal{T},\bm{x})$, $S_{\theta\theta}(\bm{\Psi};\mathcal{R},\mathcal{T},\bm{x})$, $S_{\theta\phi}(\bm{\Psi};\mathcal{R},\mathcal{T},\bm{x})$, $S_{\phi\phi}(\bm{\Psi};\mathcal{R},\mathcal{T},\bm{x})$ all equal to zero. Moreover, the distribution functions $F(\mathcal{U};\bm{x},\bm{\Psi})$, $F(\mathcal{T};\bm{x},\bm{\Psi})$, $F^*(\mathcal{U};\bm{x},\bm{\Psi})$, $F^*(\mathcal{T};\bm{x},\bm{\Psi})$,
$F^{**}(\mathcal{U};\bm{x},\bm{\Psi})$, $F^{**}(\mathcal{T};\bm{x},\bm{\Psi})$ all equal to $1$, and hence their derivatives w.r.t. any parameters are zero. Plugging these numbers to Equations (\ref{apx:lemma:Dt}) to (\ref{apx:asymp:s_pp}), we easily obtain the results stated in the theorem.

\subsection{Proof of Theorem \ref{thm:censtrun:diag}} \label{apx:thm:censtrun:diag}
Define a meta extended SWLE score function $\mathcal{S}_n^{\text{meta}}(\bm{\Psi};\mathcal{D},\bm{X})$ as
\begin{align} \label{apx:diag:swle}
\mathcal{S}_n^{\text{meta}}(\bm{\Psi}^{\text{meta}};\mathcal{D},\bm{X})=
\begin{pmatrix}
\mathcal{S}_n^{(1)}(\bm{\Psi}^{(1)};\mathcal{D},\bm{X})\\
\vdots\\
\mathcal{S}_n^{(K)}(\bm{\Psi}^{(K)};\mathcal{D},\bm{X})
\end{pmatrix},
\end{align}
where $\mathcal{S}_n^{(k)}(\bm{\Psi};\mathcal{D},\bm{X})$ is defined as the extended SWLE score function in Equation (\ref{eq:censtrun:swle}) evaluated at weight function hyperparameters $\tilde{\bm{\Psi}}^{(k)}$, for $k=1,\ldots,K$. Also, $\bm{\Psi}^{\text{meta}}=(\bm{\Psi}^{(1)},\ldots,\bm{\Psi}^{(K)})$ is a collection of $K$ sets of parameters for the $K$ individual SWLE score functions. Then, it is obvious that $\hat{\bm{\Psi}}_n^{\text{meta}}$ is the solution of $\mathcal{S}_n^{\text{meta}}(\hat{\bm{\Psi}}^{\text{meta}};\mathcal{D},\bm{X})=\bm{0}$. Correspondingly, define the meta individual score function $\mathcal{S}^{\text{meta}}(\bm{\Psi};\mathcal{D},\bm{x})$ as
\begin{align} \label{apx:diag:swle_ind}
\mathcal{S}^{\text{meta}}(\bm{\Psi}^{\text{meta}};\mathcal{D},\bm{x})=
\begin{pmatrix}
\mathcal{S}^{(1)}(\bm{\Psi}^{(1)};\mathcal{D},\bm{x})\\
\vdots\\
\mathcal{S}^{(K)}(\bm{\Psi}^{(K)};\mathcal{D},\bm{x})
\end{pmatrix},
\end{align}
with $\mathcal{S}^{(k)}(\bm{\Psi};\mathcal{D},\bm{x})$ being the individual score function evaluated as hyperparameters $\tilde{\bm{\Psi}}^{(k)}$. Again, with a slight abuse of notations, $\mathcal{D}$ in Equation (\ref{apx:diag:swle_ind}) is simply $\mathcal{D}_i$ with subscript $i$ dropped, as opposed to Equation (\ref{apx:diag:swle}) where $\mathcal{D}$ represent observed information across all losses. Applying Theorems 5.41 and 5.42 of \cite{van2000asymptotic}, consistency is resulted from Equation (\ref{apx:asymp:consis}), which shows that $E_{\mathcal{D},\bm{x}}[\mathcal{S}^{(k)}(\bm{\Psi}_0;\mathcal{D},\bm{x})]=\bm{0}$ for every $k=1,\ldots ,K$, and hence $E_{\mathcal{D},\bm{x}}[\mathcal{S}^{\text{meta}}(\bm{\Psi}_0^{\text{meta}};\mathcal{D},\bm{x})]=\bm{0}$. 

Before proving asymptotic normality, we define the following notations. First, denote $W^{(k)}(y;\bm{x})$ as the weight function in Equation (\ref{eq:thm:weight}) evaluated at hyperparameters $\tilde{\bm{\Psi}}^{(k)}$. Then, denote $\phi^{(k)}=(\phi^{-1}+(\tilde{\phi}^{(k)})^{-1}-c)^{-1}$, $\theta^{(k)}=(\theta/\phi+\tilde{\theta}^{(k)}/\tilde{\phi}^{(k)})\phi^{(k)}$, $\phi^{(k,k')}=(\phi^{-1}+(\tilde{\phi}^{(k)})^{-1}+(\tilde{\phi}^{(k')})^{-1}-2c)^{-1}$ and $\theta^{(k,k')}=(\theta/\phi+\tilde{\theta}^{(k)}/\tilde{\phi}^{(k)}+\tilde{\theta}^{(k')}/\tilde{\phi}^{(k')})\phi^{(k,k')}$ with $\tilde{\theta}^{(k)}=\xi(\bm{x}^T\tilde{\bm{\beta}}^{(k)})$, we further define corresponding bias adjustment terms and transformed density functions
\begin{align}
\lambda^{(k)}(\bm{\Psi};\bm{x})=\exp\left\{\frac{A(\theta^{(k)})}{\phi^{(k)}}-b(\phi^{(k)})-\frac{A(\theta)}{\phi}+b(\phi)\right\},
\end{align}
\begin{align}
\lambda^{(k,k')}(\bm{\Psi};\bm{x})=\exp\left\{\frac{A(\theta^{(k,k')})}{\phi^{(k,k')}}-b(\phi^{(k,k')})-\frac{A(\theta)}{\phi}+b(\phi)\right\},
\end{align}
\begin{align}
f^{(k)}(y;\bm{x},\bm{\Psi})=\exp\left\{\frac{\theta^{(k)}y-A(\theta^{(k)})}{\phi^{(k)}}+C(y,\phi^{(k)})\right\},
\end{align}
\begin{align}
f^{(k,k')}(y;\bm{x},\bm{\Psi})=\exp\left\{\frac{\theta^{(k,k')}y-A(\theta^{(k,k')})}{\phi^{(k,k')}}+C(y,\phi^{(k,k')})\right\}.
\end{align}
After that, define $f^{(k)}_{\mathcal{T}}(y;\bm{x},\bm{\Psi})$ and $f^{(k,k')}_{\mathcal{T}}(y;\bm{x},\bm{\Psi})$ as the truncated density functions of $f^{(k)}(y;\bm{x},\bm{\Psi})$ and $f^{(k,k')}(y;\bm{x},\bm{\Psi})$ respectively in the same way as Equation (\ref{eq:censtrun:density}). The corresponding distribution functions is given by $F^{(k)}_{\mathcal{T}}(\cdot;\bm{x},\bm{\Psi})$ and $F^{(k,k')}_{\mathcal{T}}(\cdot;\bm{x},\bm{\Psi})$. Finally, we denote $D^{(k)}_{\theta}(\mathcal{U};\bm{x},\bm{\Psi})$, $D^{(k)}_{\phi}(\mathcal{U};\bm{x},\bm{\Psi})$, $D^{(k)}_{\theta\theta}(\mathcal{U};\bm{x},\bm{\Psi})$, $D^{(k)}_{\theta\phi}(\mathcal{U};\bm{x},\bm{\Psi})$ and $D^{(k)}_{\phi\phi}(\mathcal{U};\bm{x},\bm{\Psi})$ as $D_{\theta}(\mathcal{U};\bm{\Psi})$, $D_{\phi}(\mathcal{U};\bm{\Psi})$, $D_{\theta\theta}(\mathcal{U};\bm{\Psi})$, $D_{\theta\phi}(\mathcal{U};\bm{\Psi})$ and $D_{\phi\phi}(\mathcal{U};\bm{\Psi})$ in Equations (\ref{apx:lemma:Dt}) to (\ref{apx:lemma:Dtp}) evaluated at $\theta^{(k)}$ and $\phi^{(k)}$ with covariates $\bm{x}$, and similarly $D^{(k,k')}_{\theta}(\mathcal{U};\bm{x},\bm{\Psi})$, $D^{(k,k')}_{\phi}(\mathcal{U};\bm{x},\bm{\Psi})$, $D^{(k,k')}_{\theta\theta}(\mathcal{U};\bm{x},\bm{\Psi})$, $D^{(k,k')}_{\theta\phi}(\mathcal{U};\bm{x},\bm{\Psi})$ and $D^{(k,k')}_{\phi\phi}(\mathcal{U};\bm{x},\bm{\Psi})$ as those evaluated at $\theta^{(k,k')}$ and $\phi^{(k,k')}$.

Theorems 5.41 of \cite{van2000asymptotic} shows that $\sqrt{n}\left(\hat{\bm{\Psi}}_n^{\text{meta}}-\bm{\Psi}_0^{\text{meta}}\right)\overset{d}{\rightarrow}\mathcal{N}(\bm{0},\bm{\Sigma}^{\text{meta}})$, with $\bm{\Sigma}^{\text{meta}}=\left({[\Gamma^{\text{meta}}]}^{-1}\right)\Lambda^{\text{meta}}\left({[\Gamma^{\text{meta}}]}^{-1}\right)^T$, where
\begin{align}
\Gamma^{\text{meta}}=
\begin{pmatrix}
\Gamma^{(1,1)} & \Gamma^{(1,2)} & \dots & \Gamma^{(1,K)} \\
\Gamma^{(2,1)} & \Gamma^{(2,2)} & \dots & \Gamma^{(2,K)} \\
\vdots & \vdots & \ddots & \vdots \\
\Gamma^{(K,1)} & \Gamma^{(K,2)} & \dots & \Gamma^{(K,K)}
\end{pmatrix}
\end{align}
and
\begin{align}
\Lambda^{\text{meta}}=
\begin{pmatrix}
\Lambda^{(1,1)} & \Lambda^{(1,2)} & \dots & \Lambda^{(1,K)} \\
\Lambda^{(2,1)} & \Lambda^{(2,2)} & \dots & \Lambda^{(2,K)} \\
\vdots & \vdots & \ddots & \vdots \\
\Lambda^{(K,1)} & \Lambda^{(K,2)} & \dots & \Lambda^{(K,K)}
\end{pmatrix},
\end{align}
with $\Gamma^{(k,k')}$ and $\Lambda^{(k,k')}$ being $(P+1)\times (P+1)$ matrices for $k,k'=1,\ldots,K$ given by
\begin{align}
\Gamma^{(k,k')}:=\Gamma^{(k,k')}(\bm{\Psi}_0)=E_{\mathcal{D},\bm{x}}\left[\frac{\partial}{\partial\bm{\Psi}^{(k)}}\mathcal{S}^{(k')}(\bm{\Psi}^{(k')};\mathcal{D},\bm{x})^T\right]\Bigg|_{\bm{\Psi}^{(k)},\bm{\Psi}^{(k')}=\bm{\Psi}_0},
\end{align}
\begin{align}
\Lambda^{(k,k')}:=\Gamma^{(k,k')}(\bm{\Psi}_0)=E_{\mathcal{D},\bm{x}}\left[\mathcal{S}^{(k)}(\bm{\Psi}^{(k)};\mathcal{D},\bm{x})\mathcal{S}^{(k')}(\bm{\Psi}^{(k')};\mathcal{D},\bm{x})^T\right]\bigg|_{\bm{\Psi}^{(k)},\bm{\Psi}^{(k')}=\bm{\Psi}_0}.
\end{align}

Obviously, $\Gamma^{(k,k')}=0$ for $k\neq k'$. Denote $\Gamma^{(k)}=\Gamma^{(k,k)}$, we have $\Gamma^{\text{meta}}=\text{diag}(\Gamma^{(1)},\ldots,\Gamma^{(K)})$. Evaluating the matrix inverses and products $\bm{\Sigma}^{\text{meta}}=\left({[\Gamma^{\text{meta}}]}^{-1}\right)\Lambda^{\text{meta}}\left({[\Gamma^{\text{meta}}]}^{-1}\right)^T$, we obtain the form of Equation (\ref{eq:diag:sigma}) for $\bm{\Sigma}^{\text{meta}}$. We evaluate each term as follows. First, we have
\begin{align}
\Gamma^{(k)}:=\Gamma^{(k)}(\bm{\Psi}_0)
:=E_{\mathcal{D},\bm{x}}\left[\tilde{\Gamma}^{(k)}(\bm{\Psi}_0;\mathcal{R},\mathcal{T},\bm{x})\right],
\end{align}
where $\tilde{\Gamma}^{(k)}(\bm{\Psi}_0;\mathcal{R},\mathcal{T},\bm{x})$ is simply $\tilde{\Gamma}(\bm{\Psi}_0;\mathcal{R},\mathcal{T},\bm{x})$ in Equation (\ref{eq:apx:gamma}) with weight function hyperparameters selected as $\tilde{\bm{\Psi}}^{(k)}$. Then, we have
\begin{align}
\Lambda^{(k,k')}:=\Lambda^{(k,k')}(\bm{\Psi}_0)
:=E_{\mathcal{D},\bm{x}}\left[\tilde{\Lambda}^{(k,k')}(\bm{\Psi}_0;\mathcal{R},\mathcal{T},\bm{x})\right],
\end{align}
and write
\begin{align}
\tilde{\Lambda}^{(k,k')}(\bm{\Psi};\mathcal{R},\mathcal{T},\bm{x}):=
\begin{pmatrix}
\tilde{\Lambda}^{(k,k')}_{\theta\theta}(\bm{\Psi};\mathcal{R},\mathcal{T},\bm{x}) & \tilde{\Lambda}^{(k,k')}_{\theta\phi}(\bm{\Psi};\mathcal{R},\mathcal{T},\bm{x})  \\
\tilde{\Lambda}^{(k,k')}_{\phi\theta}(\bm{\Psi};\mathcal{R},\mathcal{T},\bm{x})  & \tilde{\Lambda}^{(k,k')}_{\phi\phi}(\bm{\Psi};\mathcal{R},\mathcal{T},\bm{x})
\end{pmatrix},
\end{align}
the four elements above are expressed as follows. First, we derive $\tilde{\Lambda}^{(k,k')}_{\theta\theta}(\bm{\Psi};\mathcal{R},\mathcal{T},\bm{x})$ as

\begin{align} \label{apx:diag:lambda_tt}
\tilde{\Lambda}^{(k,k')}_{\theta\theta}(\bm{\Psi};\mathcal{R},\mathcal{T},\bm{x})
&=\int_{\mathcal{U}}W^{(k)}(y;\bm{x})W^{(k')}(y;\bm{x})\left(\frac{\partial}{\partial\bm{\beta}}\log f_{\mathcal{T}}^{(k)}(y;\bm{x},\bm{\Psi})\right)\left(\frac{\partial}{\partial\bm{\beta}}\log f_{\mathcal{T}}^{(k')}(y;\bm{x},\bm{\Psi})\right)^Tf_{\mathcal{T}}(y;\bm{x},\bm{\Psi})dy\nonumber\\
&\qquad +\sum_{m=1}^{M}\lambda^{(k)}(\bm{\Psi};\bm{x})\lambda^{(k')}(\bm{\Psi};\bm{x})\frac{F^{(k)}(\mathcal{I}_{m};\bm{x},\bm{\Psi})F^{(k')}(\mathcal{I}_{m};\bm{x},\bm{\Psi})}{F(\mathcal{I}_{m};\bm{x},\bm{\Psi})^2}\left(\frac{\partial}{\partial\bm{\beta}} \log F^{(k)}_{\mathcal{T}}(\mathcal{I}_{m};\bm{x},\bm{\Psi})\right)\nonumber\\
&\hspace{10em}\times\left(\frac{\partial}{\partial\bm{\beta}} \log F^{(k')}_{\mathcal{T}}(\mathcal{I}_{m};\bm{x},\bm{\Psi})\right)^TF_{\mathcal{T}}(\mathcal{I}_{m};\bm{x},\bm{\Psi})\nonumber\\
&=\int_{\mathcal{U}}W^{(k)}(y;\bm{x})W^{(k')}(y;\bm{x})\left(\frac{\partial}{\partial\theta^{(k)}}\log f_{\mathcal{T}}^{(k)}(y;\bm{x},\bm{\Psi})\right)\left(\frac{\partial}{\partial\theta^{(k')}}\log f_{\mathcal{T}}^{(k')}(y;\bm{x},\bm{\Psi})\right)\nonumber\\
&\hspace{10em}\times f_{\mathcal{T}}(y;\bm{x},\bm{\Psi})dy\left(\frac{\partial\theta^{(k)}}{\partial\bm{\beta}}\right)\left(\frac{\partial\theta^{(k')}}{\partial\bm{\beta}}\right)^T\nonumber\\
&\qquad +\sum_{m=1}^{M}\lambda^{(k)}(\bm{\Psi};\bm{x})\lambda^{(k')}(\bm{\Psi};\bm{x})\frac{F^{(k)}(\mathcal{I}_{m};\bm{x},\bm{\Psi})F^{(k')}(\mathcal{I}_{m};\bm{x},\bm{\Psi})}{F(\mathcal{I}_{m};\bm{x},\bm{\Psi})^2}\left(\frac{\partial}{\partial\theta^{(k)}} \log F^{(k)}_{\mathcal{T}}(\mathcal{I}_{m};\bm{x},\bm{\Psi})\right)\nonumber\\
&\hspace{10em}\times \left(\frac{\partial}{\partial\theta^{(k')}} \log F^{(k')}_{\mathcal{T}}(\mathcal{I}_{m};\bm{x},\bm{\Psi})\right) \times F_{\mathcal{T}}(\mathcal{I}_{m};\bm{x},\bm{\Psi})\left(\frac{\partial\theta^{(k)}}{\partial\bm{\beta}}\right)\left(\frac{\partial\theta^{(k')}}{\partial\bm{\beta}}\right)^T\nonumber\\
&:=(R^{(k,k')}_{\theta\theta}(\bm{\Psi};\mathcal{R},\mathcal{T},\bm{x})+S^{(k,k')}_{\theta\theta}(\bm{\Psi};\mathcal{R},\mathcal{T},\bm{x}))\left(\frac{\partial\theta^{(k)}}{\partial\bm{\beta}}\right)\left(\frac{\partial\theta^{(k')}}{\partial\bm{\beta}}\right)^T
:=V^{(k,k')}_{\theta\theta}(\bm{\Psi};\mathcal{R},\mathcal{T},\bm{x})\bm{x}\bm{x}^T,
\end{align}
with
\begin{align} \label{apx:diag:v_tt}
V_{\theta\theta}^{(k,k')}(\bm{\Psi};\mathcal{R},\mathcal{T},\bm{x})=\frac{\phi^{(k)}\phi^{(k')}}{\phi^2}\left(\xi'(\bm{x}^T\bm{\beta})\right)^2\left(R^{(k,k')}_{\theta\theta}(\bm{\Psi};\mathcal{R},\mathcal{T},\bm{x})+S^{(k,k')}_{\theta\theta}(\bm{\Psi};\mathcal{R},\mathcal{T},\bm{x})\right),
\end{align}
\begin{align}
R^{(k,k')}_{\theta\theta}(\bm{\Psi};\mathcal{R},\mathcal{T},\bm{x})
&=\int_{\mathcal{U}}W^{(k)}(y;\bm{x})W^{(k')}(y;\bm{x})\left(\frac{\partial}{\partial\theta^{(k)}}\log f_{\mathcal{T}}^{(k)}(y;\bm{x},\bm{\Psi})\right)\left(\frac{\partial}{\partial\theta^{(k')}}\log f_{\mathcal{T}}^{(k')}(y;\bm{x},\bm{\Psi})\right)f_{\mathcal{T}}(y;\bm{x},\bm{\Psi})dy\nonumber\\
&=\frac{\lambda^{(k,k')}(\bm{\Psi};\bm{x})}{F(\mathcal{T};\bm{x},\bm{\Psi})}\int_{\mathcal{U}}\left\{\frac{1}{\phi^{(k)}}(y-A'(\theta^{(k)}))-\frac{\partial}{\partial\theta^{(k)}}\log F^{(k)}(\mathcal{T};\bm{x},\bm{\Psi})\right\}\nonumber\\
&\hspace{10em}\times \left\{\frac{1}{\phi^{(k')}}(y-A'(\theta^{(k')}))-\frac{\partial}{\partial\theta^{(k')}}\log F^{(k')}(\mathcal{T};\bm{x},\bm{\Psi})\right\}f^{(k,k')}(y;\bm{x},\bm{\Psi})dy\nonumber\\
&=\frac{\lambda^{(k,k')}(\bm{\Psi};\bm{x})}{F(\mathcal{T};\bm{x},\bm{\Psi})}\frac{1}{\phi^{(k)}\phi^{(k')}}\int_{\mathcal{U}}\left[(y-A'(\theta^{(k,k')}))+U^{(k,k')}_{\theta\theta}(\bm{\Psi};\mathcal{R},\mathcal{T},\bm{x})\right]\nonumber\\
&\hspace{12em}\times \left[(y-A'(\theta^{(k,k')}))+U^{(k',k)}_{\theta\theta}(\bm{\Psi};\mathcal{R},\mathcal{T},\bm{x})\right]f^{(k,k')}(y;\bm{x},\bm{\Psi})dy\nonumber\\
&=\frac{\lambda^{(k,k')}(\bm{\Psi};\bm{x})}{F(\mathcal{T};\bm{x},\bm{\Psi})}\frac{1}{\phi^{(k)}\phi^{(k')}}
\bigg[D^{(k,k')}_{\theta\theta}(\mathcal{U};\bm{x},\bm{\Psi})+U^{(k,k')}_{\theta\theta}(\bm{\Psi};\mathcal{R},\mathcal{T},\bm{x})D^{(k,k')}_{\theta}(\mathcal{U};\bm{x},\bm{\Psi})\nonumber\\
&\hspace{12em}+U^{(k',k)}_{\theta\theta}(\bm{\Psi};\mathcal{R},\mathcal{T},\bm{x})D^{(k,k')}_{\theta}(\mathcal{U};\bm{x},\bm{\Psi})\nonumber\\
&\hspace{12em}+U^{(k,k')}_{\theta\theta}(\bm{\Psi};\mathcal{R},\mathcal{T},\bm{x})U^{(k',k)}_{\theta\theta}(\bm{\Psi};\mathcal{R},\mathcal{T},\bm{x})F^{(k,k')}(\mathcal{U};\bm{x},\bm{\Psi})\bigg],
\end{align}
\begin{align}
S^{(k,k')}_{\theta\theta}(\bm{\Psi};\mathcal{R},\mathcal{T},\bm{x})
&=\sum_{m=1}^{M}\lambda^{(k)}(\bm{\Psi};\bm{x})\lambda^{(k')}(\bm{\Psi};\bm{x})\frac{F^{(k)}(\mathcal{I}_{m};\bm{x},\bm{\Psi})F^{(k')}(\mathcal{I}_{m};\bm{x},\bm{\Psi})}{F(\mathcal{I}_{m};\bm{x},\bm{\Psi})^2}\nonumber\\
&\hspace{4em}\times\left(\frac{\partial}{\partial\theta^{(k)}} \log F^{(k)}_{\mathcal{T}}(\mathcal{I}_{m};\bm{x},\bm{\Psi})\right)\left(\frac{\partial}{\partial\theta^{(k')}} \log F^{(k')}_{\mathcal{T}}(\mathcal{I}_{m};\bm{x},\bm{\Psi})\right)
F_{\mathcal{T}}(\mathcal{I}_{m};\bm{x},\bm{\Psi}),
\end{align}
\begin{align}
U^{(k,k')}_{\theta\theta}(\bm{\Psi};\mathcal{R},\mathcal{T},\bm{x})
:=-\left(A'(\theta^{(k)})-A'(\theta^{(k,k')})\right)-\phi^{(k)}\frac{\partial}{\partial\theta^{(k)}}\log F^{(k)}(\mathcal{T};\bm{x},\bm{\Psi}).
\end{align}

Second, we derive $\tilde{\Lambda}^{(k,k')}_{\theta\phi}(\bm{\Psi};\mathcal{R},\mathcal{T},\bm{x})$ as
\begin{align}
\tilde{\Lambda}^{(k,k')}_{\theta\phi}(\bm{\Psi};\mathcal{R},\mathcal{T},\bm{x})
&=\left(R^{(k,k')}_{\theta\theta}(\bm{\Psi};\mathcal{R},\mathcal{T},\bm{x})+S^{(k,k')}_{\theta\theta}(\bm{\Psi};\mathcal{R},\mathcal{T},\bm{x})\right)\left(\frac{\partial\theta^{(k)}}{\partial\bm{\beta}}\right)\left(\frac{\partial\theta^{(k')}}{\partial\phi}\right)\nonumber\\
&\quad + \left(R^{(k,k')}_{\theta\phi}(\bm{\Psi};\mathcal{R},\mathcal{T},\bm{x})+S^{(k,k')}_{\theta\phi}(\bm{\Psi};\mathcal{R},\mathcal{T},\bm{x})\right)\left(\frac{\partial\theta^{(k)}}{\partial\bm{\beta}}\right)\left(\frac{\partial\phi^{(k')}}{\partial\phi}\right)\nonumber\\
&:=V^{(k,k')}_{\theta\phi}(\bm{\Psi};\mathcal{R},\mathcal{T},\bm{x})\bm{x},
\end{align}
with
\begin{align}
V^{(k,k')}_{\theta\phi}(\bm{\Psi};\mathcal{R},\mathcal{T},\bm{x})
&=\frac{\phi^{(k)}\phi^{(k')2}}{\phi^3}\xi'(\bm{x}^T\bm{\beta})\Bigg[\left(\left(c-\frac{1}{\tilde{\phi}^{(k')}}\right)\theta+\frac{\tilde{\theta}^{(k')}}{\tilde{\phi}^{(k')}}\right)\left(R^{(k,k')}_{\theta\theta}(\bm{\Psi};\mathcal{R},\mathcal{T},\bm{x})+S^{(k,k')}_{\theta\theta}(\bm{\Psi};\mathcal{R},\mathcal{T},\bm{x})\right)\nonumber\\
&\hspace{15em}+\left(R^{(k,k')}_{\theta\phi}(\bm{\Psi};\mathcal{R},\mathcal{T},\bm{x})+S^{(k,k')}_{\theta\phi}(\bm{\Psi};\mathcal{R},\mathcal{T},\bm{x})\right)\Bigg],
\end{align}
\begin{align}
R^{(k,k')}_{\theta\phi}(\bm{\Psi};\mathcal{R},\mathcal{T},\bm{x})
&=\int_{\mathcal{U}}W^{(k)}(y;\bm{x})W^{(k')}(y;\bm{x})\left(\frac{\partial}{\partial\theta^{(k)}}\log f_{\mathcal{T}}^{(k)}(y;\bm{x},\bm{\Psi})\right)\left(\frac{\partial}{\partial\phi^{(k')}}\log f_{\mathcal{T}}^{(k')}(y;\bm{x},\bm{\Psi})\right)f_{\mathcal{T}}(y;\bm{x},\bm{\Psi})dy\nonumber\\
&=-\frac{\lambda^{(k,k')}(\bm{\Psi};\bm{x})}{F(\mathcal{T};\bm{x},\bm{\Psi})}\frac{1}{\phi^{(k)}\phi^{(k')2}}\int_{\mathcal{U}}\left[(y-A'(\theta^{(k,k')}))+U^{(k,k')}_{\theta\theta}(\bm{\Psi};\mathcal{R},\mathcal{T},\bm{x})\right]\nonumber\\
&\hspace{13em}\times[(\theta^{(k,k')}y-A(\theta^{(k,k')})+g(y)-\phi^{(k,k')2}b'(\phi^{(k,k')}))\nonumber\\
&\hspace{15em}+(\theta^{(k')}-\theta^{(k,k')})(y-A'(\theta^{(k,k')}))\nonumber\\
&\hspace{15em}+U^{(k',k)}_{\phi\phi}(\bm{\Psi};\mathcal{R},\mathcal{T},\bm{x})] f^{(k,k')}(y;\bm{x},\bm{\Psi})dy\nonumber\\
&=-\frac{\lambda^{(k,k')}(\bm{\Psi};\bm{x})}{F(\mathcal{T};\bm{x},\bm{\Psi})}\frac{1}{\phi^{(k)}\phi^{(k')2}}\Big\{D^{(k,k')}_{\theta\phi}(\mathcal{U};\bm{x},\bm{\Psi})
+U^{(k,k')}_{\theta\theta}(\bm{\Psi};\mathcal{R},\mathcal{T},\bm{x})D^{(k,k')}_{\phi}(\mathcal{U};\bm{x},\bm{\Psi})\nonumber\\
&\hspace{13em}+(\theta^{(k')}-\theta^{(k,k')})D^{(k,k')}_{\theta\theta}(\mathcal{U};\bm{x},\bm{\Psi})\nonumber\\
&\hspace{13em}+[(\theta^{(k')}-\theta^{(k,k')})U^{(k,k')}_{\theta\theta}(\bm{\Psi};\mathcal{R},\mathcal{T},\bm{x})+U^{(k',k)}_{\phi\phi}(\bm{\Psi};\mathcal{R},\mathcal{T},\bm{x})]\nonumber\\
&\hspace{15em}\times D^{(k,k')}_{\theta}(\mathcal{U};\bm{x},\bm{\Psi})\nonumber\\
&\hspace{13em}+U^{(k,k')}_{\theta\theta}(\bm{\Psi};\mathcal{R},\mathcal{T},\bm{x})U^{(k',k)}_{\phi\phi}(\bm{\Psi};\mathcal{R},\mathcal{T},\bm{x})F^{(k,k')}(\mathcal{U};\bm{x},\bm{\Psi})\Big\},
\end{align}
\begin{align}
S^{(k,k')}_{\theta\phi}(\bm{\Psi};\mathcal{R},\mathcal{T},\bm{x})
&=\sum_{m=1}^{M}\lambda^{(k)}(\bm{\Psi};\bm{x})\lambda^{(k')}(\bm{\Psi};\bm{x})\frac{F^{(k)}(\mathcal{I}_{m};\bm{x},\bm{\Psi})F^{(k')}(\mathcal{I}_{m};\bm{x},\bm{\Psi})}{F(\mathcal{I}_{m};\bm{x},\bm{\Psi})^2}\nonumber\\
&\hspace{4em}\times\left(\frac{\partial}{\partial\theta^{(k)}} \log F^{(k)}_{\mathcal{T}}(\mathcal{I}_{m};\bm{x},\bm{\Psi})\right)\left(\frac{\partial}{\partial\phi^{(k')}} \log F^{(k')}_{\mathcal{T}}(\mathcal{I}_{m};\bm{x},\bm{\Psi})\right)F_{\mathcal{T}}(\mathcal{I}_{m};\bm{x},\bm{\Psi}),
\end{align}
\begin{align}
U^{(k,k')}_{\phi\phi}(\bm{\Psi};\mathcal{R},\mathcal{T},\bm{x})
&=(\theta^{(k)}-\theta^{(k,k')})A'(\theta^{(k,k')})-(A(\theta^{(k)})-A(\theta^{(k,k')}))\nonumber\\
&\hspace{7em}-(\phi^{(k)2}b'(\phi^{(k)})-\phi^{(k,k')2}b'(\phi^{(k,k')}))+\phi^{(k)2}\frac{\partial}{\partial\phi^{(k)}}\log F^{(k)}(\mathcal{T};\bm{x},\bm{\Psi}).
\end{align}

Third, $\tilde{\Lambda}^{(k,k')}_{\phi\theta}(\bm{\Psi};\mathcal{R},\mathcal{T},\bm{x})$ is given by

\begin{align}
\tilde{\Lambda}^{(k,k')}_{\phi\theta}(\bm{\Psi};\mathcal{R},\mathcal{T},\bm{x})=\tilde{\Lambda}^{(k',k)}_{\theta\phi}(\bm{\Psi};\mathcal{R},\mathcal{T},\bm{x})^T.
\end{align}

Finally, $\tilde{\Lambda}^{(k,k')}_{\phi\phi}(\bm{\Psi};\mathcal{R},\mathcal{T},\bm{x})$ is evaluated as

\begin{align}
\tilde{\Lambda}^{(k,k')}_{\phi\phi}(\bm{\Psi};\mathcal{R},\mathcal{T},\bm{x})
&:=V^{(k,k')}_{\phi\phi}(\bm{\Psi};\mathcal{R},\mathcal{T},\bm{x}),
\end{align}
with
\begin{align}
V^{(k,k')}_{\phi\phi}(\bm{\Psi};\mathcal{R},\mathcal{T},\bm{x})
&=\frac{\phi^{(k)2}\phi^{(k')2}}{\phi^4}\Bigg[\left(\left(c-\frac{1}{\tilde{\phi}^{(k)}}\right)\theta+\frac{\tilde{\theta}^{(k)}}{\tilde{\phi}^{(k)}}\right)\left(\left(c-\frac{1}{\tilde{\phi}^{(k')}}\right)\theta+\frac{\tilde{\theta}^{(k')}}{\tilde{\phi}^{(k')}}\right)\nonumber\\
&\hspace{8em}\times\left(R^{(k,k')}_{\theta\theta}(\bm{\Psi};\mathcal{R},\mathcal{T},\bm{x})+S^{(k,k')}_{\theta\theta}(\bm{\Psi};\mathcal{R},\mathcal{T},\bm{x})\right)\nonumber\\
&\hspace{6em}+\left(\left(c-\frac{1}{\tilde{\phi}^{(k)}}\right)\theta+\frac{\tilde{\theta}^{(k)}}{\tilde{\phi}^{(k)}}\right)
\left(R^{(k,k')}_{\theta\phi}(\bm{\Psi};\mathcal{R},\mathcal{T},\bm{x})+S^{(k,k')}_{\theta\phi}(\bm{\Psi};\mathcal{R},\mathcal{T},\bm{x})\right)\nonumber\\
&\hspace{6em}+\left(\left(c-\frac{1}{\tilde{\phi}^{(k')}}\right)\theta+\frac{\tilde{\theta}^{(k')}}{\tilde{\phi}^{(k')}}\right)
\left(R^{(k',k)}_{\theta\phi}(\bm{\Psi};\mathcal{R},\mathcal{T},\bm{x})+S^{(k',k)}_{\theta\phi}(\bm{\Psi};\mathcal{R},\mathcal{T},\bm{x})\right)\nonumber\\
&\hspace{6em}+\left(R^{(k,k')}_{\phi\phi}(\bm{\Psi};\mathcal{R},\mathcal{T},\bm{x})+S^{(k,k')}_{\phi\phi}(\bm{\Psi};\mathcal{R},\mathcal{T},\bm{x})\right)\Bigg],
\end{align}
\begin{align}
R^{(k',k)}_{\phi\phi}(\bm{\Psi};\mathcal{R},\mathcal{T},\bm{x})
&=\int_{\mathcal{U}}W^{(k)}(y;\bm{x})W^{(k')}(y;\bm{x})\left(\frac{\partial}{\partial\phi^{(k)}}\log f_{\mathcal{T}}^{(k)}(y;\bm{x},\bm{\Psi})\right)\left(\frac{\partial}{\partial\phi^{(k')}}\log f_{\mathcal{T}}^{(k')}(y;\bm{x},\bm{\Psi})\right)f_{\mathcal{T}}(y;\bm{x},\bm{\Psi})dy\nonumber\\
&=\frac{\lambda^{(k,k')}(\bm{\Psi};\bm{x})}{F(\mathcal{T};\bm{x},\bm{\Psi})}\frac{1}{\phi^{(k)2}\phi^{(k')2}}\int_{\mathcal{U}}\Big[(\theta^{(k,k')}y-A(\theta^{(k,k')})+g(y)-\phi^{(k,k')2}b'(\phi^{(k,k')}))\nonumber\\
&\hspace{14em}+(\theta^{(k)}-\theta^{(k,k')})(y-A'(\theta^{(k,k')}))
+U^{(k,k')}_{\phi\phi}(\bm{\Psi};\mathcal{R},\mathcal{T},\bm{x})\Big] \nonumber\\
&\hspace{13em}\times\Big[(\theta^{(k,k')}y-A(\theta^{(k,k')})+g(y)-\phi^{(k,k')2}b'(\phi^{(k,k')}))\nonumber\\
&\hspace{14em}+(\theta^{(k')}-\theta^{(k,k')})(y-A'(\theta^{(k,k')}))
+U^{(k',k)}_{\phi\phi}(\bm{\Psi};\mathcal{R},\mathcal{T},\bm{x})\Big] \nonumber\\
&\hspace{13em}\times f^{(k,k')}(y;\bm{x},\bm{\Psi})dy\nonumber\\
&=\frac{\lambda^{(k,k')}(\bm{\Psi};\bm{x})}{F(\mathcal{T};\bm{x},\bm{\Psi})}\frac{1}{\phi^{(k)2}\phi^{(k')2}}
\Big\{D^{(k,k')}_{\phi\phi}(\mathcal{U};\bm{x},\bm{\Psi})
+(\theta^{(k)}-\theta^{(k,k')})(\theta^{(k')}-\theta^{(k,k')})D^{(k,k')}_{\theta\theta}(\mathcal{U};\bm{x},\bm{\Psi})\nonumber\\
&\hspace{12em}+\left[(\theta^{(k)}-\theta^{(k,k')})+(\theta^{(k')}-\theta^{(k,k')})\right]D^{(k,k')}_{\theta\phi}(\mathcal{U};\bm{x},\bm{\Psi})\nonumber\\
&\hspace{12em}+\left[U^{(k,k')}_{\phi\phi}(\bm{\Psi};\mathcal{R},\mathcal{T},\bm{x})+U^{(k',k)}_{\phi\phi}(\bm{\Psi};\mathcal{R},\mathcal{T},\bm{x})\right]D^{(k,k')}_{\phi}(\mathcal{U};\bm{x},\bm{\Psi})\nonumber\\
&\hspace{12em}+(\theta^{(k')}-\theta^{(k,k')})U^{(k,k')}_{\phi\phi}(\bm{\Psi};\mathcal{R},\mathcal{T},\bm{x})D^{(k,k')}_{\theta}(\mathcal{U};\bm{x},\bm{\Psi})\nonumber\\
&\hspace{12em}+(\theta^{(k)}-\theta^{(k,k')})U^{(k',k)}_{\phi\phi}(\bm{\Psi};\mathcal{R},\mathcal{T},\bm{x})D^{(k,k')}_{\theta}(\mathcal{U};\bm{x},\bm{\Psi})\nonumber\\
&\hspace{12em}+U^{(k,k')}_{\phi\phi}(\bm{\Psi};\mathcal{R},\mathcal{T},\bm{x})U^{(k',k)}_{\phi\phi}(\bm{\Psi};\mathcal{R},\mathcal{T},\bm{x})F^{(k,k')}(\mathcal{U};\bm{x},\bm{\Psi})\Big\},
\end{align}
\begin{align} \label{apx:diag:s_pp}
S^{(k,k')}_{\phi\phi}(\bm{\Psi};\mathcal{R},\mathcal{T},\bm{x})
&=\sum_{m=1}^{M}\lambda^{(k)}(\bm{\Psi};\bm{x})\lambda^{(k')}(\bm{\Psi};\bm{x})\frac{F^{(k)}(\mathcal{I}_{m};\bm{x},\bm{\Psi})F^{(k')}(\mathcal{I}_{m};\bm{x},\bm{\Psi})}{F(\mathcal{I}_{m};\bm{x},\bm{\Psi})^2}\nonumber\\
&\hspace{4em}\times\left(\frac{\partial}{\partial\phi^{(k)}} \log F^{(k)}_{\mathcal{T}}(\mathcal{I}_{m};\bm{x},\bm{\Psi})\right)\left(\frac{\partial}{\partial\phi^{(k')}} \log F^{(k')}_{\mathcal{T}}(\mathcal{I}_{m};\bm{x},\bm{\Psi})\right)F_{\mathcal{T}}(\mathcal{I}_{m};\bm{x},\bm{\Psi}),
\end{align}

Finally, Equation (\ref{eq:censtrun:thm_wald}) is an ordinary Wald statistic resulted from Equation (\ref{eq:censtrun:thm_asymp}).

\subsection{Proof of Theorems \ref{thm:diag:asymp} and \ref{thm:diag:chisq}} \label{apx:thm:diag:chisq}
Theorems \ref{thm:diag:asymp} and \ref{thm:diag:chisq} combined represent a special case of Theorem \ref{thm:censtrun:diag} with $M=0$, $\mathcal{U}=\mathcal{T}=\mathbb{R}$ and $\mathcal{D}=Y$.

\subsection{Derivations of Equations (\ref{eq:data_wgt_ratio_num}) and (\ref{eq:data_wgt_ratio_den}) in the manuscript}\label{apx:thm:exp_w_data}
Taking a conditional expectation on $Y$, we have
\begin{align}
E_{\mathcal{D},\bm{x}}[W(Y,\bm{x})|Y>q_{\alpha}]
&=E_{\mathcal{D},\bm{x}}\left[E_Y\left[W(Y,\bm{x})|Y>q_{\alpha},\mathcal{R},\mathcal{T},\bm{x}\right]\right]\nonumber\\
&=E_{\mathcal{D},\bm{x}}\left[\int_{\mathcal{Q}_{\alpha}}W(y,\bm{x})\frac{f_{\mathcal{T}}(y;\bm{x},\bm{\Psi})}{F_{\mathcal{T}}(\mathcal{Q}_{\alpha};\bm{x},\bm{\Psi})}dy\right] \nonumber\\
&=E_{\mathcal{D},\bm{x}}\left[\int_{\mathcal{Q}_{\alpha}}W(y,\bm{x})\frac{f(y;\bm{x},\bm{\Psi})1\{y\in\mathcal{T}\}}{F(\mathcal{Q}_{\alpha}\cap\mathcal{T};\bm{x},\bm{\Psi})}dy\right]\nonumber\\
&=E_{\mathcal{D},\bm{x}}\left[\frac{\lambda^{*}(\bm{\Psi};\bm{x})}{F(\mathcal{Q}_{\alpha}\cap\mathcal{T};\bm{x},\bm{\Psi})}\int_{\mathcal{Q}_{\alpha}\cap\mathcal{T}}f^{*}(y;\bm{x},\bm{\Psi})dy\right]\nonumber\\
&=E_{\mathcal{D},\bm{x}}\left[\lambda^{*}(\bm{\Psi};\bm{x})\frac{F^*(\mathcal{Q}_{\alpha}\cap\mathcal{T};\bm{x},\bm{\Psi})}{F(\mathcal{Q}_{\alpha}\cap\mathcal{T};\bm{x},\bm{\Psi})}\right],
\end{align}

Equation (\ref{eq:data_wgt_ratio_den}) in the main paper follows immediately the above equation by setting $q_{\alpha}=-\infty$ and hence $\mathcal{Q}_{\alpha}=\mathcal{Y}$.

\section{Additional plots and tables for the real data analysis}
\begin{table}[!h]
\centering
\begin{tabular}{llll}
\hline
Variable & Description & Type & Notes \\ \hline
$x_{i2}$ & Policyholder age & Discrete &  \\
$x_{i3}$ & Car age & Discrete &  \\
$x_{i4}$ & Car fuel & Categorical & Diesel: $x_{i4}=1$ \\
 &  &  & Gasoline: $x_{i4}=0$ \\
$x_{i5}$--$x_{i8}$ & Geographical location & Categorical & Region I: $x_{i5}=1$ \\
 &  &  & Region II: $x_{i6}=1$ \\
 &  &  & Region III: $x_{i7}=1$ \\
 &  &  & Region IV: $x_{i8}=1$ \\
 &  &  & Capital: $x_{i5}=x_{i6}=x_{i7}=x_{i8}=0$ \\
$x_{i9}$--$x_{i10}$ & Car brand class & Categorical & Class A: $x_{i9}=1$ \\
 &  &  & Class B: $x_{i10}=1$ \\
 &  &  & Class C: $x_{i9}=x_{i10}=0$ \\
$x_{i11}$ & Contract type & Categorical & Renewal contract: $x_{i11}=1$ \\
 &  &  & New contract: $x_{i11}=0$ \\ 
\hline
\end{tabular}
\caption{[European automobile claims] Summary of the covariates.}
\label{table:real2_cov}
\end{table}

\begin{figure}[!h]
\begin{center}
\includegraphics[width=\linewidth]{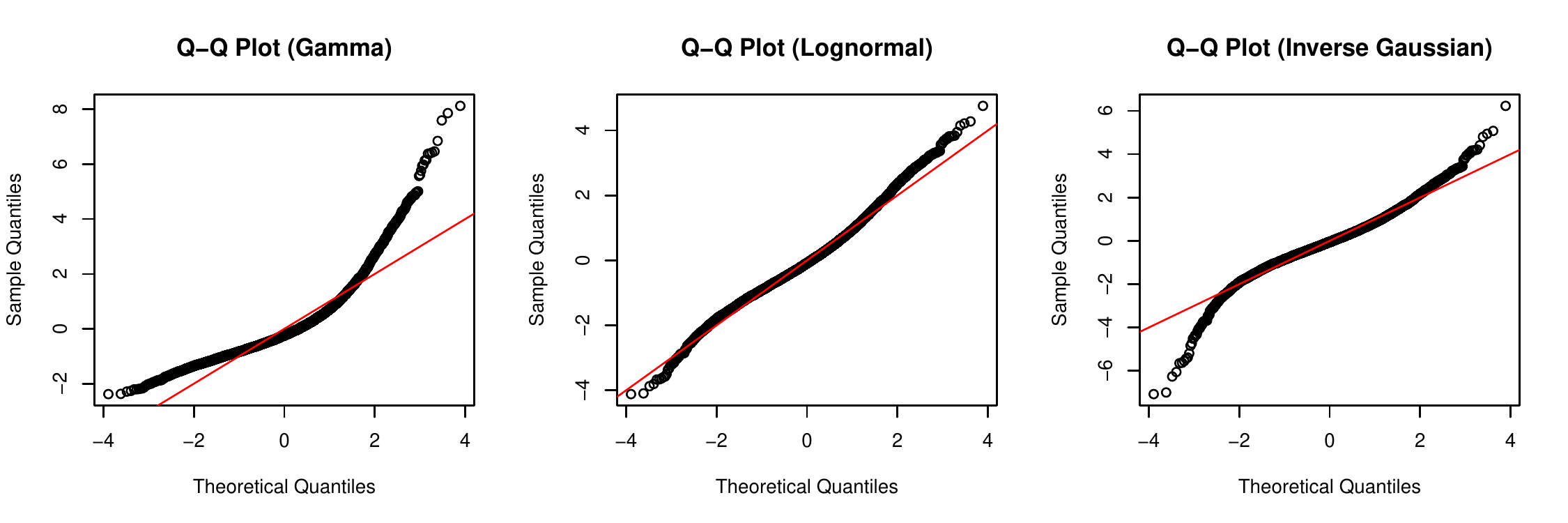}
\end{center}
\vspace{-0.5cm}
\caption{[European automobile claims] The Q-Q plots of the normalized residuals based on Gamma (left panel), log-normal (middle panel) and inverse Gaussian (right panel) distributions.}
\label{fig:real2_qq}
\end{figure}
\end{appendices}

\end{document}